\documentclass[11pt,preprint,3p,authoryear,a4paper]{elsarticle}
\usepackage{graphicx}
\usepackage{xcolor}
\usepackage{dsfont}
\usepackage{amsmath}
\usepackage{amssymb}
\usepackage{amsfonts}
\usepackage{amsbsy} 
\usepackage{amsthm}
\usepackage{array}
\usepackage{booktabs} 
\usepackage{verbatim} 
\usepackage[english]{babel}
\usepackage{float}
\usepackage{epsfig}
\usepackage{pdfsync}
\usepackage{verbatim}
\usepackage{mathrsfs} 
\usepackage{subcaption}
\usepackage{ulem}
\usepackage{pdfpages}
\usepackage{varioref}
\usepackage{setspace}
\usepackage{makecell}

\newcolumntype{C}[1]{>{\centering\let\newline\\\arraybackslash\hspace{0pt}}m{#1}}
\renewcommand\appendix{\par
\setcounter{section}{0}%
\setcounter{subsection}{0}%
\setcounter{table}{0}
\setcounter{table}{0}
\setcounter{figure}{0}
\gdef\thetable{\Alph{table}}
\gdef\thefigure{\Alph{figure}}
\gdef\thesection{\Alph{section}}
\setcounter{section}{0}}

\usepackage{titlesec}

\titleformat{\paragraph}
{\normalfont\normalsize\itshape}{\theparagraph}{1em}{}
\titlespacing*{\paragraph}
{0pt}{3.25ex plus 1ex minus .2ex}{1.5ex plus .2ex}

\newcommand{\E}{\ensuremath{\mathbb{E}}}

\newtheorem{theorem}{Theorem}[section]
\newtheorem{lemma}[theorem]{Lemma}
\newtheorem{proposition}[theorem]{Proposition}
\newtheorem{corollary}[theorem]{Corollary}
\newtheorem{definition}[theorem]{Definition}
\newtheorem{assumption}[theorem]{Assumption}

\newtheorem{remark}{Remark}[section]

\newcommand{\T}{\mathbb{T}}
\newcommand{\Tc}{\mathcal{T}}
\newcommand{\Ex}{\mathbb{E}_x}

\newcommand{\parD}[1]{\frac{\partial}{\partial #1}}

\newcommand{\WqrB}{\overline{W}_{\gamma+\delta}}
\newcommand{\WqrBB}{\overline{\overline{W}}_{\gamma+\delta}}

\newcommand{\Wq}{W_{\delta}}
\newcommand{\Wqr}{W_{\gamma+\delta}}
\newcommand{\Wb}{\overline{W}_{\gamma+\delta}}
\newcommand{\Wbb}{\overline{\overline{W}}_{\gamma+\delta}}
\newcommand{\Wa}{W_{\gamma,\delta,a_p}}

\newcommand{\corr}[1]{\textcolor{black}{#1}}

\begin{document}


\normalem 



\begin{frontmatter}

\title{
On the  optimality of joint periodic and extraordinary dividend strategies}

\author[UMelb]{Benjamin Avanzi}
\ead{b.avanzi@unimelb.edu.au}

\author[UNSW]{Hayden Lau\corref{cor}}
\ead{kawai.lau@unsw.edu.au}

\author[UNSW]{Bernard Wong}
\ead{bernard.wong@unsw.edu.au}

\cortext[cor]{Corresponding author.}

\address[UMelb]{Centre for Actuarial Studies, Department of Economics, University of Melbourne VIC 3010, Australia}
\address[UNSW]{School of Risk and Actuarial Studies, UNSW Australia Business School, UNSW Sydney NSW 2052, Australia}

\begin{abstract}
In this paper, we model the cash surplus (or equity) of a risky business with a Brownian motion (with a drift). Owners can take cash out of the surplus in the form of ``dividends'', subject to transaction costs. However, if the surplus hits 0 then ruin occurs and the business cannot operate any more.

We consider two types of dividend distributions: (i) periodic, regular ones (that is, dividends can be paid only at countably many points in time, according to a specific arrival process); and (ii) extraordinary dividend payments that can be made immediately at any time (that is, the dividend decision time space is continuous and matches that of the surplus process). Both types of dividends attract proportional transaction costs, and extraordinary distributions also attracts fixed transaction costs, a realistic feature. A dividend strategy that involves both types of distributions (periodic and extraordinary) is qualified as ``hybrid''.

We determine which strategies (either periodic, immediate, or hybrid) are optimal, that is, we show which are the strategies that maximise the expected present value of dividends paid until ruin, net of transaction costs. Sometimes, a liquidation strategy (which pays out all monies and stops the process) is optimal. Which strategy is optimal depends on the profitability of the business, and the level of (proportional and fixed) transaction costs. Results are illustrated.

\end{abstract}


\end{frontmatter}

\numberwithin{equation}{section}

\section{Introduction}\label{S_intro}

\subsection{\corr{Background}}
 The literature on risk processes and their optimal control is rich \citep[see, e.g.][for reviews]{AlTh09,OkSu10}. Such processes consider the surplus (or equity) of a risky business. A risky but profitable business will see cash accumulate (on average). They typically would not let their surplus grow to infinity, but to guard against the downside risks, they would retain some cash earnings in order to prevent bankruptcy or financial distress. 
 In this paper, we model such surplus of cash with a stochastic process, and money distributed to shareholders will be interpreted as `dividends'; see also \citet{AvTuWo16c} for a discussion of such surplus models from a corporate finance perspective. The question then is to determine what the optimal way of distributing surplus is, that is, what the optimal (so-called) `dividend' strategy is. Note that this problem is equivalent to that of determining what the optimal level of retained cash earnings is (noting that inflows come from the business dynamics, and all outflows are labelled as `dividends'), but is formulated in function of what owners can \emph{control} (the `dividends').

 
 The natural and usual objective of this optimisation problem is to maximise the expected present value of dividends paid until ruin (which occurs as soon as the surplus becomes negative). Additional historical notes and discussion of dividends in that context can be found in \citet{Ava09}. This objective is also a good criterion of ``stability'' for the company, as it balances profitability (more dividends but earlier ruin) with safety (less dividends but delayed or even absence of ruin); see, e.g. \citet{Buh70}. Nevertheless, quantities such as the finite time ruin probability have also been considered \citep[e.g.][]{DiRo10,DiKaZh14}, and the maximisation can be achieved on more sophisticated objectives, such as involving utility functions \citep[e.g.][]{BaEg10,BaJa15}.


The recent decade has focused a lot on more \emph{realistic} formulations for the dividends \citep[see][for a detailed discussion of what we mean by `realistic']{AvTuWo16c}. One of the axes of development recognises that whilst surplus models are continuous, in real life often delays occur \citep[see, e.g.,][who consider dividend payments with implementation delays]{ChWo17}, and also dividend decisions are usually made at periodic intervals \citep[see, e.g.,][]{AlGeSh11}. 

\corr{Literature on ``periodic'' dividends is relatively new, but attracted a lot of attention. \citet{AlChTh11a} first proposed to use an erlangisation technique \citep{AsAvUs02} to approximate the time between dividend decision times. The idea of the erlangisation technique is to set parameters such that the time between decisions is Erlang$(n/\gamma,n)$  distributed (hereafter denoted ``Erlang$(n)$'') such that the time between decisions becomes deterministic with mean $1/\gamma$ as $n$ goes to infinity. This convergence was illustrated in the dual model  setting (with surplus as a spectrally positive compound Poisson process) in \citet{AvChWoWo13}. \citet{AvTuWo14} confirmed that a periodic barrier strategy is optimal in dual model when the inter-dividend decision time is a simpler Erlang$(1)$ variable. \citet{PeYa16} 
extended those results by considering general spectrally positive L\'evy processes as the underlying surplus model. \cite*{AvTuWo18} studied the optimal problem when the inter-dividend time is a Erlang($n$) random variable. The authors provided a verifying method for a Brownian setting and demonstrated the optimality of a periodic barrier strategy when $n=2$. In all those cases, the type of the optimal periodic dividend strategy is that of a barrier strategy, mirroring the analogous result for dividend decisions that can occur at any time \citep[see][]{BaKyYa13}. Optimal strategies with spectrally negative L\'evy processes were considered in \citet{NoPeYaYa17}. Of closest relevance to this paper is consideration of optimal periodic (only) dividend strategies with fixed transaction costs, developed in \citet{AvLaWo20d,AvLaWo20c} for spectrally positive and negative L\'evy processes, respectively.}

\subsection{\corr{Types of dividends and fixed transaction costs}}

\corr{We} consider two types of dividend distributions: (i) periodic, regular ones (that is, dividends can be paid only at countably many points in time, according to a specific arrival process); and (ii) immediate dividend payments that can be made at any time (that is, the dividend decision time space is continuous and matches that of the surplus process). This matches the behaviour of companies in real life, as most established firms would pay dividends regularly. If they feel the need to distribute more, then they would clearly label those extra payments as `extraordinary' (and sometimes also do it in a different way, such as with share buy-backs, which is not in contradiction with our framework). One can find real life examples \citep[e.g.,][]{Woodside13,Wesfarmer14}, and was further explained by, for instance, \citet{Morningstar14}: ``From time to time, companies pay out special dividends when they have had an extraordinarily good period of profitability. These dividends fall outside the scope of the ``normal'' half-year or full-year result.'' This possibly is to avoid signalling the fact that those extra payments should be expected to continue in the future. Furthermore, it does make sense that those extra distributions carry heavier costs than the regular ones (actual costs, but also undesirable signalling costs such as we just explained). We will hence penalise them with heavier fixed transaction costs. 

\citet{AvTuWo16}, in a dual model framework and with both types of dividends being admissible, showed that when transaction costs are moderately cheaper for periodic dividends, then both types of dividends can be optimal, leading to an optimal \emph{hybrid} dividend strategy. These results were extended to spectrally positive L\'evy processes by \citet{PeYa18}. However, those papers consider \corr{\emph{proportional transaction costs only}}, and in reality fixed costs are likely to be the ones that truly differentiate the cost of ``periodic'' versus ``immediate'' dividends. \corr{It is hence a non trivial extension, which introduces a number of technical difficulties as explained later.}

\subsection{\corr{Statement of contributions and structure of the paper}}

\corr{In this paper we extend results on ``hybrid'' dividend strategies by introducing fixed transaction costs,} 
which results in a comprehensive, more realistic treatment of optimal hybrid strategies. 
\corr{While fixed costs has been used extensively in other fields such as asset allocation problems \citep[see, e.g.,][]{OkSu02, BaMa14, BaMa16}, frameworks and objectives are different, and fixed costs are rarely studied in optimal dividend problems due to the additional complexities in the proofs. Nonetheless, with fixed transaction costs,} results are materially different, richer and more realistic as explained below. Furthermore, the cash flow of the company is modeled by a diffusion process, which leads to transparent and many explicit results, and is sufficient to get insights about the optimal strategies. 

When the company is profitable, an optimal strategy is a hybrid $(a_p,a_c,b)$ strategy which (1) pays non-regular dividends only when the surplus is too high and (2) pays regular (periodic) dividends when the surplus is moderate. This strategy has some desirable properties. Namely, regular dividends are either zero or bounded. When a regular dividend is zero, either the company is at risk of bankruptcy or a recent special dividend has been paid. In either case, such behaviour is reasonable. When the company is non-profitable, the model has a different (and no less interesting) interpretation. The main results of the paper are summarised in Section \ref{S_map}, after our notation is introduced. \corr{A major contribution of this paper is the proof of the existence of a hybrid $(a_p,a_c,b)$ strategy when such strategy is optimal. The main difficulty lies in that the existence problem is equivalent to the existence of a solution to three seemingly unrelated non-linear equations in three parameters, where in general there is no reason to have a solution. The well-known problem with two parameters (barriers) is already difficult in general, and to the best of our knowledge, the problem with three parameters (barriers) as in this paper has not been  solved thus far.}

This paper is organised as follows. Section \ref{S.the.model} introduces our mathematical framework. Section \ref{S.Verification} proposes a set of sufficient conditions for a strategy to be optimal, regardless of whether the business is profitable. From there until Section \ref{S.Optimality}, it is assumed \corr{that the business is profitable}. As an application of the results developed in Section \ref{S.Verification}, \corr{in Section \ref{S_betasucks}, }an optimal strategy is formulated when the proportional cost is higher than a certain threshold. Section \ref{S.hybridG} introduces the class of hybrid $(a_p,a_c,b)$ strategies and calculates the value function of a general hybrid $(a_p,a_c,b)$ strategy \corr{and} shows constructively that our candidate strategy exists among the class of hybrid $(a_p,a_c,b)$ strategy, when the proportional cost is low (lower than a certain threshold). \corr{Following that, Section \ref{S.VD} derives some auxiliary results regarding the properties of the value function of the candidate strategy.}
Section \ref{S.Optimality} proves that our candidate strategy is optimal, when the proportional cost is low. Section \ref{S.mu.neg} studies the remaining case  when the business is strictly non-profitable. Section \ref{S.Convergence} discusses how the different optimal strategies are ``connected'' (i.e., across the Table in Section \ref{S_map}). Finally, \corr{Section} \ref{S.Numerical} presents numerical illustrations, and Section \ref{S.conclusion} concludes.

\section{The model }\label{S.the.model}

\subsection{Surplus model before dividends}
We define the surplus process $X=\{X(t);t\geq 0\}$ under the family of laws $(\mathbb{P}_x;x\in\mathbb{R})$ to be a diffusion process that starts at $x\geq0$, i.e.
\begin{equation}\label{Def.Diffusion}
X(t)=x+\mu t+\sigma W(t),
\end{equation}
where $W=\{W(t);t\geq 0\}$ is a standard Brownian motion. This surplus process is to be interpreted as the excess, discretionary equity available to the company to pay dividends. It is assumed that it is sufficiently liquid to pay dividends immediately when it is so decided.

We denote the expected profit per unit of time of the business as $\E[X(t+1)-X(t)]:=\mu$. Unless stated otherwise, we assume that
\begin{equation}\label{Ass.mu.positive}
\mu \geq 0,
\end{equation}
which means that the business is profitable. The opposite case will be studied in Section \ref{S.mu.neg}, and the connection of the optimal strategies between the cases $\mu$ greater than, equal to, and small than $0$ is conducted in Section \ref{S.VaryingMu} (continuity of the barriers).

\subsection{The introduction of dividends}
In this paper, a dividend strategy is comprised of two components. Dividends can be paid at any time, but there are periodic opportunities to pay dividends at lower transaction costs. A \emph{dividend strategy} must hence determine how much periodic dividends to pay and how much ``immediate'' (extraordinary) dividends to pay and when. For a dividend strategy $\pi$, we denote the accumulated periodic ``regular'' dividend process as $D^\pi_p=\{D^\pi_p(t);t\geq 0\}$ and the accumulated non-periodic ``immediate'' dividend process as $D^\pi_c=\{D^\pi_c(t);t\geq 0\}$. The strategy $\pi$ is then specified through $(D^\pi_p,D^\pi_c)$, and the accumulated total dividend process under strategy $\pi$ is denoted as $D^\pi=\{D^\pi(t);t\geq 0\}$. This means
\begin{equation}
D^\pi(t)=D^\pi_p(t)+D^\pi_c(t),~t\geq 0.
\end{equation}
Note that the subscripts $p$ and $c$ refer to the timing of the dividend decision process, be it `periodic' or `continuous', in line with previous literature. 

We need to clarify mathematically how the ``regular'', or periodic payment times are defined. Define $N_\gamma=\{N_\gamma(t);t\geq 0\}$ as a Poisson process (independent of $W$) with rate $E[N_\gamma(1)]=\gamma>0$, which serves as our periodic \emph{dividend decision times}. In other words, periodic dividends can only be paid when $N_\gamma$ has increments. Such times are denoted as $\T=\{T_i;i\in\mathbb{N}\}$ with
\begin{equation}\label{E_Ti}
T_i=\inf\{t\geq 0: N_\gamma(t)=i\}.
\end{equation}
This implies that {$T_1$} and $T_{i+1}-T_i$, are i.i.d. exponential random variables with mean $1/\gamma$, for all $i\in\mathbb{N}$.

The surplus process after the dividend payments is therefore $X^\pi=\{X^\pi(t);t\geq 0\}$ with
\begin{equation}
X^\pi(t)=X(t)-D^\pi(t).
\end{equation}
We define $\tau^\pi$ to be the ruin time of the process $X^\pi$, i.e.
\begin{equation}
\tau^\pi=\corr{\tau^\pi_x:=\inf\{t\geq 0:X^\pi(t)<0,X(0)=x\}},
\end{equation}
that is, the company must stop its operations as soon as its surplus hits zero, and no further dividends will be paid.

A Markovian \emph{stationary} strategy is a strategy where the control at time $t$ is a deterministic function of $X^\pi(t-)$ known at time $0-$ which maps the surplus and its characteristics into a dividend payment, i.e. $(\Delta D_p^\pi(t),\Delta D_c^\pi(t))=(f_p(X^\pi(t-))1_{\{t\in \mathbb{T}\}},f_c(X^\pi(t-))1_{\{t\notin \mathbb{T}\}})$ for a given function $f=(f_p,f_c)$, \corr{ where for a c\`adl\`ag function $f$, the $\Delta$ operator maps $f(x)$ to $\Delta f(x):=f(x)-\lim_{y\uparrow x}f(y)$}. For such a strategy $\pi$, if $D^\pi_c(t)\equiv 0$, we call it a (pure) periodic strategy (with regular payments only). If $D^\pi_p(t)\equiv 0$, we call it a (pure) continuous strategy (with immediate payments only). Otherwise, we refer it as a hybrid strategy, as there is a non-zero probability that both components are present. 

\begin{remark}
	Note that if dividends can only be paid after every $n$-increment then the time between dividend decision times is Erlang distributed with shape parameter $n$ and rate parameter {$n\gamma$} (the sum of $n$ independent exponential{($n\gamma$)} random variables). This random variable can have arbitrarily small variance for appropriate choices of parameters. This is what led to the so-called ``Erlangisation'' technique as discussed in \citet{AsAvUs02,AlChTh11a}. Indeed, letting the parameter $n$ increase to infinity causes the variance of the Erlang($n$) random variable to vanish, which means deterministic numbers can be approximated sufficiently well by choosing a large enough $n$.
	
	This motivates model setups with `simple' Poissonian distribution strategies (whereby inter-dividend decision times are exponentially distributed), which is adaopted in this paper. This is an important first step to solving the more general Erlang with $n\ge2$ case. Showing optimality for $n\ge 2$ is surprisingly difficult, but not impossible; see \citet{AvTuWo18}.
\end{remark}

By defining the filtration generated by the process $(X,N_\gamma)$ by $\mathbb{F}=\{\mathscr{F}_t:t\geq 0\}$, we say a (hybrid) dividend strategy $\pi:=\{(D^\pi_p(t),D^\pi_c(t));t\geq 0\}$ is admissible if both $D^\pi_p$ and $D^\pi_c$ are non-decreasing, right continuous and $\mathbb{F}$-adapted process where any sample path of the process $D^\pi_c$ is an increasing step function in time (as a fixed cost will be incurred at each payment), and where the cumulative amount of periodic dividends $D^\pi_p$ admits the form
\begin{equation}
D^\pi_p(t)=\int_{[0,t]}\nu^\pi(s)dN_\gamma(s),~t\geq 0
\end{equation} 
for some non-negative adapted process $\nu^\pi:=\{\nu^\pi(t),t\geq 0\}$.
Furthermore, note by definition the sample paths of $X$ are continuous ($X(t)=X(t-)$) and hence we require
\begin{equation}
\Delta D^\pi(t)\leq X^\pi(t-),\quad t\leq \tau^\pi
\end{equation}
that is, the dividend paid at time $t$ cannot exceed the current value of the surplus. Denote this set of admissible strategies $\Pi$.

\subsection{The expected present value of dividends until ruin}
To measure the performance of the strategies, we will focus on the expected present value of dividends until ruin
\begin{equation}\label{E_EPVD}
V_{1-\beta,\chi}(x;\pi)=V(x;\pi):=\Ex\int_{0}^{\tau^\pi}e^{-\delta t} \Big(dD^\pi_p(t)+(\beta dD^\pi_c(t)-\chi)1_{\{\Delta D^\pi_c(t)>0\}}\Big),
\end{equation}
where $\Ex[\cdot]:=\mathbb{E}[\cdot|X(0)=x]$ is the mathematical expectation under the law $\mathbb{P}_x$  (for each $x\in\mathbb{R}$), and where $\delta>0$ is a time-preference parameter (or discount factor). Furthermore, non-periodic ``immediate'' dividend payments of amount $\xi>0$ incur a transaction costs $(1-\beta)\xi+\chi$. In other words, there is a proportional transaction rate of $1-\beta$, and fixed transaction costs of $\chi$. 

We seek to maximise the expected present value of dividends, which means that we will look for an optimal strategy $\pi^*\in\Pi$ such that 
\begin{equation}\label{Def.Optimal}
V(x;\pi^*)=\sup_{\pi\in\Pi}V(x;\pi):=v(x)=v_{1-\beta,\chi}(x),\quad x\geq 0.
\end{equation}
Because the process is ruined immediately when it reaches $0$,  we have 
\begin{equation}\label{E_V0}
V(0;\pi)=0 \quad \text{for}\quad \pi\in\Pi.
\end{equation}
Note that we will also write $\mathbb{P}$ and $\mathbb{E}$ for $\mathbb{P}_0$ and $\mathbb{E}_0$ respectively.

\begin{remark}\label{Remark.Rational}
	An optimal strategy should demonstrate the following 2 rational behaviours:
	\begin{enumerate}
		\item The non-periodic dividend payment at time $t$, $\Delta D^\pi(t)$, is either $0$ or strictly greater than ${\chi}/{\beta}$. This is because any strategy that pays a non-periodic dividend less that ${\chi}/{\beta}$ does not contribute positively to the value function and therefore has at most the same value function as the same strategy without negative contributions.
		\item At periodic dividend time $t=T_i$ for some $i\in\mathbb{N}$, we do not pay non-periodic dividends. Otherwise, a higher transaction cost is paid, yielding at most the same value function.
	\end{enumerate}
\end{remark}

\begin{remark} \label{R_betacp}
	Note that in \eqref{E_EPVD} periodic dividends do not attract any transaction costs. 
	With respect to \emph{proportional} transaction costs this is without of loss of generality, as long as proportional transaction costs on periodic dividends (say, $1-\beta_p$) are smaller than that on immediate dividends (say, $1-\beta_c$), which is what you would expect in practice as discussed earlier. In this case, \eqref{E_EPVD} would become
	\begin{align}
	V(x;\pi)~=~&\Ex\int_{0}^{\tau^\pi}e^{-\delta t} \Big(\beta_p dD^\pi_p(t)+(\beta_cdD^\pi_c(t)-\chi)1_{\{\Delta D^\pi_c(t)>0\}}\Big)\label{prob.1} \\
	~\equiv~& \beta_p\Ex\int_{0}^{\tau^\pi}e^{-\delta t} \Big(dD^\pi_p(t)+(\beta dD^\pi_c(t)-\chi/\beta_p)1_{\{\Delta D^\pi_c(t)>0\}}\Big) \quad \text{with }\beta=\frac{\beta_c}{\beta_p}\le1.\label{prob.2}
	\end{align}
	That is, the objective is simply scaled by a constant ($\beta_p$), which will not affect the generality of our set-up. In other words, an optimal strategy in problem \eqref{prob.2} is optimal in problem \eqref{prob.1}. However, note that the fixed transaction cost amount $\chi$ needs to be appropriately scaled if one wants to obtain accurate numerical valued for one problem from the other.
	
	On the other hand, introduction of fixed transaction costs $\chi_p$ on periodic dividends would likely alter the form of the optimal dividend strategy fundamentally. We expect that the optimal periodic barrier would be split into a higher trigger barrier, and lower dividend payment barrier, as is often the case in classical impulse cases (because of the reason explained under item 1 in Remark \ref{Remark.Rational}). Furthermore, we postulate that ascertaining which type of dividends attracts higher transaction costs \emph{on average} or \emph{in an expected sense} would be critical in determining the optimal dividend strategy. We believe the optimal dividend strategy would depend on some sort of `expected' overall transaction costs for each type, which is not trivial to determine as the number and timing of dividends are random in both cases, and do not match. That being said, if one assume that both proportional \emph{and} fixed transaction costs are lower on regular dividends (as opposed to immediate dividends), then extension of the current paper should be relatively straightforward.
\end{remark}

\subsection{Definition of relevant dividend strategies}
In this section, we define all dividend strategies that we will refer to in this paper. Note that they are all Markovian \emph{stationary} strategies as defined above just after \eqref{E_Ti}.

\begin{definition}\label{Def.periodic.b}
	A periodic $b$ strategy, denoted ${\pi_b}$, is a periodic dividend strategy which pays a dividend $$\Delta D^{\pi_b}_p(t)=(X^{\pi_b}(T_i-)-b)1_{\{X^\pi(T_i-)\geq b\}},\quad \Delta D^{\pi_b}_c(t)\equiv 0.$$ at time $T_i$, as long as ruin has not occurred yet, that is, for all $T_i\le\tau^{\pi_b}$, $i\in\mathbb{N}$.
\end{definition}

We now define the class of strategies that we prove optimal in some cases later in the paper. 

\begin{definition}\label{Def.hybrid.apacb}
	A hybrid $(a_p,a_c,b)$ strategy with $0\leq a_p\leq a_c\leq b$, denoted as $\pi_{a_p,a_c,b}$, is a strategy which
	\begin{enumerate}
		\item pays (before ruin) periodic dividend that brings the surplus down to $a_p$ whenever the (controlled) surplus $X^{\pi_{a_p,a_c,b}}$ is above or equal to {$a_p$} right before the dividend payment times,
		\item pays (before ruin) an immediate dividend that brings the surplus down to $a_c$ whenever the surplus $X^{\pi_{a_p,a_c,b}}$ is above or equal to $b$ outside the periodic dividend times.
	\end{enumerate}  
	In mathematical notation, it means 
	\begin{equation}
	\begin{cases}
	\Delta D^\pi_p(T_i)=(X^\pi(T_i-)-a_p)1_{\{X^\pi(T_i-)\geq \corr{a_p}\}}1_{\{T_i\leq \tau^\pi\}}\\
	\Delta D^\pi_c(t)=(X^\pi(t-)-a_c)1_{\{X^\pi(t-)\geq b\}}1_{\{t\neq T_i\}}1_{\{t\leq \tau^\pi\}}
	\end{cases},
	\end{equation}
	with $\pi=\pi_{a_p,a_c,b}$.
\end{definition}

Figure \ref{fig.apacb} illustrates the strategy described in Definition \ref{Def.hybrid.apacb}. It charts a typical sample path when a hybrid $(a_p,a_c,b)$ is applied, where $a_p=1$, $a_c=2$, and $b=4$. The dotted vertical lines indicate periodic dividend decision times $T_i$'s. From the graph, we see that there is an immediate payment just before $T_1$, of amount $b-a_c=2$ (before transaction costs). On the other hand, all $T_i$'s in the graph trigger periodic payments.
Note that Definition \ref{Def.hybrid.apacb} indicates that $a_p\leq a_c\leq b$, but in fact Remark \ref{Remark.Rational} implies that only the cases $b>a_c+{\chi}/{\beta}$ (and hence $b\geq a_p+{\chi}/{\beta}$ since $a_c>a_p$) make sense in order to avoid negative contribution to the value function. Note also that the solid vertical lines are not part of the processes. They are displayed ``artificially'' to illustrate the ``jump'' of the processes.

\begin{figure}[htb]
	\centering
	\begin{minipage}{0.49\textwidth}
		\centering
		\includegraphics[width=1\textwidth]{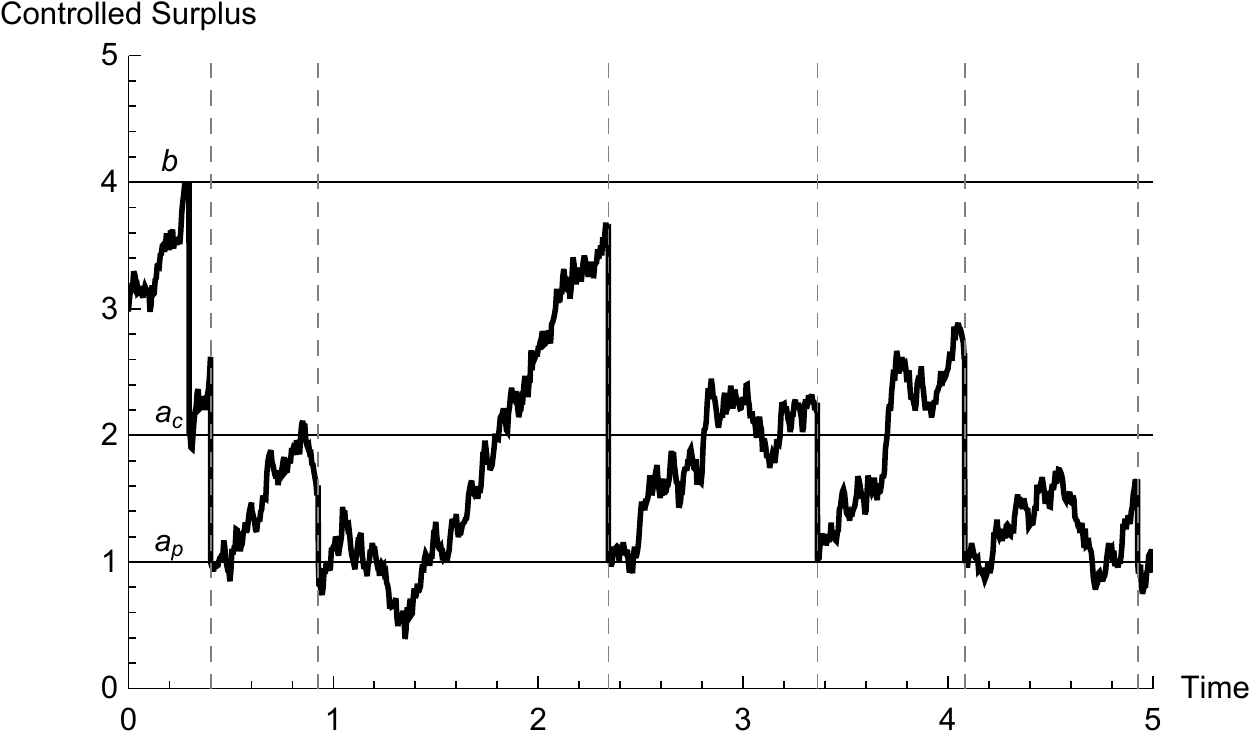}
		\subcaption[first caption.]{An illustration of a hybrid $(a_p,a_c,b)$ strategy}\label{fig.apacb}
	\end{minipage}%
	\begin{minipage}{0.49\textwidth}
		\centering
		\includegraphics[width=1\textwidth]{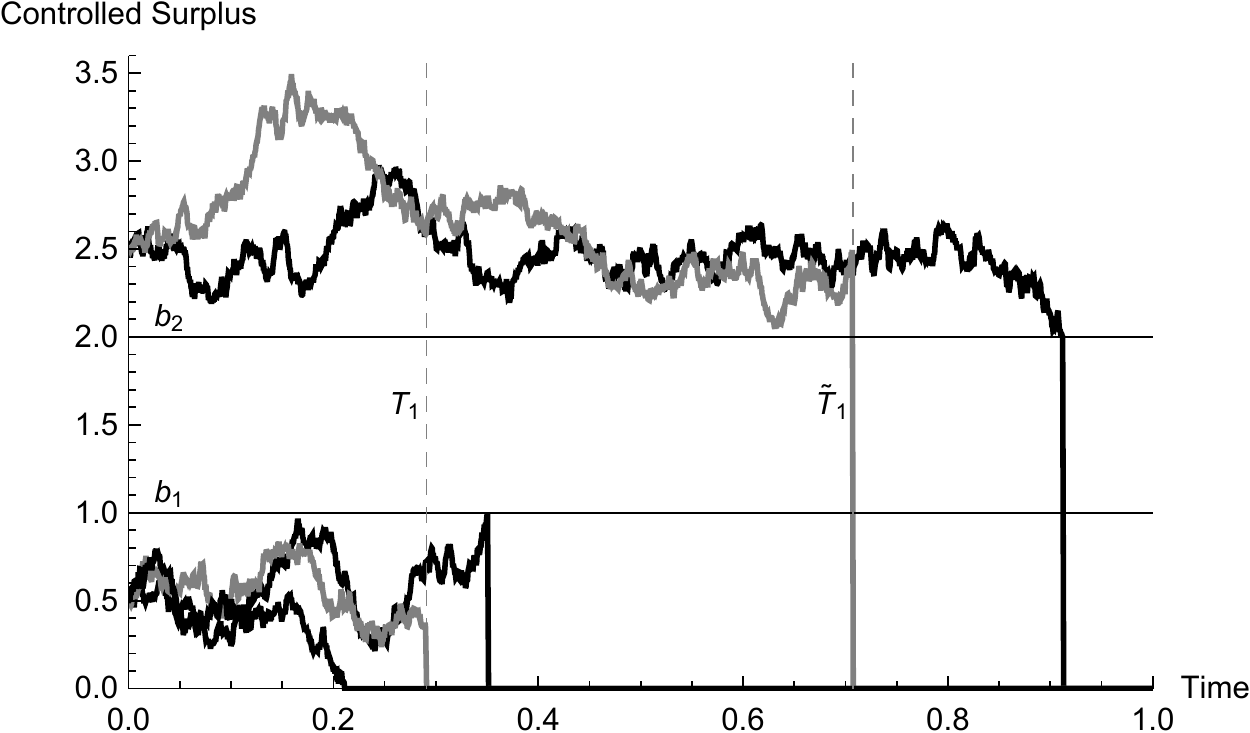}
		\subcaption[second caption.]{An illustration of a liquidation $(b_1,b_2)$ strategy}\label{fig.b1b2}
	\end{minipage}%
	\caption{Illustrations of the main optimal strategies} \label{fig.strategies}
\end{figure}

\begin{definition}\label{def.Liq}
	A liquidation $(b_1,b_2)$ strategy, characterised by 2 parameters $0<b_1<b_2\leq \infty$, denoted as $\pi_{b_1,b_2}$, is the strategy that 
	\begin{enumerate}
		\item pays (before ruin) non-periodic dividend $X(\theta)$ {(surplus just before ruin caused by this final dividend)}, where $\theta=\inf\{t\geq 0: X(t)\in(b_1,b_2)\}$, the first time the surplus is within the open interval $(b_1,b_2)$;
		\item pays (before ruin) periodic dividend of size $X^\pi(T_1-)$ {(surplus just before ruin caused by this dividend)} when $X^\pi(T_1-)\leq b_1$ or $X^\pi(T_1-)\geq b_2$, where $\pi=\pi_{b_1,b_2}$.
	\end{enumerate}
	In mathematical notation, it means 
	\begin{equation}
	\begin{cases}
	\Delta D^\pi_c(t)=X(\vartheta-)1_{\{t=\theta< T_1\}}1_{\{t\leq \tau^\pi\}},\quad \vartheta=\inf\{t\geq 0: X(t)\in(b_1,b_2)\}, \\
	\Delta D^\pi_p(T_1)=X^\pi(T_1-)1_{\{T_1\leq \tau^\pi\}},
	\end{cases}
	\end{equation}
	with $\pi=\pi_{b_1,b_2}$.
\end{definition}

The class of liquidation strategies as defined in Definition \ref{def.Liq} will be sometimes optimal when $\mu<0$, and we distinguish two cases here for notation purposes:
\begin{enumerate}
	\item Liquidation $(b_1,b_2)$ strategy, with $b_2<\infty$.
	\item Liquidation $(b,\infty)$ strategy, denoted as $\pi_{b,\infty}$, which is the liquidation $(b,b_2)$ strategy with $b_2=\infty$. 
\end{enumerate}

{Figure \ref{fig.b1b2} illustrates the strategy of Definition \ref{def.Liq}. It shows different sample paths from when $X(0)=2.5$ (upper part), and $X(0)=0.5$ (lower part), with $b_1=1$ and $b_2=2$. When $X(0)=2.5$, the path in grey represents the case when the first dividend decision time ($\widetilde{T}_1$ in the figure) comes before the surplus reaches $b_2$. As a result, the company is liquidated at $\widetilde{T}_1$ at that surplus level. With the other scenario (represented by the path in black), the first dividend decision time (not shown in the figure) comes after the surplus hits $b_2$ and therefore it is immediately liquidated at that time, and a surplus of $b_2=2$ (before transaction costs) is distributed. Similar behaviour can be seen when $X(0)=0.5$. When to liquidate depends on whether the surplus touches $b_1$ first (black path) or $T_1$ comes first (grey path), conditioning on survival. Otherwise, the company is ruined (the lowest black path).}

\begin{remark}
	The strategies mentioned above are related:
	\begin{enumerate}
		\item The periodic $0$ strategy, denoted as $\pi_0$ (see Definition \ref{Def.periodic.b}) will sometimes also be optimal when $\mu<0$; see Section \ref{S_map}. This strategy pays $X^{\pi_0}(T_1-)$ when $T_1\leq \tau^{\pi_0}$, and can be seen as the limit of a liquidation $(b_1,b_2)$ strategy when $b_2-b_1\downarrow 0$, or simply $b_1\uparrow \infty$. It is also denoted as $\pi_{a,a}$, or $\pi_{\infty,\infty}$.
		\item The liquidation $(b,\infty)$ strategy can be seen as a hybrid $(0,0,b)$ strategy, i.e. $\pi_{b,\infty}=\pi_{0,0,b}$.
	\end{enumerate}
Further convergence results are developed and illustrated in Section \ref{S.Convergence}.
\end{remark}

\subsection{Main results}\label{S_map}

The nature of the optimal strategy will depend on the value of some key parameters, as is shown in this paper. Our main results are summarised in Table \ref{T_roadmap}, which can also be used as a road map for reading the paper. In addition, the transition between cells is ``continuous'', except for the cells in the second row for $\mu<0$, as they are disjoint unless $\beta=\gamma/(\gamma+\delta)$
; see Section \ref{S.Convergence} for details and proofs of that statement. \corr{Note that the case $\chi = 0$ in the first column corresponds to what was considered in \citet[noting that we have set $\eta_p=1$ here because proportional transaction costs are identical for both types of dividends in this paper; see also Remark \ref{R_betacp}]{AvTuWo16}.}

\begin{table}[H]
	\centering
	
	\begin{tabular}{l|c|c|c}
		&$\mu \ge 0$ & \multicolumn{2}{c}{$\mu < 0$} \\
		&$\chi \ge 0$ & $\chi/\beta \ge -\frac{\mu}{\gamma+\delta} $ & $\chi/\beta < -\frac{\mu}{\gamma+\delta} $ \\ \hline
		$0\le \beta \le \beta_0$ & Periodic barrier $\pi_b$  (Thm \ref{Thm.Small.Beta}) & \multicolumn{2}{c}{Periodic barrier $\pi_0$ (Thm \ref{Thm.mu.neg})}  \\
		$\beta_0< \beta \le \frac{\gamma}{\gamma+\delta}$ & Periodic barrier $\pi_b$ (Thm \ref{Thm.Small.Beta})& \thead{Periodic barrier $\pi_0$ \\ (Thm \ref{Thm.mu.neg})} & \thead{Liquidation $\pi_{b_1,b_2}$ \\ (Thm \ref{Thm.mu.neg})}\\
		$\frac{\gamma}{\gamma+\delta} < \beta \le 1 $ & Hybrid  $\pi_{a_p,a_c,b}$ (Thm \ref{Thm}) & \multicolumn{2}{c}{Liquidation $\pi_{b,\infty}$ (Thm \ref{Thm.mu.neg})}\\
	\end{tabular}
	\caption{Map of the dividend strategies proven as optimal in the different cases considered in the paper}
	\label{T_roadmap}
\end{table}

The results can be interpreted as follows.

Recall that $\gamma/(\gamma+\delta)$ is the expected present value at time 0 of a payment 1 paid after an $\gamma$-exponentially distributed random amount of time, discounted with a continuous force of interest $\delta$. At a very high level, this explains why this ratio is involved in most thresholds in the table: at any point in time, the model balances the choice between (i) a dollar of dividend paid immediately, with net value involving $\beta$ and $\chi$, and (ii) a dollar paid at the next periodic time (without transaction costs), with expected present value $\gamma/(\gamma+\delta)$.

Let us first focus first on the threshold $\gamma/(\gamma+\delta)$ for $\beta$. Ruin is unlikely in the next instant, and if $\mu\geq 0$ then we do not want to liquidate at first opportunity, so we ignore ruin for now. If a dividend of size of $\xi$ is to be paid, the decision between paying now (as immediate dividend) or paying later (as a periodic dividend) should depend on whether 
\begin{equation}
\beta\xi-\chi > \frac{\gamma}{\gamma+\delta}\xi,
\end{equation}
in an expected sense. This would be the case if and only if 
\begin{equation}
\xi>\frac{\chi}{\beta-\frac{\gamma}{\gamma+\delta}}\quad\text{and}\quad  \beta>\frac{\gamma}{\gamma+\delta}.
\end{equation}
This condition will indeed re-appear later when we construct our candidate strategy (which will be proved optimal); for instance in Proposition \ref{prop.2} where we use $\alpha$ to denote $\beta-\gamma/(\gamma+\delta)$. It plays a significant role in determining the minimum distance between barriers $a_c$ and $b$. 

Now, when $\mu<0$, we must liquidate as soon as possible so only $\xi>\chi/\beta$ is required for immediate dividend, {because this is the amount of fixed transaction costs that need to be paid, and the optimisation won't require extra to compensate for future possible gains (since the business is not profitable)---we only need the dividend to be admissible}. This explains the distinction between the two right columns. This can also be interpreted as follows. The quantity
\begin{equation}
-\frac{\mu}{\gamma+\delta} = -\frac{\mu}{\gamma}\frac{\gamma}{\gamma+\delta}
\end{equation}
is in fact the expected present value of the expected loss $-\mu/\gamma$ (in absolute terms) that will be accumulated until the next periodic payment. Whether $\beta$ times this quantity is more or less than the fixed transaction cost $\chi$ impacts the optimal strategy, which makes intuitive sense. This is especially the case when $\chi$ is large, that is, when immediate dividends are very expensive. In this case, even for sufficiently high $\beta_0< \beta \le \gamma/(\gamma+\delta)$, it will be optimal to liquidate with a periodic payment at first opportunity ($\pi_0$), but not immediately. The threshold $\beta_0$ will be defined in Section \ref{S.mu.neg}.

The expected present value of a Periodic barrier $\pi_b$ can be found in \citet{AvTuWo16} or \citet{PeYa16b}, that of a hybrid strategy $(a_p,a_c,b)$ in Section \ref{S.hybridG} (with optimal parameters in Section \ref{S.hybridG}), an that of a Liquidation $(b_1,\infty)$ strategy in Section \ref{S.mu.neg}. Optimality of those strategies is established thanks to the Verification lemma in Section \ref{S.Verification} through the referenced Theorems in Sections \ref{S_betasucks}, \ref{S.Optimality}, and \ref{S.mu.neg}, respectively. Further illustrations are provided in Section \ref{S.Numerical}.

\section{A verification lemma}\label{S.Verification}

In this section, we provide a set of sufficient conditions for a strategy $\pi\in\Pi$ to be optimal, in the sense of (\ref{Def.Optimal}). Recall that for a real-valued function $F$, the extended generator for the stochastic process $X$ on a real-valued function $F$ is defined to be
$$\mathscr{A}F(x):=\frac{\sigma^2}{2}F''(x)+\mu F'(x)$$
for $x\in\mathbb{R}$ such that the above makes sense. Throughout this paper, we will repeatedly use the following lemma to prove the optimality of different dividend strategies in different cases.

\begin{lemma}\label{Verification.lemma}
	For a strategy $\pi^*\in\Pi$, denote its value function $H(x):=V(x;\pi^*)$. Suppose there is a finite set $E\subseteq \mathbb{R}_+$ such that $H$ satisfies
	\begin{enumerate}
		\item $H\geq 0$,
		\item $H\in \mathscr{C}^1(\mathbb{R}_+)\cap\mathscr{C}^2(\mathbb{R}_+\backslash E)$,
		\item $H'$ is bounded on sets $[1/n,n]$ for all $n\in\mathbb{N}$,
		\item On $\mathbb{R}_+\backslash E$, $H$ satisfies 
		\begin{equation}\label{eqt.HJB1}
		(\mathscr{A}-\delta)H(x)+\gamma \sup_{\xi\in[0,x]}\Big(\xi+H(x-\xi)-H(x)\Big)
		\leq 0,
		\end{equation}
		\item On $\mathbb{R}_+$, $H$ satisfies 
		\begin{equation}\label{eqt.HJB2}
		\sup_{\xi\in[0,x]}\Big((\beta\xi-\chi)1_{\{\xi>0\}}+H(x-\xi)-H(x)\Big)
		= 0,
		\end{equation}
	\end{enumerate}
	then $\pi^*$ is optimal, i.e. $V(x;\pi^*)=v(x)$ for all $x\geq 0$.
\end{lemma}

\begin{proof}
The proof which is provided in Appendix \ref{A.ver.lemma} is standard; see for instance \citet{PeYa16}. However, it requires careful treatment of (1) the different types of strategies (2) immediate dividend at time $0$ (3) the approximation for It\^o's lemma at the points when $H$ is not smooth.
\end{proof}

\section{Optimality of a periodic barrier strategy when proportional transaction costs $1-\beta$ are high} \label{S_betasucks}

In this section, we show that a periodic $b$ strategy (see Definition \ref{Def.periodic.b}) is optimal when $\beta\leq {\gamma}/{(\gamma+\delta)}$ and $\mu\geq 0$. This case corresponds to the top left cell of Table \ref{T_roadmap}.
From \citet{PeYa16} it follows that there exists a optimal barrier $b_0^*\geq 0$ such that the periodic $b_0^*$ strategy, $\pi_{b_0^*}$, is optimal when dividends are only allowed to paid at the dividend payment times. Note that this strategy is also admissible in our setting and our definition of value functions for $\pi_{b_0^*}$ agrees with theirs, as there are no dividends to be paid outside the (Poissonian) dividend payment times and periodic dividends do not attract transaction costs. Therefore, we can borrow the results from \citet{PeYa16} regarding the behaviour of the value function $V(\cdot;\pi_{b_0^*})$, which is summarised as follows:
\begin{enumerate}
	\item The first 4 conditions in Lemma \ref{Verification.lemma} hold, with the finite set $E=\{b_0^*\}$.
	\item When $b_0^*>0$, the function is concave. In particular, we have
	$$V'(x;\pi_{b_0^*})\begin{cases}
	>1,\quad&x\in(0,b_0^*)\\
	=1,\quad&x=b_0^*\\
	\in(\frac{\gamma}{\gamma+\delta},1),\quad&x\in(b_0^*,\infty)
	\end{cases}.$$
	\item When $b_0^*=0$, we have
	\begin{enumerate}
		\item If $\mu>0$, then $1\geq V'(0+;\pi_{b_0^*})>V'(x;\pi_{b_0^*})>{\gamma}/{(\gamma+\delta)}\geq \beta>0$, for $x>0$.
		\item If $\mu=0$, then $1>V'(x;\pi_{b_0^*})={\gamma}/{(\gamma+\delta)}\geq \beta>0$, for $x>0$.
	\end{enumerate}
\end{enumerate}

Hence, to show that $\pi_{b_0^*}$ is optimal, it suffices to show (\ref{eqt.HJB2}), i.e.
$$\sup_{\xi\in[0,x]}\Big((\beta\xi-\chi)1_{\{\xi>0\}}+V(x-\xi;\pi_{b_0^*})-V(x;\pi_{b_0^*})\Big)\leq 0,\quad x>0.$$
Denote
\begin{equation*}
H_0(\xi)=\beta\xi-\chi+V(x-\xi;\pi_{b_0^*})-V(x;\pi_{b_0^*}),\quad x>0,
\end{equation*}
and by taking derivative w.r.t. $\xi$, we get
$$H_0^\prime(\xi)=\beta-V'(x-\xi;\pi_{b_0^*})$$
which is always non-positive when $\mu\geq 0$. Hence, the supremum of $H_0$ on $[0,x]$ is attained at $\xi=0$ with value $H_0(0)=-\chi<0$. This shows that the left hand side of (\ref{eqt.HJB2}) is $\max(0,-\chi)=0$.

The above result is restated as the following theorem.

\begin{theorem}\label{Thm.Small.Beta}
	When $\mu\geq 0$ and $0\leq \beta\leq {\gamma}/{(\gamma+\delta)}$, the periodic $b_0^*$ strategy is optimal, where $b_0^*$ is specified in the third item in Proposition \ref{Lemma.convergence.beta}.
\end{theorem}

\section{The hybrid $(a_p,a_c,b)$ strategy}\label{S.hybridG}

In this section, we calculate the expected present value of dividends of a general hybrid $(a_p,a_c,b)$ strategy and then pick a candidate strategy from the class using a ``maximisation principle''. We will first use scale functions to derive some general results then use the classical PDE method to specialise to the case of diffusions. We make use of the fluctuation theory for L\'evy processes \citep[see, e.g.][Section 6, and references therein]{PeYa16b}.


\subsection{Value function}
Denote $\Psi(\theta)=1/t \log\E(e^{\theta X(t)})$, $\theta\in\mathbb{R}$, as the Laplace exponent of a spectrally negative L\'evy process $X$. 
Then for $q\geq 0$, the $q$-scale function $W_q$ is the mapping from $\mathbb{R}$ to $[0,\infty)$ that takes the value zero on the negative half-line, while on the positive half-line, it is a strictly increasing function that is defined by its Laplace transform:
\[\int_0^\infty e^{-\theta x}W_q(x)dx=\frac{1}{\Psi(\theta)-q},\quad \theta>\phi(q),
\quad \text{where} \quad
\phi(q)=\sup\{\lambda\geq 0:\Psi(\lambda)=q\}.\]
In particular, when $X$ is a diffusion process (defined by (\ref{Def.Diffusion})) and $q>0$, we have 
\begin{equation}\label{W.scalefcn.Diffusion}
W_q(x)=\frac{e^{r^{(q)} x}-e^{s^{(q)} x}}{\frac{\sigma^2}{2}(r^{(q)}-s^{(q)})}1_{\{x\geq 0\}},\quad x\in\mathbb{R},
\end{equation}
where $r^{(q)}>0$ and $s^{(q)}<0$ are the two distinct roots of 
$$\Psi(\theta)-q=0\iff \frac{\sigma^2}{2}\theta^2+\mu\theta-q=0.$$

In addition, we also define
\begin{align*}
W_{r,q,a}(x):=~&W_q(x)+r\int_0^{x-a}W_{q+r}(x-y-a)W_q(y+a)dy,\quad x\geq a,\\
\overline{W}_q(x):=~&\int_0^x W_q(y)dy,\quad x\geq 0,\\
\overline{\overline{W}}_q(x):=~&\int_0^x \overline{W}_q(y)dy,\quad x\geq 0.
\end{align*}

In the following, we slightly abuse the notation and assume that the barriers $(a_p,a_c,b)$ are given as $(a,a_c,b)$ and therefore denote the value function $V(x):=V(x;\pi_{a,a_c,b})$. If the dependence on the strategy $\pi$ or costs $1-\beta$ and $\chi$ needs to be stressed, we will write the full version $V(x;\pi_{a,a_c,b})$, or $V_{1-\beta,\chi}(\cdot)$, respectively. The value function $V$ is given by the following lemma.

\begin{lemma}
	\label{lemma.value.function.hybrid.scale}
	For a given hybrid $(a_p,a_c,b)$ strategy (with barrier levels $0< a_p=a\leq a_c<b$), its value function is continuous and is given by
	\begin{equation}\label{value.hybrid}
	V(x)=\begin{cases}
	\frac{V(a)}{\Wq(a)}\Wq(x)&x\in(-\infty,a),\\
	\frac{V(a)}{\Wq(a)}G(a,x-a)-\gamma\WqrBB(x-a)&x\in(a,b),\\
	\beta(x-a_c)-\chi+V(a_c)&x\in(b,\infty),
	\end{cases}
	\end{equation}
	with
	\begin{equation}
	G(a,x-a):=\Wq(x)+\gamma\int_0^{x-a}\Wqr(x-a-y)\Big(\Wq(a+y)-\Wq(a)\Big)dy
	\end{equation}
	where $V(a)$, $V(a_c)$ and $V(b)$ can be found by solving three linear equations in them.
	
	For barrier levels $0= a_p=a\leq a_c<b$, we denote $y=b-a_c,
	    l=a_c-a_p$ and use
	\begin{equation}\label{def.VW0}
	\frac{V(a)}{\Wq(a)}:=\frac{\beta y-\chi+\gamma(\WqrBB(y+l)-\WqrBB(l))}{G(a,y+l)-G(a,l)},
	\end{equation}
	which also holds for the above case when $a_p=a>0$.
\end{lemma}
\begin{proof}

	Using the notations introduced above, with a minor modification of the proofs in Sections 5.1, 6.1 and 6.2 in \citet{PeYa16b}, we can deduce that for $a_p>0$,
	$$
	V(x)=\frac{\Wq(x)}{\Wq(a_p)}V(a_p),\quad x\in(-\infty,a_p]
	$$
	and
	$$
	V(x)=V(a_p)(\frac{\Wa(x)}{\Wq(a_p)}-\gamma\Wb(x-a_p))-\gamma\Wbb(x-a_p),\quad x\in[a_p,b],
	$$
	which is the same as \eqref{value.hybrid} after some rearrangement. The case for $a=0$ can be carried over by a limit argument as in \citet{NoPeYaYa17}.
\end{proof}

From \eqref{value.hybrid} we see that
\begin{align*}
V(x)=~&\beta x -\chi+\Big(V(a_c)-\beta a_c\Big),\quad x>b,\\
V(x)=~&\Wq(x)\Big(\frac{V(a)}{\Wq(a)}\Big),\quad x\leq a.
\end{align*}
Hence it is reasonable to attempt maximisation of $V(a_c)-\beta a_c$ or $V(a)/\Wq(a)$ w.r.t. the parameters $(a,a_c,b)$ (and we will see both approaches are equivalent).

\subsection{Choice of candidate strategies}

We now proceed to pick a candidate strategy from the class of hybrid $(a_p,a_c,b)$ strategies.
A ``nice'' hybrid $(a_p,a_c,b)$ strategy is characterised by the derivatives of its value function at the boundaries, see e.g. \citet[Remark 9.2]{AvLaWo20d} for an intuitive explanation. We postulate (and later show) that those ``nice'' properties will lead to the optimal set of strategies, and hence refer to those as candidates.

\begin{definition}\label{Def.Nice.hybrid}
	A strategy is said to be a nice hybrid $(a_p,a_c,b)$ strategy if the following are satisfied:
	\begin{enumerate}
		\item It is a hybrid $(a_p,a_c,b)$ strategy (see Definition \ref{Def.hybrid.apacb});
		\item $b\geq a_c+{\chi}/{\beta}$ and $V'(b-)=\beta$;
		\item Either $a_c=a_p$ and $V'(0)\leq \beta$, or $V'(a_c)=\beta$;
		\item Either $a_p=0$ and $V'(0)\leq 1$, or $V'(a_p)=1$.
	\end{enumerate}
\end{definition}

In the following, we will re-parametrise $(a,a_c,b)$ using $(a,l,y)$ with $l:=a_c-a$ and $y:=b-a_c$. 
The support of $(a,l,y)$ is $[0,\bar{a}]\times [0,\infty)\times [\chi/\beta,\infty)$, where $\bar{a}$ is the unique solution for $\Wq''(x)=0$ if it exists, otherwise $\bar{a}=0$. 
We chose to maximise $V(a_c)-\beta a_c$ with respect to $(a,l,y)$. Regarding the auxiliary function $G$, it is easy to see 
\begin{equation}
\parD{a}G(a,d)=\parD{d}G(a,d)-\gamma \Wq'(a)\WqrB(d)
\quad
\text{and}
\quad
\parD{d}G(a,d)>0.
\end{equation}
For the derivatives of the value function at $a_p=a$, $a_c$ and $b$, we have
\begin{align}
V'(a)=~&\frac{V(a)}{\Wq(a)}\Wq'(a),\label{eq1a}\\
V'(a_c)=~&\frac{V(a)}{\Wq(a)}\parD{l}G(a,l)-\gamma \WqrB(l),\\
V'(b-)=~&\frac{V(a)}{\Wq(a)}\parD{(y+l)}G(a,y+l)-\gamma \WqrB(y+l).\label{eq1c}
\end{align}

We will first show that the derivative conditions are satisfied, provided a maximiser exists for our objective function
\begin{equation}\label{obj.fcn}
V(a_c)-\beta a_c=\frac{V(a)}{\Wq(a)}\Big(\Wq(x)+\gamma\int_0^{l}\Wqr(l-y)\Big(\Wq(a+y)-\Wq(a)\Big)dy\Big)-\gamma\WqrBB(l)-\beta(a+l).
\end{equation}
We will then show the existence of the maximiser. The following proposition illustrate the properties of a set of optimal parameters $(a,l,y)$, assuming its existence. For the moment, we make the following assumption which will be lifted by Proposition \ref{Prop.3}.

\begin{assumption}\label{Ass0}
	For any $a\geq 0$, we have 
	\begin{equation}
	\parD{x}\frac{G(a,x)}{\gamma\WqrB(x)}<0,\quad x> 0.
	\end{equation}
\end{assumption}

\begin{proposition}\label{prop.1}
	Denote $(a^*,l^*,y^*)$ a maximiser of the objective function $V(a_c)-\beta a_c$ and we recall the support is a subset of $a\in[0,\bar{a}]$. 
	Under Assumption \ref{Ass0}, if $(a^*,l^*,y^*)$ lies in the interior of the support, then we can conclude that with such choice of parameters, we have
	\begin{equation}\label{dev.con}
	V'(a^*)=1,\quad
	V'(a_c^*)=\beta,\quad
	V'(b^*-)=\beta.
	\end{equation}
	Otherwise, if $a^*=0$, then $V'(0)\leq 1$; if $l^*=0$, we have $a^*=l^*=0$ and $V'(0)\leq 1$. These are the only boundary cases.
\end{proposition}

\begin{proof} 
Since $(a^*,l^*,y^*)$ is a maximiser and the objective function is differentiable in the arguments, all the partial derivatives are zero (except in the boundary which requires extra care). In summary, the proof requires a direct checking in the argument $y$, then $l$ then $a$, assisted with the help of equations \eqref{eq1a}-\eqref{eq1c}; see Appendix \ref{A_Prop55} for details.
\end{proof}
	
Although the proof of Proposition \ref{prop.1} is simple and is similar to existing proofs in the literature \citep[e.g.,][]{Loe08a}, it presents the main ingredients in showing the existence of a candidate strategy characterised by \emph{three} non-zero parameters. This is generally a difficult problem since explicit calculation is often impossible. To our best knowledge, this is the first time such problem has been solved.

\begin{remark}
	From the proof in Appendix \ref{A_Prop55}, we see that maximising $V(a_c)-\beta a_c$ is the same as maximising $V(a)/\Wq(a)$. From the formula of $V(a)/\Wq(a)$ in \eqref{def.VW0}, we see that the $a$-argument of maximiser of $V(a)/\Wq(a)$ cannot live outside $[0,\bar{a}]$, which justifies our choice of a narrower support $a\in[0,\bar{a}]$.
\end{remark}

We now proceed to show the existence of a local maximiser $(a^*,l^*,y^*)$. Due to the complexity of the calculation using scale functions \citep[generally one assumes completely monotonic L\'evy density and proceed with complicated calculations, see e.g.][]{NoPeYaYa17}, we specialise our calculations using the classical PDE methods. 
	Denote the following functions:
\begin{align}
\psi(\theta):=~&\frac{\sigma^2}{2}\theta^2+\mu \theta\\
f(x):=~&e^{r_0 x}-e^{s_0 x},\\
g(x):=~&e^{r_1 x}-e^{s_1 x},\\
J(x):=~&-s_1g(x)+(r_1-s_1)(e^{s_1x}-1),\\
J'(x)=~&-r_1s_1g(x),
\end{align}
where $(r_0,s_0)$ and $(r_1,s_1)$ are the positive and negative roots of equations $\psi(\theta)-\delta=0$ and $\psi(\theta)-\gamma-\delta=0$ respectively, with $|s_i|>|r_i|$, $i=0,1$ (since $\mu>0$). Note $f$, $g$ and $J$ are proportional to the \emph{scale functions} $\Wq$, $\Wqr$ and $\WqrB$, respectively, see equation (\ref{W.scalefcn.Diffusion}). 

Before stating the value function in terms of $f$, $g$ and $J$, we shall discuss the smoothness of the value function for the PDEs to be solved. From the proof of Lemma \ref{lemma.value.function.hybrid.scale}, we can conclude that the value function $V$ is continuous, continuously differentiable except at $\{0\}$, and twice differentiable except at $0$ and at $b$, i.e. $V\in\mathscr{C}(\mathbb{R})\cap\mathscr{C}^1(\mathbb{R}\backslash\{0,b\})\cap\mathscr{C}^2(\mathbb{R}\backslash\{0,b\})$. 

The following proposition provides an alternative characterisation for the value function of a hybrid $(a_p,a_c,b)$ strategy. 
\begin{proposition}\label{Prop.Vfcn}
	For given ($a,a_c,b$) with $b>a_c+\chi/\beta$ and $ a_c\geq a\geq 0$, the value function of the hybrid $(a,a_c,b)$ strategy is given by
	\begin{equation}\label{eq.PDE}
	V(x;\pi_{a,a_c,b})=\begin{cases}
	0,&x\in(-\infty,0)\\
	C(e^{r_0 x}-e^{s_0 x}),&x\in[0,a)\\
	A(e^{r_1(x-a)}-e^{s_1(x-a)})+Be^{s_1(x-a)}+\frac{\gamma}{\gamma+\delta}(x-a+\frac{\mu}{\gamma+\delta}+V(a)),&x\in[a,b)\\
	\beta(x-a_c)-\chi+V(a_c),&x\in[b,\infty)
	\end{cases},
	\end{equation}
	with $l=a_c-a$, $d=b-a$, $g(d,l)=g(d)-g(l)$, $J(d,l)=J(d)-J(l)$,
	\begin{equation}\label{C.formula}
	C=\frac{(r_1-s_1)\big((\beta-\frac{\gamma}{\gamma+\delta})(d-l)-\chi\big)+\frac{\gamma}{\gamma+\delta}g(d,l)+\frac{\gamma\mu}{(\gamma+\delta)^2}J(d,l)}{\frac{\delta}{\gamma+\delta}f(a)J(d,l)+f'(a)g(d,l)},
	\end{equation}
	\begin{equation}\label{B.formula}
	B=\frac{\delta}{\gamma+\delta}Cf(a)-\frac{\gamma}{\gamma+\delta}\frac{\mu}{\gamma+\delta},
	\end{equation}
	\begin{equation}\label{A.formula}
	A=\frac{1}{r_1-s_1}\Big(Cf'(a)-Bs_1-\frac{\gamma}{\gamma+\delta}\Big)
	\end{equation}
	and
	\begin{align}
	V(a)=~&Cf(a)\\
	V(a_c)=~&Ag(l)+Be^{s_1 l}+\frac{\gamma}{\gamma+\delta}(l+\frac{\mu}{\gamma+\delta}+Cf(a)).
	\end{align}
	We also adopt the (unusual) convention that $[0,0)=\emptyset$ in \emph{(\ref{eq.PDE})}.
\end{proposition}

\begin{proof} Formulas can be derived by either directly substitute the formula of \eqref{W.scalefcn.Diffusion} in Lemma \ref{lemma.value.function.hybrid.scale}, or by a classical PDE approach.  \end{proof}

So far, Proposition \ref{prop.1} holds for general spectrally negative L\'evy processes (where Assumption \ref{Ass0} may or may not hold). We now specialise our results in the diffusion setting, and show in Proposition \ref{Prop.3} that Assumption \ref{Ass0} always holds for diffusion processes, so we can use the conclusion of Proposition \ref{prop.1} freely.

\begin{proposition}\label{Prop.3}
	Assumption \ref{Ass0} holds for diffusion processes.
\end{proposition}
\begin{proof}
	The result stems directly from the explicit formula given by Proposition \ref{Prop.Vfcn}; see Appendix \ref{A.1} for details.
\end{proof}

Thanks to Proposition \ref{Prop.Vfcn}, we have an explicit formula for the value function. We are now ready to show the following proposition.

\begin{proposition}\label{prop.2}
	There exists a triplet $(a^*,l^*,y^*)\in\mathscr{B}$ such that the ``derivative conditions'' \eqref{dev.con} hold.
\end{proposition}
\begin{proof}
     Thanks to Propositions \ref{Prop.3} and \ref{prop.1}, it remains to construct a large enough box to contain the maximum of the objective $V(a_c)-\beta a_c$. See Appendix \ref{A.prop2} for details.
\end{proof}

\subsection{Sufficient conditions for liquidation strategies}\label{Remark.suff.cons}
We are now able to derive some sufficient conditions for liquidation strategies to be optimal. Denote the functions
\begin{equation}
    \label{Q.function}
    Q(a):=1-\frac{f(a)/f'(a)}{\mu/\delta}
\end{equation}
and
\begin{equation}\label{I.def}
    I(x,q):=\frac{\beta-\frac{\gamma}{\gamma+\delta}+\Big(\frac{\gamma}{\gamma+\delta}-\frac{-s_1\frac{\gamma}{\gamma+\delta}}{-s_1+q(r_1+s_1)}\Big)
		e^{s_1 x}}{g'(x)+g(x)(-r_1s_1)\frac{\mu}{\gamma+\delta}(1-q)},
\end{equation}
where $Q$ maps the periodic lower barrier $a_p\in[0,\bar{a}]$ to a number $q\in[0,1]$ in an decreasing manner, whereas $I(\cdot,q)$ is a function decreasing to $0$ at infinity after it achieves its maximum.

	For a hybrid $(a_p,a_c,b)$ strategy, if $V'(b)=\beta$, we have (using the formula for $C$ given by \eqref{C.formula})
	
	\begin{align*}
	V'(a_p)-\frac{-s_1\frac{\gamma}{\gamma+\delta}}{-s_1+Q(a_p)(r_1+s_1)}=~&(r_1-s_1)
	I(y+l,q)\\
	<~&(r_1-s_1)\frac{\beta}{g'(g_0)}=:\varepsilon_k,
	\end{align*}
	with $g_0$ is a constant representing the minimum of the denominator of \eqref{I.def} w.r.t. $(x,q)$.
	
	Hence, if
	$$\frac{-s_1\frac{\gamma}{\gamma+\delta}}{r_1}\leq 1-\varepsilon_k,$$
	we will choose $a=0$ ($Q(a)=1$). If we consider $\varepsilon_k=0$, this condition is equivalent to the condition in Remark 4.1(i) in \citet{NoPeYaYa17}, i.e. $\gamma\leq({\sigma^2}/{2})r_1^2$.
	
	Likewise, if 
	$$\frac{-s_1\frac{\gamma}{\gamma+\delta}}{r_1}\leq \beta-\varepsilon_k,$$
	we will choose $a=l=0$ ($a_p=a_c=0$). In fact, for any $q$ such that
	$\frac{-s_1\frac{\gamma}{\gamma+\delta}}{r_1+q(r_1+s_1)}\geq \beta,$
	we have $I(0,q)\leq 0$ and therefore for any $q$ there are always $(l,y+l)$ with $l>0$ such that $I(l,q)=I(y+l,q)$. This implies that when $$1-\varepsilon_k\geq \frac{-s_1\frac{\gamma}{\gamma+\delta}}{r_1}\geq \beta,$$
	we will choose $a=0$ but $l>0$ ($a_p=0<a_c$).
	
	In practice, $\varepsilon_k$ is usually negligible, compared to the size of $\beta$. Likewise, we can see that once we have the smoothness condition $V'(b)=\beta$, the right hand side is (almost) negligible and therefore the derivative at $a$ is (almost) independent of both $y$ and $l$.
	
	In terms of the parameters in the model, we have
	$$\frac{-s_1}{r_1}=1+\frac{1+\sqrt{1+2\nu}}{\nu},\quad \text{ with }\quad \nu:=\Big(\frac{\sigma}{\mu}\Big)^2(\gamma+\delta)$$
	and therefore the above two sufficient conditions can be rewritten as $$\frac{\gamma}{\gamma+\delta}\frac{1+\sqrt{1+2\nu}}{\nu}+\varepsilon_k\leq \frac{\delta}{\gamma+\delta}
	\quad \text{and} \quad
	\frac{\gamma}{\gamma+\delta}\frac{1+\sqrt{1+2\nu}}{\nu}+\varepsilon_k\geq \beta-\frac{\gamma}{\gamma+\delta}.$$
	Note further 
	$({1+\sqrt{1+2x}})/{x}$
	is decreasing in $x$, hence there are thresholds $\nu_1(\gamma,\delta)$ and $\nu_\beta(\gamma,\delta)$ (with $\nu_\beta>\nu_1$ unless $\beta\geq 1-\varepsilon_k$) such that 
	$$\begin{cases}
	\nu_\beta(\gamma,\delta)\geq\Big(\frac{\sigma}{\mu}\Big)^2(\gamma+\delta)\geq \nu_1(\gamma,\delta)&\implies\quad a=0, ~l>0\\
	\Big(\frac{\sigma}{\mu}\Big)^2(\gamma+\delta)\geq \nu_\beta(\gamma,\delta)&\implies\quad a=0,~l=0\\
	\end{cases}.$$
	
	Note the sufficient conditions for liquidation at first opportunity ($a=0$) do not depend on the transaction costs $1-\beta$ and $\chi$ as one can always ignore the opportunities to pay immediate dividends. However, it does depend on the time parameters for the frequency periodic payments and discounting ($\gamma$ and $\delta$), as well as the ``coefficient of variation'' ${\sigma}/{\mu}$, a measurement of the riskiness of the business.
	

\section{The derivative of the value function of our candidate strategy}\label{S.VD}

From the previous section, we see that there exists a nice hybrid $(a_p,a_c,b)$ strategy (see Definition \ref{Def.Nice.hybrid} and Propositions \ref{prop.1}, \ref{prop.2}). In other words, there are $(a_p,a_c,b)$ such that $a_p\leq \bar{a}$, $V'(a_p)=1$ (or $a_p=0$ and $V'(0)\leq 1$), $V'(a_c)=\beta$ (or $a_p=a_c=0$ and $V'(0)\leq \beta$) and $V'(b)=\beta$. We will pick this strategy and refer it as our candidate strategy, and use $V$ to denote its value function.

In this section, we first establish some results regarding the derivative of the value function of our candidate strategy. They will then be used to verify the optimality of our strategy in Section \ref{S.Optimality}.

As explained before (e.g., Sections \ref{S_map} and \ref{S_betasucks}) we must have
\begin{equation}\label{Ass.betaLarge}
\frac{\gamma}{\gamma+\delta}<\beta\leq 1,
\end{equation}
which becomes apparent in some areas of the proof. 

Our goal in this section is to establish Lemma \ref{Lemma.VD.Nice}. In order to do that, we need to first establish Lemma \ref{Lemma.VD.opt} below, which concerns the behaviour of the derivative of the value function of our candidate strategy.

\begin{lemma}\label{Lemma.VD.opt}
	Regarding the derivative of the value function, we have the following:
	\begin{enumerate}
		\item Suppose $a_p>0$, then 
		\begin{equation}\label{eqt.VD.opt1}
		V'(x)\begin{cases}
		>1,~&x\in[0,a_p)\\
		=1,~&x=a_p\\
		\in (\beta,1),~&x\in[a_p,a_c)\\
		=\beta,~&x=a_c\\
		\in(0,\beta),~&x\in(a_c,b)\\
		=\beta,~&x\in[b,\infty)
		\end{cases}.
		\end{equation}
		\item Suppose $a_p=0$ and $a_c>0$, then 
		\begin{equation}\label{eqt.VD.opt2}
		V'(x)\begin{cases}
		\in (\beta,1],~&x=0\\
		\in (\beta,1),~&x\in(0,a_c)\\
		=\beta,~&x=a_c\\
		\in(0,\beta),~&x\in(a_c,b)\\
		=\beta,~&x\in[b,\infty)
		\end{cases}.
		\end{equation}
		\item Suppose $a_p=a_c=0$, then 
		\begin{equation}\label{eqt.VD.opt3}
		V'(x)\begin{cases}
		\in (0,\beta],~&x=0\\
		\in(0,\beta),~&x\in(0,b)\\
		=\beta,~&x\in[b,\infty)
		\end{cases}.
		\end{equation}
	\end{enumerate}
	In any case, we have $V'>0$ on $[0,\infty)$.
\end{lemma}
\begin{proof}
The proof requires analysing the functional form of the value function with the derivative conditions imposed for the candidate strategy; see Appendix \ref{A.Lemma6.1} for details.
\end{proof}

The next Lemma shows that our candidate strategy satisfies the last 2 conditions in Lemma \ref{Verification.lemma}.

\begin{lemma}\label{Lemma.VD.Nice}
	The value function of a nice hybrid $(a_p,a_c,b)$ strategy, $V$, satisfies
	\begin{equation}
	(\mathscr{A}-\delta)V(x)+\gamma \sup_{\xi\in[0,x]}\Big(\xi+V(x-\xi)-V(x)\Big)
	\leq 0, \quad x\in\mathbb{R}^+\backslash\{b\}
	\end{equation}
	and
	\begin{equation}
	\sup_{\xi\in[0,x]}\Big((\beta\xi-\chi)1_{\{\xi>0\}}+V(x-\xi)-V(x)\Big)
	= 0,\quad x\in\mathbb{R}^+.
	\end{equation}
\end{lemma}

\begin{proof} The result comes as a straightforward consequence of Lemma \ref{Lemma.VD.opt}; see Appendix \ref{A.Lemma6.2} for details.
\end{proof}

\section{Optimality in case of profitable business ($\mu \ge 0$)}\label{S.Optimality}

In this section, we show the optimality of a nice hybrid $(a_p,a_c,b)$ strategy, which is the following theorem.

\begin{theorem}\label{Thm}
	Suppose $\mu\geq 0$. Denote $V$ the value function of a nice hybrid $(a_p,a_c,b)$ strategy. Suppose $\pi\in\Pi$, then $V(x)\geq V(x;\pi)$ for all $x\geq 0$. In other words, any nice hybrid $(a_p,a_c,b)$ strategy is optimal.
\end{theorem}

\begin{proof}
	Thanks to Lemma \ref{Verification.lemma}, we only need to check that $V$ satisfies all conditions proposed, which is essentially Lemma \ref{Lemma.VD.opt} and \ref{Lemma.VD.Nice}, with the finite set $E=\{b\}$.
\end{proof}

We now present a corollary regarding the uniqueness of nice hybrid $(a_p,a_c,b)$ strategies.

\begin{corollary}\label{Corrolary.unique.barriers}
	Suppose $\mu\geq 0$. There is one and only one nice hybrid $(a_p,a_c,b)$ strategy. Denote its parameters $(a_p^*,a_c^*,b^*)$. Hence, the hybrid $(a_p^*,a_c^*,b^*)$ strategy is optimal.
\end{corollary}
\begin{proof}
	Lemma \ref{Lemma.VD.opt} characterised the derivative of the value function, which together with the optimality implies uniqueness. A similar proof can be found in \citet[Lemma 9.3]{AvLaWo20d}.
\end{proof}

\section{Optimality in case of unprofitable business ($\mu<0$)}\label{S.mu.neg}

In this section, we discuss the optimal strategy when the business is strictly unprofitable, i.e. $\mu<0$. As such, solely in this section, we make the following assumption.
\begin{assumption}\label{Ass.Mu.Neg}
	We assume $\mu<0$.
\end{assumption}

Recall that $\mu<0$ implies that we want to liquidate the business in the most (cost-)efficient way. 
In the following, we focus on the case when $\chi>0$ and the case for $\chi=0$ can be seen as the case when $\chi\downarrow 0$ in terms of the structure of the optimal strategy.

Our candidate strategies are the Liquidation $(b_1,b_2)$ strategy, characterised by 2 parameters $0<b_1<b_2\leq \infty$, denoted as $\pi_{b_1,b_2}$ (see Definition \ref{Def.periodic.b}) and the periodic $0$ strategy, denoted as $\pi_0$  (see Definition \ref{def.Liq}). The latter pays $X^{\pi_0}(T_1-)$ when $T_1\leq \tau^{\pi_0}$, and it can also be seen as the limit of a liquidation $(b_1,b_2)$ strategy when $b_2-b_1\downarrow 0$, or simply $b_1\uparrow \infty$. Therefore, it is also denoted as $\pi_{a,a}$, or $\pi_{\infty,\infty}$.

For $\beta> {\gamma}/{(\gamma+\delta)}$, it should be intuitively clear that the form of $\pi_{b,\infty}$ is optimal if we can choose the lower barrier $b$ nicely. On the other hand, for $\beta<{\gamma}/{(\gamma+\delta)}$, we proceed the following. It is known that \citep[e.g. from][]{PeYa16} that $V'(x;\pi_0)$ is increasing in $x$ to ${\gamma}/{(\gamma+\delta)}$. Therefore, for $V'(0;\pi_0)<\beta<{\gamma}/{(\gamma+\delta)}$,  there is a unique $a_\beta>0$ such that 
$V'(a_\beta;\pi_0)=\beta$. If moreover 
\begin{equation}\label{Condition.Liqb1b2.Optimal}
V(a_\beta;\pi_0)<\beta a_\beta-\chi,
\end{equation}
then there is a unique $c_{\beta,\chi}$ such that $0<c_{\beta,\chi}<a_\beta$ with
$V(c_{\beta,\chi};\pi_0)=\beta c_{\beta,\chi}-\chi.$ Note that (\ref{Condition.Liqb1b2.Optimal}) is equivalent to
\begin{equation}\label{Condition.Liqb1b2.Optimal2}
\frac{\chi}{\beta}<a_\beta-\frac{V(a_\beta;\pi_0)}{\beta}.
\end{equation}
Denote the right hand side as a function of $\beta$, i.e. \begin{equation}\label{Def.Lambda}
\Lambda(\beta)=a_\beta-\frac{V(a_\beta;\pi_0)}{\beta},
\end{equation}
then we have that
\begin{equation}\label{Property.Lambda}
\text{$\Lambda$ is increasing from $0$ to the limit $\frac{-\mu}{\gamma+\delta}$ when $\beta$ increases on the interval $[V'(0;\pi_0),\frac{\gamma}{\gamma+\delta})$;}
\end{equation}
a proof of which is provided in Appendix \ref{A.Proof.Lambda}. Hence, (\ref{Condition.Liqb1b2.Optimal}) is only possible when $\chi<\beta({-\mu}/({\gamma+\delta}))$. To further explain what it means, note when the surplus is $x$ and we can choose either (1) liquidate now, or (2) liquidate in the next Poissonian time, we need to consider the trade-off. If we liquidate now, the fixed cost is $\chi$. If we wait, assuming ruin is not an issue, the \emph{expected} (discounted) loss in surplus is then $$\E(-\mu T_1 e^{-\delta T_1})=\frac{-\gamma\mu}{(\gamma+\delta)^2},$$where we recall that $T_1$ is an exponential random variable with mean $1/\gamma$. Hence, when $\chi\geq \beta({-\mu}/({\gamma+\delta}))$ and $\beta\leq {\gamma}/{(\gamma+\delta)}$, we shall never liquidate immediately. On the other hand, if $\chi<\beta({-\mu}/({\gamma+\delta}))$, then in view of (\ref{Condition.Liqb1b2.Optimal2}) and (\ref{Property.Lambda}), there is a $\beta_0\in(V'(0;\pi_0),{\gamma}/{(\gamma+\delta)})$ defined by $\Lambda(\beta_0)={\chi}/{\beta}$ such that (\ref{Condition.Liqb1b2.Optimal}) does not hold whenever $\beta\in(V'(0),\beta_0)$ and (\ref{Condition.Liqb1b2.Optimal}) holds whenever $\beta\in(\beta_0,{\gamma}/{(\gamma+\delta)})$.

In light of the above analysis, we should not be surprised with the results in this section. They are summarised by the following theorem.

\begin{theorem}\label{Thm.mu.neg}
	For $\mu<0$, we have the following:
	
	\begin{enumerate}
		\item[Case 1:] $\chi\geq \beta\frac{-\mu}{\gamma+\delta}$. We have
		\begin{enumerate}
			\item For $\beta\in(0,{\gamma}/{(\gamma+\delta)}]$, the periodic $0$ strategy is optimal.
			\item For $\beta\in({\gamma}/{(\gamma+\delta)},1]$, a liquidation $(b,\infty)$ strategy is optimal with $b>0$ characterised by
			$$V'(b-;\pi_{b,\infty})=\beta=V'(b+;\pi_{b,\infty}).$$
		\end{enumerate}
		
		\item[Case 2:] $\chi<\beta\frac{-\mu}{\gamma+\delta}$. Denote $\beta_0:=\Lambda^{-1}(\frac{\chi}{\beta})$, we have
		\begin{enumerate}
			\item For $\beta\in(0,\beta_0]$, the periodic $0$ strategy is optimal.
			\item For $\beta\in(\beta_0,{\gamma}/{(\gamma+\delta)})$, a liquidation $(b_1,b_2)$ strategy is optimal, with $(b_1,b_2)$ such that $0<c_{\beta,\chi}<b_1<a_\beta<b_2<\infty$ and
			$$V'(b_1-;\pi_{b_1,b_2})=V'(b_1+;\pi_{b_1,b_2})=\beta=V'(b_2-;\pi_{b_1,b_2})=V'(b_2+;\pi_{b_1,b_2}).$$
			\item For $\beta\in[{\gamma}/{(\gamma+\delta)},1]$, a liquidation $(b,\infty)$ strategy is optimal, with $b>0$ characterised by
			$$V'(b-;\pi_{b,\infty})=\beta=V'(b+;\pi_{b,\infty}).$$
		\end{enumerate}
	\end{enumerate}
	
\end{theorem}

In order to prove Theorem \ref{Thm.mu.neg}, we will need the following lemmas. Our first lemma calculates the value function for each strategy.

\begin{lemma}\label{Lemma.ValueFunction}
	The value function of a periodic $0$ strategy is given by
	\begin{equation}
	V(x;\pi_0)=-\frac{\gamma\mu}{(\gamma+\delta)^2}e^{s_1x}+\frac{\gamma}{\gamma+\delta}(x+\frac{\mu}{\gamma+\delta}),\quad x\geq 0.
	\end{equation}
	
	The value function of a liquidation $(b_1,b_2)$ strategy (with $0<b_l<b_2<\infty$) is given by 
	\begin{equation}\label{value.Liq.b1b2}
	V(x;\pi_{b_1,b_2})=\begin{cases}
	Ag(x)-\frac{\gamma\mu}{(\gamma+\delta)^2}e^{s_1x}+\frac{\gamma}{\gamma+\delta}(x+\frac{\mu}{\gamma+\delta}),\quad& x\in[0,b_1)\\
	\beta x-\chi,\quad& x\in[b_1,b_2)\\
	Be^{s_1 x}+\frac{\gamma}{\gamma+\delta}(x+\frac{\mu}{\gamma+\delta}),\quad& x\in[b_2,\infty)
	\end{cases},
	\end{equation}
	with
	\begin{equation}\label{ADef}
	A=A(b_1):=\frac{(\beta-\frac{\gamma}{\gamma+\delta})b_1-\chi-\frac{\gamma\mu}{(\gamma+\delta)^2}(1-e^{s_1b_1})}{g(b_1)}
	\end{equation}
	and $B$ can be determined using $V(b_2-)=V(b_2)$, in the case where $b_2<\infty$. 
	
	The value function of a liquidation $(b,\infty)$ strategy, with $b>0$ is simply (\ref{value.Liq.b1b2}) without the $x\in[b_2,\infty)$ branch, where the formula for $A$ is still the same and we do not need to compute $B$.
\end{lemma}

We omit the proof for Lemma \ref{Lemma.ValueFunction} as it can be obtained easily by solving PDE with the continuity of the value functions on the boundaries being the boundary conditions.

\begin{remark}\label{Remark.mu.neg.obs}
	The following results will be repeatedly used in what follows, $$V(x;\pi_{b_1,b_2})=V(x;\pi_0)+Ag(x),\quad \text{and} \quad
	V''(x;\pi_{b_1,b_2})=Ag''(x)-\frac{\gamma\mu}{(\gamma+\delta)^2}s_1^2e^{s_1x},\quad x\leq b_1,$$
	where $b_2$ can possibly be infinity. Therefore, we have $V''(x)>0$ if $A\geq 0.$
\end{remark}

The next lemma is also an argument which will be used over time.
\begin{lemma}\label{Lemma.A.increasing}
	For a liquidation $(b,b_2)$ strategy, with $0<b<b_2\leq \infty$, we have
	$$\parD{b}A(b)>0\text{ if and only if }V'(b-;\pi_{b,b_2})<\beta.$$
	We can replace the inequalities by equalities simultaneously.
\end{lemma}
\begin{proof}
	This can be derived through direct computation using (\ref{ADef}).
\end{proof}

The next 2 lemmas establish the existence of the candidate strategy described in Theorem \ref{Thm.mu.neg} (for every possible case).

\begin{lemma}\label{Lemma.b1b2.Exist}
	If $V'(0;\pi_0)<\beta<\frac{\gamma}{\gamma+\delta}$ and $V(a_\beta;\pi_0)<\beta a_\beta-\chi$, then there are $(b_1,b_2)$ with $0<c_\beta<b_1<a_\beta<b_2<\infty$ such that $V'(b_1-;\pi_{b_1,b_2})=V'(b_1+;\pi_{b_1,b_2})=\beta=V'(b_2-;\pi_{b_1,b_2})=V'(b_2+;\pi_{b_1,b_2}).$
\end{lemma}
\begin{proof}
The proof is based on continuity arguments. For example, for $b_1$, if one denotes the objective function $b_1\mapsto \widetilde{O}(b_1):=V'(b_1+)-\beta$, it suffices to show that $\widetilde{O}(c_\beta)\widetilde{O}(a_\beta)<0$. Details are provided in Appendix \ref{A.Lemma8.6}. 
\end{proof}

\begin{lemma}\label{Lemma.b.Exist}
	If ${\gamma}/{(\gamma+\delta)}<\beta\leq 1$, then there is a $b>{\chi}/{\beta}$ such that $V'(b-;\pi_{b,\infty})=\beta$.
\end{lemma}
\begin{proof}
The proof is similar to that of Lemma \ref{Lemma.b1b2.Exist}; see Appendix \ref{A.Lemma8.7}
\end{proof}

Lemmas \ref{Lemma.b1b2.Exist}-\ref{Lemma.b.Exist} shows that our candidate strategy (as described in Theorem \ref{Thm.mu.neg}) exists. 
Our last lemma below shows that the derivative of its value functions is increasing, which essentially completes the proof of Theorem \ref{Thm.mu.neg}.

\begin{lemma}
	In each case considered in Theorem \ref{Thm.mu.neg}, the derivative of the value function is increasing for our candidate strategy.
\end{lemma}
\begin{proof}
	We only need to prove the case when the liquidation $(b_1,b_2)$ strategy (with $b_2\leq \infty$) is optimal.
	
	From the proof of Lemma \ref{Lemma.b1b2.Exist}, if we choose $b_1$ to be the smallest one, then we have that $A(b_1)>0$, as $A$ is increasing from $0$ at $c_\beta$ to $b_1$. This implies that $V''(x;\pi_{b_1,b_2})>0$ and hence $V'(x;\pi_{b_1b_2})$ is an increasing function on $[0,b_2]$. Now, from $$V'(b_2-;\pi_{b_1,b_2})=Bs_1e^{s_1b}+\frac{\gamma}{\gamma+\delta}=\beta,$$
	we can deduce that $B>0$ and consequently that $V'(x;\pi_{b_1,b_2})$ is also increasing on $[b_2,\infty)$. This completes the proof for the case when  $V'(0;\pi_0)<\beta<{\gamma}/{(\gamma+\delta)}$ and $V(a_\beta;\pi_0)<\beta a_\beta-\chi$.
	
	From the proof of Lemma \ref{Lemma.b.Exist}, if we choose $b$ to be the smallest one, using the same argument, we have that $A(b)>0$, which shows that $V'(x;\pi_{b,\infty})$ is increasing on $[0,b]$ and hence on $[0,\infty)$.
\end{proof}

Note the value function of the candidate strategy in each case is continuously differentiable and twice differentiable except at the boundary points. Therefore the first 3 conditions of Lemma \ref{Verification.lemma} are automatically satisfied.
What is yet to be shown are the last 2 conditions of Lemma \ref{Verification.lemma}, which can be done in a similar way to what has been done in Sections \ref{S.Verification} and \ref{S.VD}. Hence, Theorem \ref{Thm.mu.neg} holds. 

\section{On the convergence of strategies across solution thresholds}\label{S.Convergence}

In this section, we discuss the convergence of strategies when $\mu\uparrow 0$ and when $\beta\downarrow{\gamma}/{(\gamma+\delta)}$. These correspond to major thresholds in Table \vref{S_map}. By convergence in strategy, we mean point-wise converge of the value function at each $x\geq 0$. Since barriers determine the value function, convergence in the barriers implies convergence in strategy.

\subsection{When $\mu\uparrow 0$}\label{S.VaryingMu}

{The objective of this section is to investigate how the strategies and their parameters ``connect'' when $\mu$ goes from negative to positive.}

{We start by some numerical exploration.} Our baseline parameters are $(\sigma,\chi)=(0.3,0.15)$ (dollar parameters), and $(\gamma,\delta)=(1,0.15)$ (time parameters). Note that the crucial constant $\gamma/(\gamma+\delta)\approx 0.87$, and the form of optimal strategies (rather, their connection here) will potentially differ on either side of that constant. Hence, we will illustrate the cases $\beta=0.7$ and $\beta=0.9$ separately.

\begin{figure}[htb]
	\centering
	\begin{minipage}{0.4\textwidth}
		\centering
		\includegraphics[width=1\textwidth]{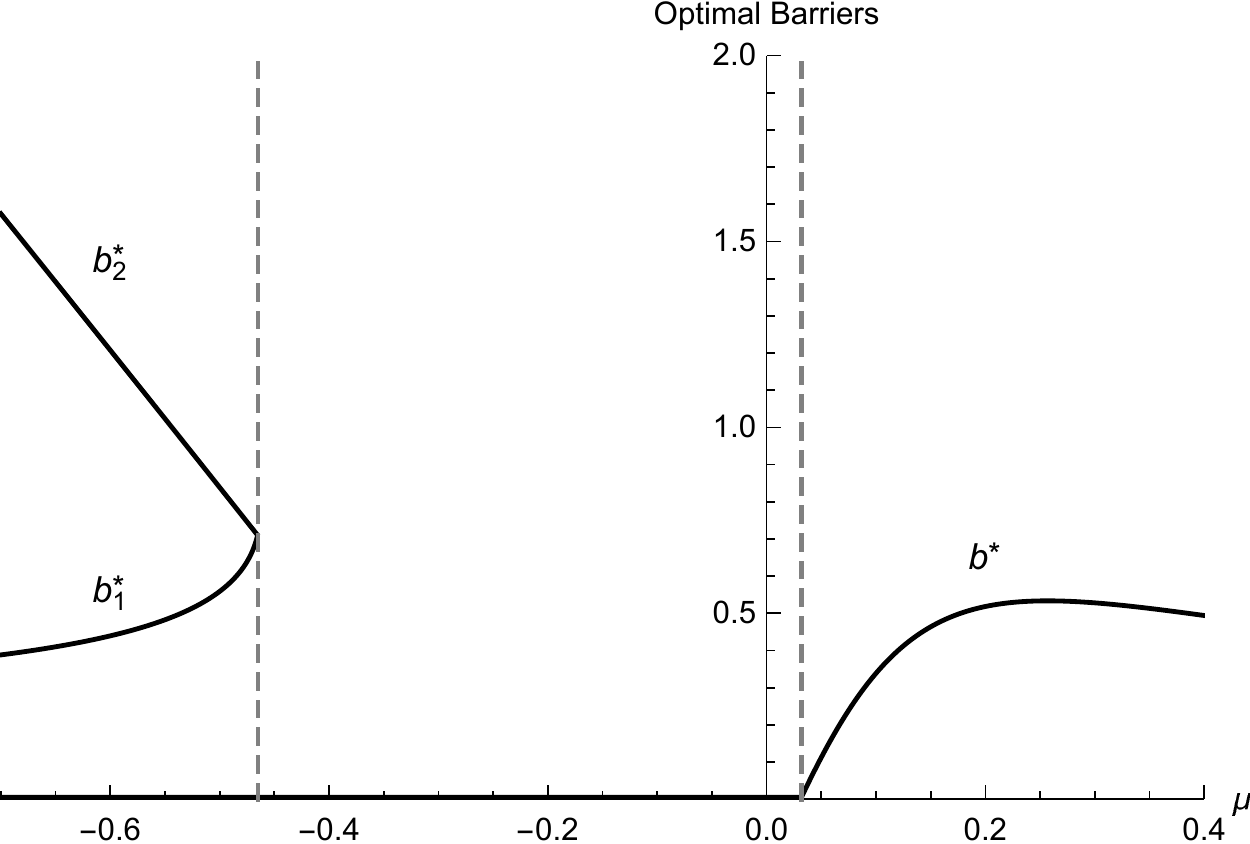}
		\subcaption[first caption.]{$\beta\corr{~<\gamma/(\gamma+\delta)}$}\label{fig.a}
	\end{minipage}%
	\hspace{0.08\textwidth}
	\begin{minipage}{0.4\textwidth}
		\centering
		\includegraphics[width=1\textwidth]{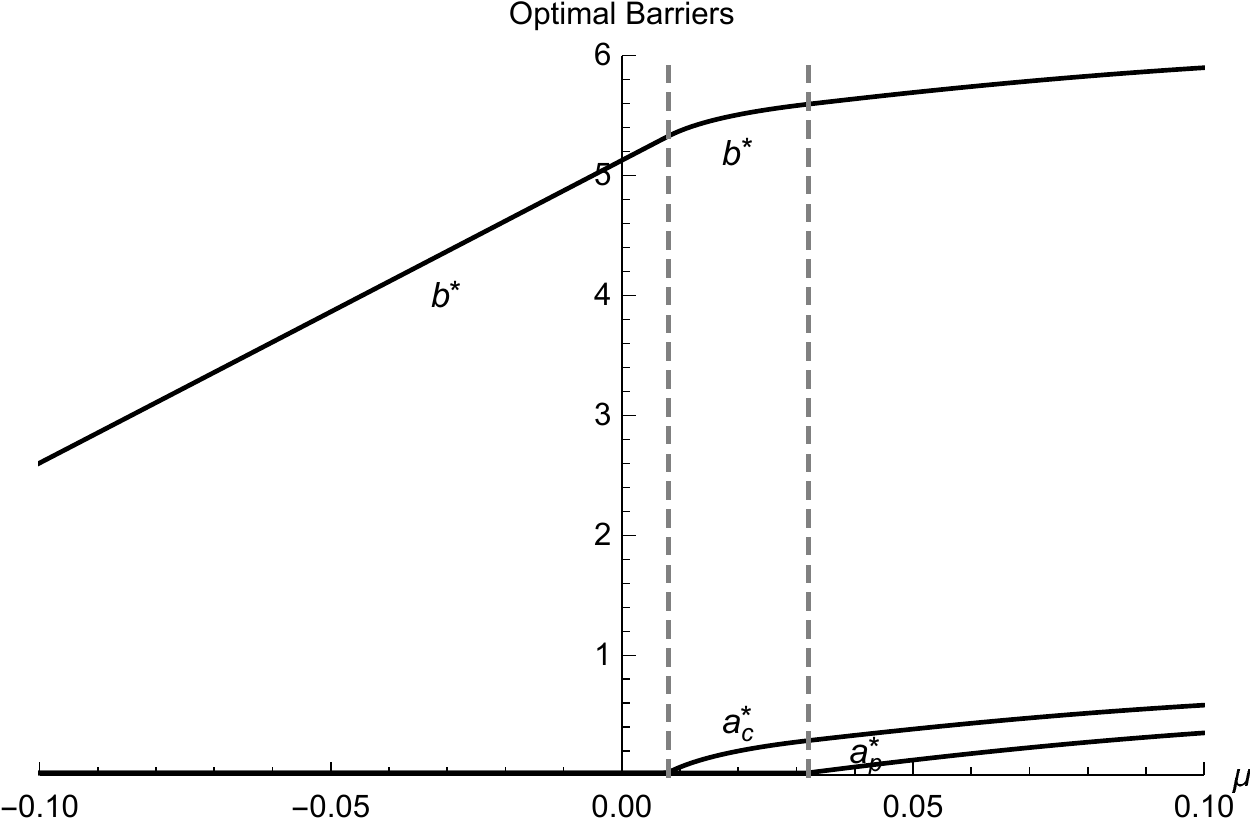}
		\subcaption[second caption.]{$\beta\corr{~>\gamma/(\gamma+\delta)}$}\label{fig.b}
	\end{minipage}%
	\caption{Continuity of optimal barriers at $\mu=0$\corr{. Left:$\beta=0.7$. Right: $\beta=0.9$.}} \label{fig.TB}
\end{figure}

From Figure \ref{fig.a} (high proportional transaction cost $1-\beta$), we can see that around the neighbourhood of $0$, the periodic $0$ strategy is optimal. On the other hand, when the proportional cost is low, we can see from Figure \ref{fig.b} that a liquidation $(b,\infty)$, or equivalently a hybrid $(0,0,b)$ strategy is optimal around the neighbourhood of $0$. This observation is true in general thanks to Lemma \ref{L_conv} below.

Further analysis of Figure \ref{fig.a} is interesting. The vertical grey dashed line in Figure \ref{fig.a} corresponds to the threshold between the second-last and last columns of Table \ref{S_map} (in the second row), so that obviously $\beta_0<0.7$ here. The right-hand side (where $b_0$ becomes strictly positive) corresponds to the first column.


\begin{lemma}\label{L_conv}
	When $\beta>{\gamma}/{(\gamma+\delta)}$, a hybrid $(0,0,b)$ strategy is optimal for $\mu=0$. Moreover, when $\mu\uparrow 0$, the lower barrier of the optimal $(b_1,\infty)$ strategy, denoted as $b_1=b_1(\mu)$ converges to $b$.
	When $\beta\leq {\gamma}/{(\gamma+\delta)}$, the periodic barrier strategy with barrier level $0$, $\pi_0$, is optimal on the neighbourhood of $\mu=0$.
\end{lemma}

\begin{proof}
	For each case, one needs to check the corresponding barrier converges to the barrier at $\mu=0$ when $\mu\uparrow 0$; see details in Appendix \ref{A.Lemma9.1}.
\end{proof}

\subsection{When $\beta\downarrow \gamma/(\gamma+\delta)$}\label{SubS.Convergence.Beta}

The convergence of the barriers when $\beta\downarrow{\gamma}/{(\gamma+\delta)}$ is described in the following proposition.

\begin{proposition}\label{Lemma.convergence.beta}
	Recall the function $Q$ in \eqref{Q.function}. When $\beta\downarrow\frac{\gamma}{\gamma+\delta}$, we have the following:
	
	\begin{enumerate}
		\item If $\frac{-s_1}{r_1}\frac{\gamma}{\gamma+\delta}>1$, $$ a_p=Q^{-1}\Big(\frac{s_1\frac{\delta}{\gamma+\delta}}{r_1+s_1}\Big),\quad a_c=\infty,\quad b=\infty,$$
		where ``$=$'' is in limit sense.
		\item If $\frac{-s_1}{r_1}\frac{\gamma}{\gamma+\delta}\leq 1$, 
		$$ a_p=0,\quad a_c=\infty,\quad b=\infty,$$
		where ``$=$'' is in limit sense.
		\item $Q^{-1}\Big(\frac{s_1\frac{\delta}{\gamma+\delta}}{r_1+s_1}\Big)=b_0^*$ and $\lim_{a_c\rightarrow \infty,b\rightarrow\infty}V(x;\pi_{a_p,a_c,b})=V(x;\pi_{b_0^*})$ for all $x\geq 0$. That is, both the barrier and the value function exhibit continuity behaviour at $\beta=\frac{\gamma}{\gamma+\delta}$.
	\end{enumerate}
	
\end{proposition}

\begin{proof}
	 The proof requires functions $Q$ and $I(x,q)$ introduced in Section \ref{Remark.suff.cons}. The first two parts requires an investigation of the condition $I(l,Q(a))=I(l+y,Q(a))$, whereas the last part calculates require a verification of the condition $V'''(b^*_0+)=V'''(b^*_0-)$ with the given formula for $b^*_0$. Details are provided in Appendix \ref{A.proof.Convergence.beta}.
\end{proof}

\subsection{Continuity for different cases in Theorem \ref{Thm.mu.neg}}

Note when $\mu<0$, similar continuity results hold, as shown in Appendix \ref{A_Thm.mu.neg.cont} sequentially for the $4$ different cases enumerated in Theorem \ref{Thm.mu.neg}.

\section{Numerical illustrations}\label{S.Numerical}

In this section, We illustrate numerically some results from previous sections. The first and the second sections are devoted to the case when $\mu>0$ and $\mu<0$ respectively. The {connection between both cases (when $\mu=0$) has been discussed in Section \ref{S.VaryingMu}.}

\subsection{When the business is profitable ($\mu> 0$)}\label{S.numericalP}

Our baseline setting includes: scale parameters $(\mu,\sigma,\chi)=(1,0.3,0.01)$, time parameters $(\gamma,\delta)=(1,0.15)$ and proportional transaction cost parameter $\corr{\beta=0.9}$. \corr{Except the parameter under consideration, all other parameters will be set to the baseline.} In particular, \corr{under the baseline,} we have $\beta>{\gamma}/{(\gamma+\delta)}$ which guarantees that the hybrid $(a_p^*,a_c^*,b^*)$ strategy is optimal. In the following, we will study the impact of the parameters on the optimal barriers.

\subsubsection{Transaction costs}
\begin{figure}[htb]\label{fig.muP.costs}
	\centering
	\begin{minipage}{0.33\textwidth}
		\centering
		\includegraphics[width=1\textwidth]{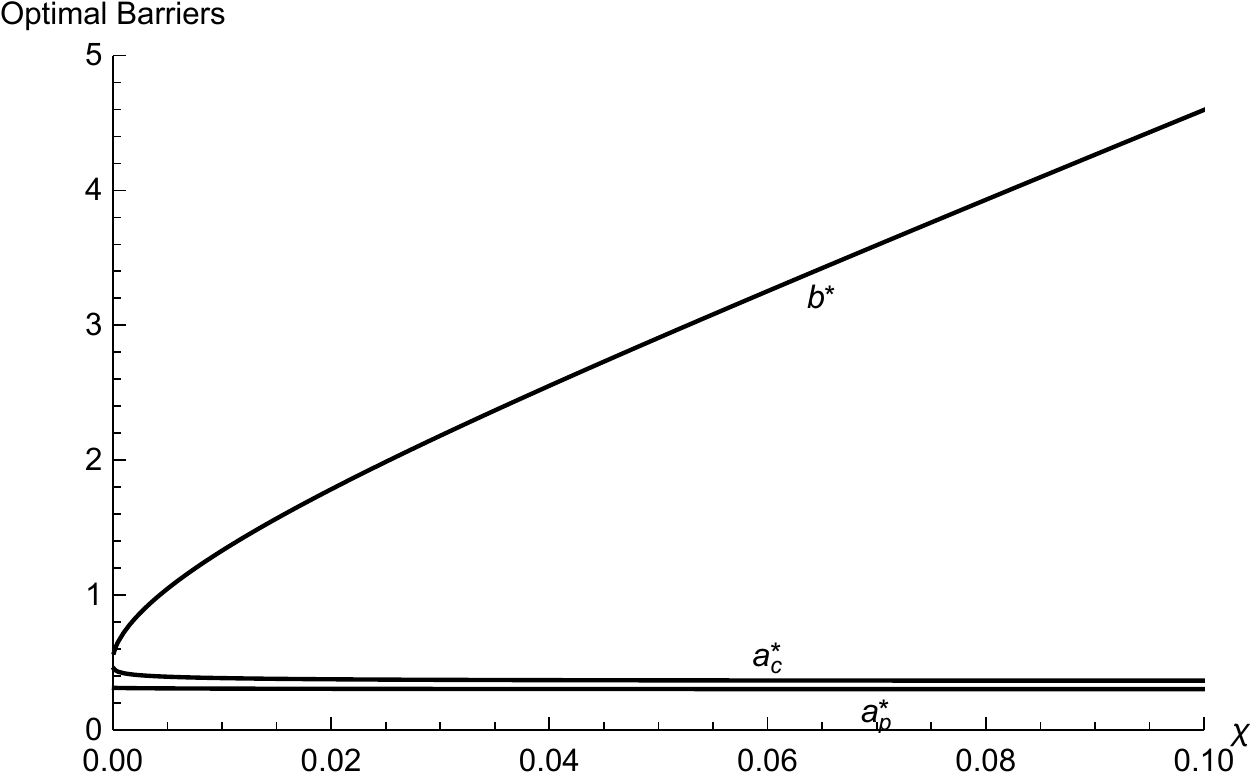}
		\subcaption[first caption.]{Fixed cost $\chi$}\label{fig.muP.a}
	\end{minipage}
	\begin{minipage}{0.33\textwidth}
		\centering
		\includegraphics[width=1\textwidth]{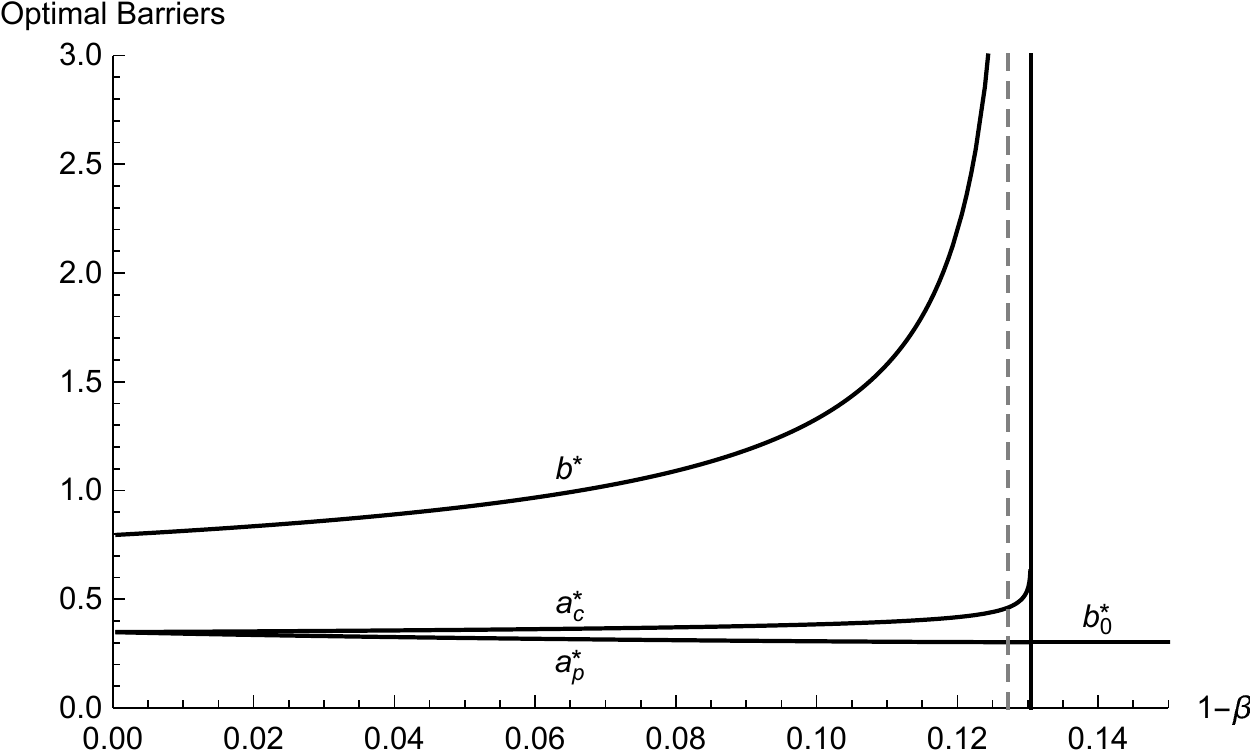}
		\subcaption[second caption.]{Proportional cost $1-\beta$}\label{fig.muP.b}
	\end{minipage}%
	\begin{minipage}{0.33\textwidth}
		\centering
		\includegraphics[width=1\textwidth]{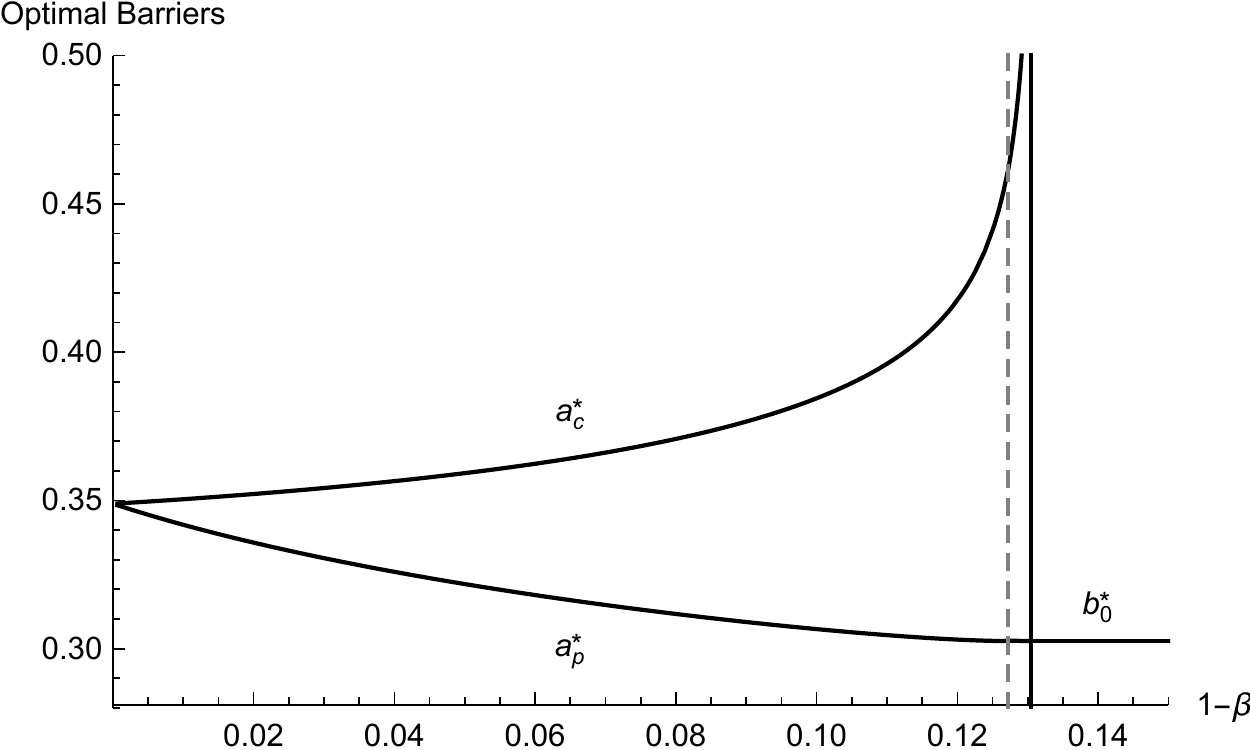}
		\subcaption[third caption.]{$1-\beta$ (zoomed in)}\label{fig.muP.c}
	\end{minipage}	
	\caption{Impact of transaction costs. Solid vertical line: $\beta={\gamma}/{(\gamma+\delta)}$. Dotted line: approximation line (see Appendix \ref{S_comput})}
\end{figure}

Figure \ref{fig.muP.a} plots the barriers $(a_p^*,a_c^*,b^*)$ when the fixed cost $\chi$ increases from $0.001$ to $0.1$. As we can see, the increase in $\chi$ is compensated primarily by the increase in $b^*$, with almost insignificant drops in both $a_p^*$ and $a_c^*$. This makes sense because with an increased difficulty in paying dividends outside the periodic times, one would simply choose to pay more often at the periodic times. Although not obvious in the figure, one should expect that $a_p^*$ and $b^*$ coincide when $\chi=0$. This corresponds to the special case described in \citet[without fixed transaction costs]{AvTuWo16}.

Figure \ref{fig.muP.b} plots the barriers $(a_p^*,a_c^*,b^*)$ when the proportional cost rate $1-\beta$ increases from $0$ to $0.15$, i.e. across and beyond the threshold ${\delta}/({\gamma+\delta})$. As $1-\beta$ increases from $0$, the two barriers $a_p^*$ and $a_c^*$ split. 
In addition, from Figure \ref{fig.muP.c}, we can see that while $a_p^*$ decreases with a converging behaviour (to $b_0^*$), $a_c^*$ is increasing with a diverging behaviour, as predicted and due to Lemma \ref{Lemma.convergence.beta}. As another illustration of Lemma \ref{Lemma.convergence.beta}, there is a continuity behaviour between $a_p^*$ and the optimal periodic barrier $b_0^*$ at $\beta={\gamma}/{(\gamma+\delta)}$.

\subsubsection{Volatility}
\begin{figure}[htb]\label{fig.muP.volatility}
	\centering
	\begin{minipage}{0.45\textwidth}
		\centering
		\includegraphics[width=1\textwidth]{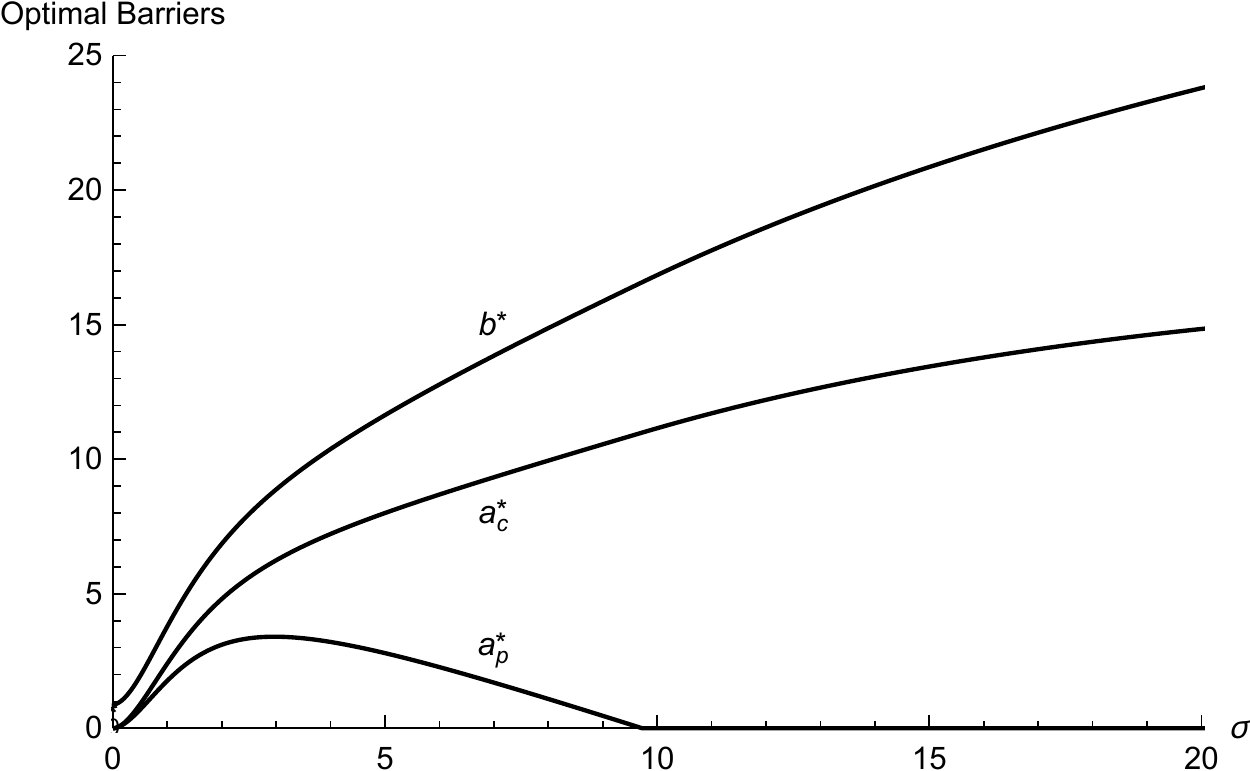}
		\subcaption[fourth caption.]{Small $\sigma$}\label{fig.muP.d}
	\end{minipage}\hspace{0.05\textwidth}
	\begin{minipage}{0.45\textwidth}
		\centering
		\includegraphics[width=1\textwidth]{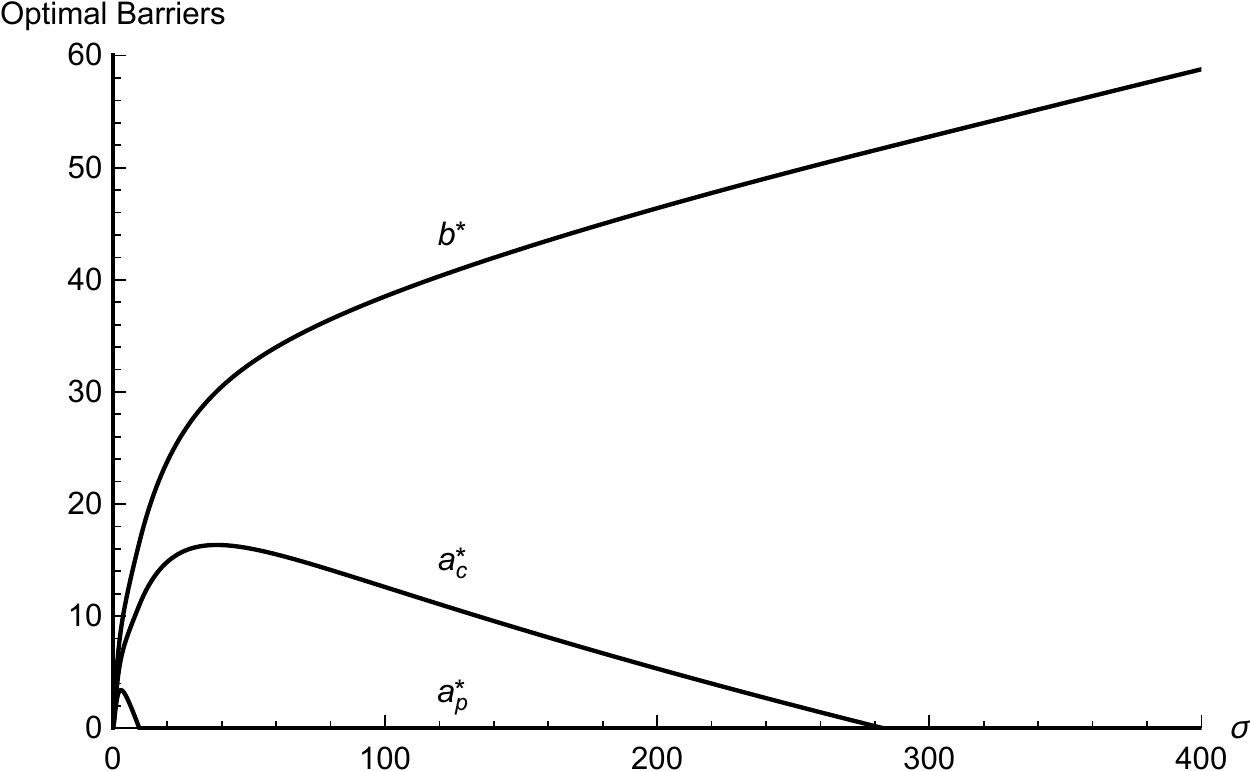}
		\subcaption[fourth caption.]{Large $\sigma$}\label{fig.muP.d2}
	\end{minipage}
	
	\caption{Sensitivities to the volatility parameter $\sigma$.}
\end{figure}

Figure \ref{fig.muP.d} plots the barriers $(a_p^*,a_c^*,b^*)$ when the volatility parameter $\sigma$ increases from $0.01$ to $20$. When the volatility is small, the business is virtually riskless and excess capital is not needed as a buffer. Therefore, both $a_p^*$ and $a_c^*$ are close to zero. When $\sigma$ increases, the business becomes more risky and hence all 3 barriers increase. However, as $\sigma$ further increases beyond a certain level, the business is deemed too risky and early exit would be a better choice. This is reflected by the decrease in the lower barrier $a_p^*$. Furthermore, we can see from Figure \ref{fig.muP.d2} that $a_c^*$ is also going down eventually but the behaviour of $b^*$ is unclear. Heuristically we expect that $b^*\uparrow\infty$ so that the optimal strategy converges to a ``liquidation at first opportunity'' strategy. The idea is that $\sigma\uparrow\infty$ is equivalent  to $\mu,\kappa\downarrow 0$ (after scaling), where a hybrid $(0,0,b)$ strategy is optimal (Lemma \ref{L_conv}). Converting to the original scale, we have $b^*\uparrow\infty$.

\subsubsection{Time parameters}

\begin{figure}[htb]
	\centering
	
	\begin{minipage}{0.45\textwidth}
		\centering
		\includegraphics[width=1\textwidth]{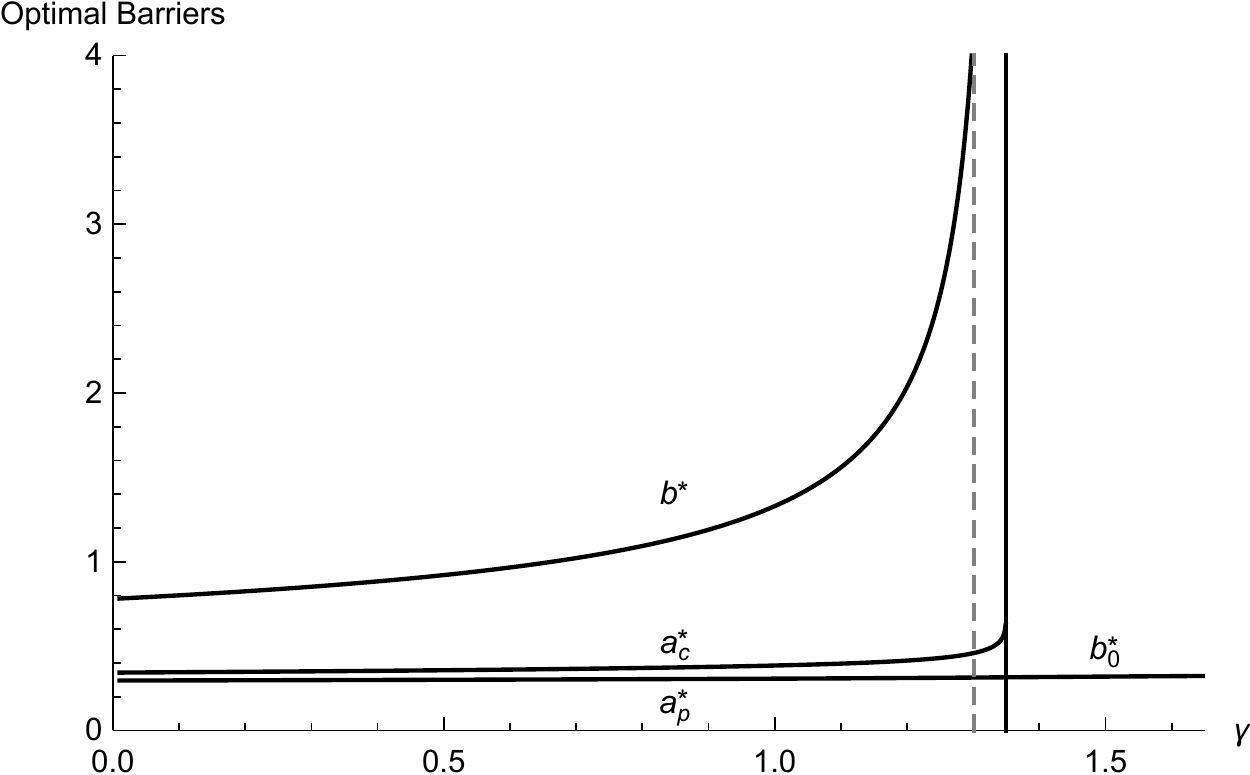}
		\subcaption[fifth caption.]{$\gamma$ }\label{fig.muP.e}
	\end{minipage}\hspace{0.05\textwidth}
	\begin{minipage}{0.45\textwidth}
		\centering
		\includegraphics[width=1\textwidth]{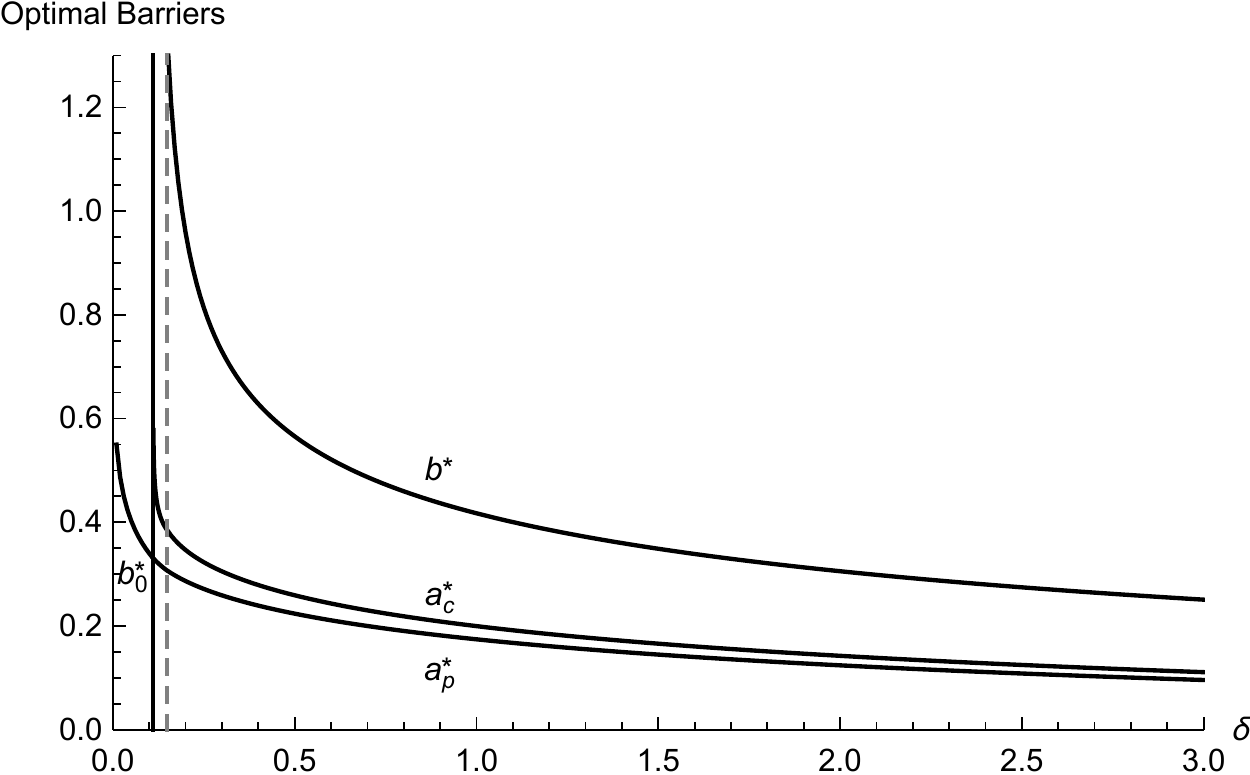}
		\subcaption[sixth caption.]{$\delta$}\label{fig.muP.f}
	\end{minipage}

	\caption{Sensitivities to the time parameters Solid line: $\beta={\gamma}/{(\gamma+\delta)}$. Dotted line: approximation line (see Appendix \ref{S_comput})}
\end{figure}

Figure \ref{fig.muP.e} plots the barriers $(a_p^*,a_c^*,b^*)$ when the dividend frequency parameter $\gamma$ increases from $0.01$ to beyond $1.5$. It is clear that all three barriers are increasing with $\gamma$. This is consistent with the intuition that with more frequent chances to pay periodic dividends (which attract no fixed costs), one does not have the urgency to pay more which puts the company at risk. When $\gamma$ increases to the point that ${\gamma}/{(\gamma+\delta)}$ approaches $\beta$, both $a_c^*$ and $b^*$ should increase to infinity while $a_p^*$ increases to $b_0^*$, which resemble a periodic $b_0^*$ strategy. This behaviour is similar to the change in $1-\beta$ studied in Figure \ref{fig.muP.b}.

Finally, Figure \ref{fig.muP.f} plots the barriers $(a_p^*,a_c^*,b^*)$ when the time preference parameter $\delta$ ranges from $0.01$ to $3$. Unsurprisingly, the effect is qualitatively the reverse of that of $\gamma$ in Figure \ref{fig.muP.e}, with also a smooth connection with the periodic $b_0^*$ strategy.

\subsection{When the business is unprofitable ($\mu<0$)}\label{S.NumericalN}

Our baseline setting includes: scale parameters $(\mu,\sigma,\chi)=(-1,0.3,0.15)$, time parameters $(\gamma,\delta)=(1,0.15)$ and proportional transaction cost parameter $\beta=0.7$. \corr{Except the parameter under consideration, all other parameters will be set to the baseline.} Section \ref{SubS.cost.N} explores the impact of the 2 types of costs to the optimal barriers. Following that, Section \ref{SubS.sens.N} illustrates the sensitivities of other parameters to the optimal barriers.

\subsubsection{Transaction costs}\label{SubS.cost.N}

We start by discussing the impact of the two types of transaction costs (proportional and fixed) on the optimal barriers. Remember that $\beta\le 1$ is the ratio of net dividends of immediate dividends, as compared to periodic (see Remark \ref{R_betacp}). This means that $1-\beta$ is the net level of proportional transaction costs and high levels further penalise the immediate dividends as compared to the periodic ones. Recall as well that immediate liquidation occurs as soon as the surplus level its the area between $b_1$ and $b_2$, and at first opportunity otherwise.

\begin{figure}[htb]
	\centering
	\begin{minipage}{0.45\textwidth}
		\centering
		\includegraphics[width=1\textwidth]{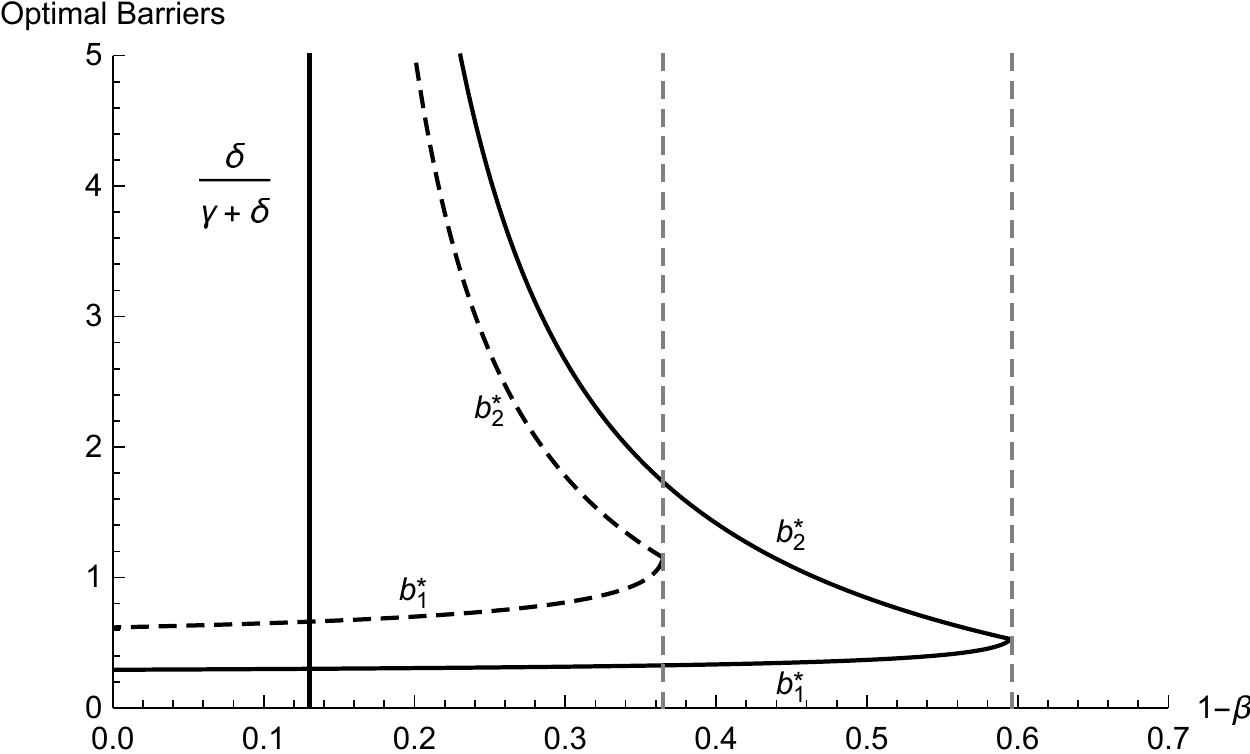}
		\subcaption[first caption.]{Small fixed transaction costs $\chi$ \\ Solid line: $\chi=0.15$; Dashed line: $\chi=0.3$}\label{fig.muN.betaKappaSmall}
	\end{minipage}%
	\hspace{0.08\textwidth}
	\begin{minipage}{0.45\textwidth}
		\centering
		\includegraphics[width=1\textwidth]{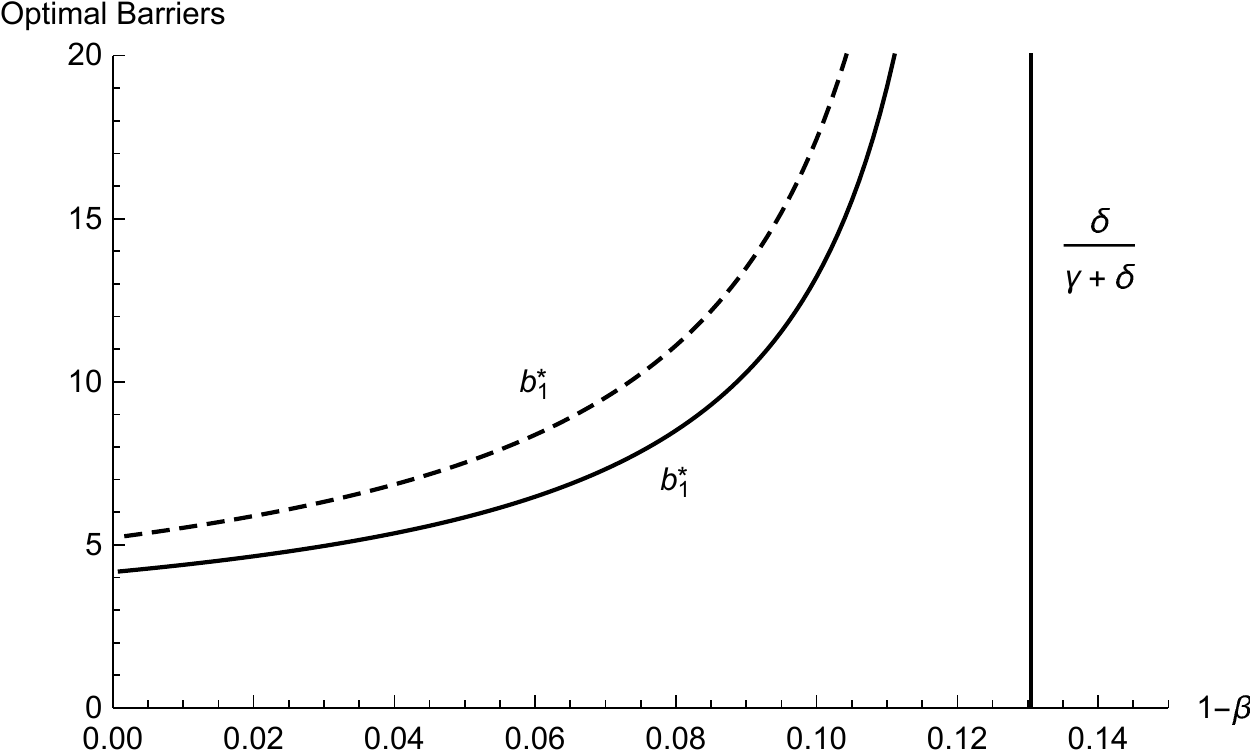}
		\subcaption[second caption.]{Large fixed transaction costs $\chi$ \\ Solid line: $\chi=0.9$; Dashed line: $\chi=1$}\label{fig.muN.betaKappaLarge}
	\end{minipage}%
	
	\caption{Interplay between proportional and fixed transaction costs. An empty region means $\pi_0$ is optimal.} \label{fig.muN.betaKappa}
\end{figure}

Figure \ref{fig.muN.betaKappa} illustrates the change in the optimal barriers $(b_1^*,b_2^*)$ with increasing proportional transaction cost $1-\beta$ and different fixed transaction costs $\chi$. First, when $\chi$ is relatively large, the periodic $0$ strategy is optimal for $1-\beta>{\delta}/{(\gamma+\delta)}$ (high proportional transaction cost), which is evident in Figure \ref{fig.muN.betaKappaLarge}, compared to Figure \ref{fig.muN.betaKappaSmall}. If we imagine the periodic $0$ strategy as a liquidation $(b_1,b_2)$ strategy with both barriers being infinity, then the two graphs are consistent. Hence, we can focus on Figure \ref{fig.muN.betaKappaSmall}.

From Figure \ref{fig.muN.betaKappaSmall}, we can see that when the proportional transaction cost $1-\beta$ decreases, immediate dividends become optimal, and the associated two barriers appear and diverge. The upper barriers increase to infinity when $1-\beta$ approaches ${\delta}/{(\gamma+\delta)}$ from above, and the lower barrier stabilises to a certain level. On the other hand, we can see that the two barriers degenerate to one level which corresponds to the periodic $0$ strategy when the proportional transaction cost $1-\beta$ is large. This continuity feature is quite surprising and remarkable, especially the continuity of the lower barrier $b_1^*$ at $\beta={\gamma}/{(\gamma+\delta)}$.

The collapse of \corr{the area} between the two barriers is quite intuitive as an increase in the cost $1-\beta$ makes the decision to liquidate the company immediately very expensive compared to waiting for the next dividend decision time and liquidate at that first opportunity. Further, when $1-\beta$ is too large, we totally ignore the option to liquidate the company immediate and choose to wait. Similarly, when the fixed cost $\chi$ increases, the option to liquidate the company now becomes less favourable. This is indicated by the smaller area covered by the 2 dotted lines compared to the solid lines, when $\chi$ increases from $0.15$ to $0.3$. Obviously, when $\chi$ increases, we should expect an increase in $b_1^*$, as displayed in Figure \ref{fig.muN.betaKappaSmall}.

\subsubsection{Sensitivities }\label{SubS.sens.N}

\begin{figure}[htb]
	\centering
	\begin{minipage}{0.33\textwidth}
		\centering
		\includegraphics[width=1\textwidth]{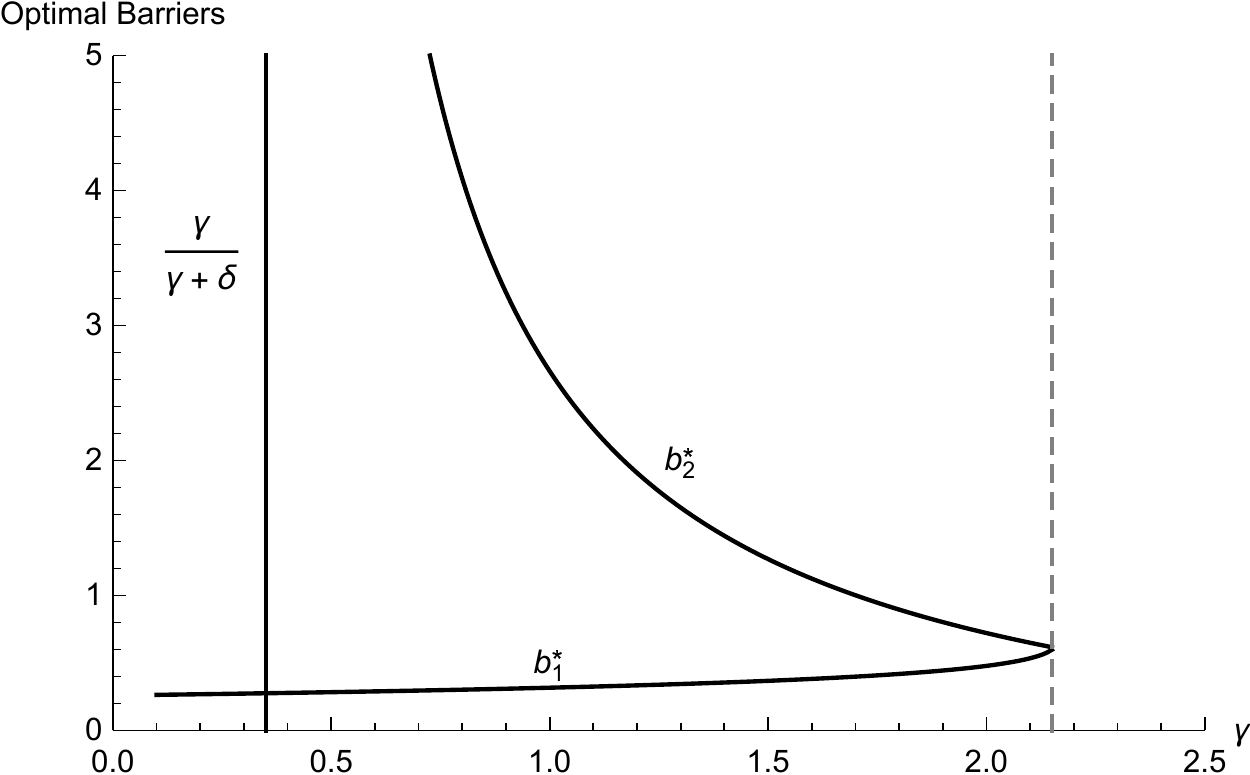}
		\subcaption[first caption.]{$\gamma$}\label{fig.muN.a}
	\end{minipage}%
	\begin{minipage}{0.33\textwidth}
		\centering
		\includegraphics[width=1\textwidth]{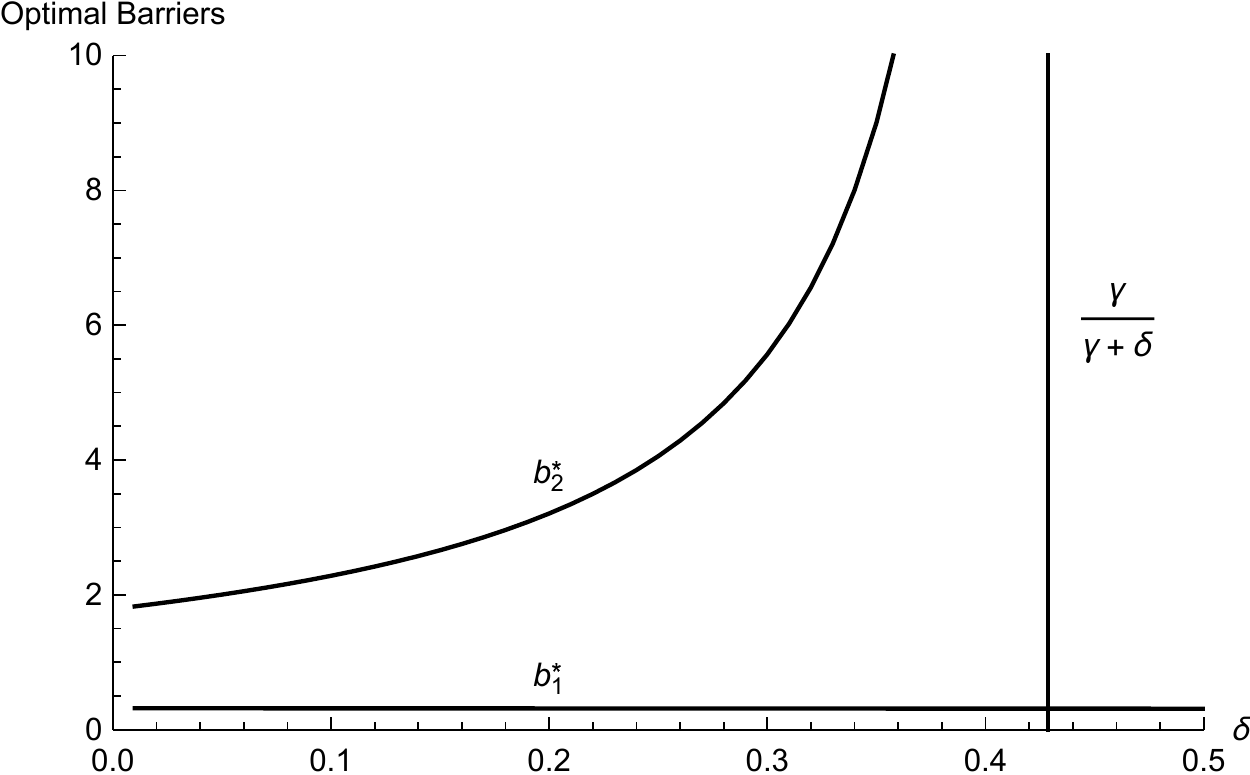}
		\subcaption[second caption.]{$\delta$}\label{fig.muN.b}
	\end{minipage}%
	\begin{minipage}{0.33\textwidth}
		\centering
		\includegraphics[width=1\textwidth]{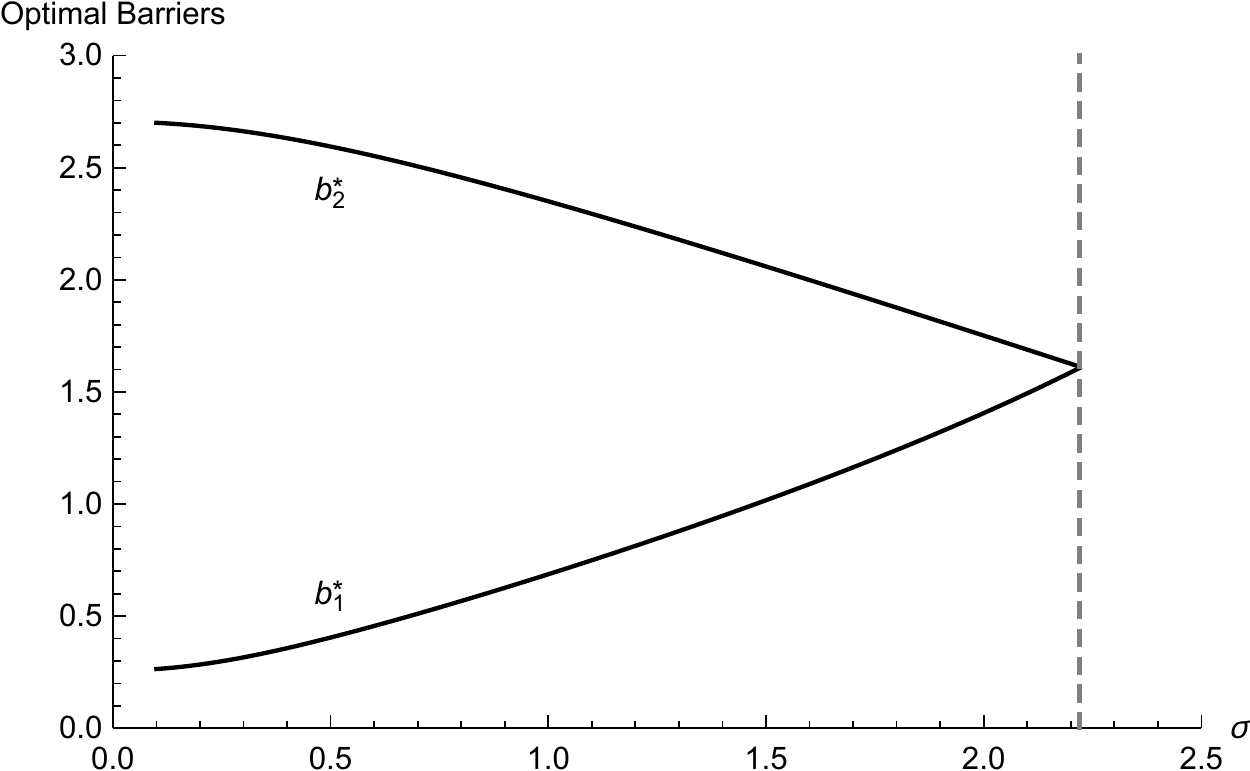}
		\subcaption[third caption.]{$\sigma$}\label{fig.muN.c}
	\end{minipage}
	
	\caption{Sensitivities to parameters. Empty region means $\pi_0$ is optimal.} \label{fig.muN}
\end{figure}

Figure \ref{fig.muN} displays the sensitivities of the barriers to the parameters $\gamma$, $\delta$ and $\sigma$. Note that the (baseline) fixed cost $\chi$ is chosen to be ``small'' to showcase the presence of two barriers. When $\gamma$ increases, the chance of being able to liquidate the company at low cost improves. This favours the option to wait instead of liquidating the company now and is clearly indicated in Figure \ref{fig.muN.a} where the area between the two barriers is shrinking. The opposite effect is present for the impatience parameter $\delta$. Note if we chose a larger base value for $\gamma$, we will see both barriers meet just as in Figure \ref{fig.muN.a}. This is because $\delta$ and $\gamma$ have somewhat inverse roles, and are both functions of how time is defined.

Because the company is non-profitable (negative $\mu$), waiting is speculative because there is nothing left if the company gets ruined before it is liquidated. Figure \ref{fig.muN.c} shows that increased volatility makes such a speculation increasingly worthwhile. When $\sigma$ is low, the lower barrier to be very close to ${\chi}/{\beta}$, which means we will liquidate the company as long as the outcome gives us a positive value, since there is no chance of recovering: indeed the negative drift $\mu<0$ will occur with little chance of being compensated by a positive random diffusion path because $\sigma$ is too low. On the other hand a very high $\sigma$ means it is worth trying one's luck and wait. Note that Figure \ref{fig.muN.c} uses $\beta\leq {\gamma}/{\gamma+\delta}$; the case when $\beta>{\gamma}/{\gamma+\delta}$ is similar, except we do not have $b_2^*$ as it is infinity.

\section{Conclusion}\label{S.conclusion}
In this paper, we considered a diffusion model for the retained cash earnings of a risk business, and studied comprehensively its optimal control via dividends (cash payments) of two different types as observed in real life. Under realistic transaction cost assumptions, we were able to replicate dividend payment behaviour actually observed as optimal. In particular, for realistic ranges of parameters a \emph{hybrid} dividend strategy is optimal, whereby periodic dividends are paid regularly, and extraordinary dividends are paid when the surplus becomes too high. All results summarised in Section \ref{S_map} and Table \ref{T_roadmap} are rigorously shown in the paper and its online supplements.

\section*{Acknowledgments}

This paper was presented at the 23rd International Congress on Insurance: Mathematics and Economics (IME) in July 2019 (Munich, Germany) and at the 54th Actuarial Research Conference (ARC) in August 2019 (Purdue University, USA). The authors are grateful for constructive comments received from colleagues who attended those events\corr{, as well as comments from two anonymous referees, which led to significant improvements of the paper.}

This research was  supported under Australian Research Council's Linkage (LP130100723) and Discovery (DP200101859) Projects funding schemes.  Hayden Lau acknowledges financial support from an Australian Postgraduate Award and supplementary scholarships provided by the UNSW Australia Business School. The views expressed herein are those of the authors and are not necessarily those of the supporting organisations. 

\section*{References}

\bibliographystyle{elsarticle-harv}
\bibliography{libraries}

\newpage

\appendix



\section{Proof of Lemma \ref{Verification.lemma}}\label{A.ver.lemma}

Based on Remark \ref{Remark.Rational}, we can restrict the strategies to have non-negative contribution to the value function. We denote the collection of those strategies $\widetilde{\Pi}$.

By the definition of $v$, it suffices to show that under the hypothesis, we have $H(x)\geq V(x;\pi)$ for all $\pi\in\widetilde{\Pi}$.

We first prove the case when $D^\pi(0)=0$, i.e. there is no dividend at time $0$.

Since $H\in \mathscr{C}^1(\mathbb{R}^+)\cap\mathscr{C}^2(\mathbb{R}^+\backslash E)$, we need to use It\=o-Meyer \citep[e.g. Thm IV.70 in][]{Pro05}, where \citet{Pes05} shows that $H\in\mathscr{C}^1$ is enough to kill the local time at $E$. As a result, we can still apply the It\=o Lemma in its standard form. \corr{Therefore, in the following proof, the local time term will be omitted.}

There is nothing to prove when $x=0$, see \eqref{E_V0}. Hence, we assume $x>0$. For each $n\in\mathbb{N}$, we define a family of increasing stopping time $(T_n,n\in\mathbb{N})$ with $T_n:=\inf\{t> 0: X^\pi(t)> n\text{ or }X^\pi(t)< \frac{1}{n}\}$. By applying the It\=o Lemma to the semi-martingale $\{e^{-\delta(t\wedge T_n)}H(X^\pi(t\wedge T_n));t\geq 0\}$ (with $a\wedge b=\min(a,b)$ for $a,b\in\mathbb{R}$), conditioning on the event $\{X(0)=x\}$, we have
\begin{align*}
&e^{-\delta(t\wedge T_n)}H(X^\pi(t\wedge T_n))-H(x)\\
=~&\int_0^{t\wedge T_n}-\delta e^{-\delta s}H(X^\pi(s-))ds+\int_0^{t\wedge T_n}e^{-\delta s} H'(X^\pi(s-))dX^\pi(s)\\&+\frac{1}{2}\int_0^{t\wedge T_n}e^{-\delta s}H''(X^\pi(s-))1_{\{X^\pi(s-)\notin E\}}d[X^\pi,X^\pi]^c(s)\\&+\sum_{0<s\leq t\wedge T_n}e^{-\delta s}\Big(H(X^\pi(s))-H(X^\pi(s-))-H'(X^\pi(s-))\Delta X^\pi(s)\Big),
\end{align*}
where for a function $F$, $\Delta F(t)=F(t)-F(t-)=F(t)-\lim_{s\uparrow t}F(s)$.

As $X^\pi=X-D^\pi$, $X$ is a diffusion process and $D^\pi$ is a finite variation (FV) process, we have that $d[X^\pi,X^\pi]^c(s)=d[X,X]^c(s)=\sigma^2 ds$. On the other hand, $X$ being a diffusion implies that all the jumps in $X^\pi$ come from $D^\pi$. Therefore, we can rewrite the above as 
{\small\begin{align*}
&e^{-\delta(t\wedge T_n)}H(X^\pi(t\wedge T_n))-H(x)\\
=~&\int_0^{t\wedge T_n}-\delta e^{-\delta s}H(X^\pi(s-))ds+
\int_0^{t\wedge T_n} e^{-\delta s}H'(X^\pi(s-))dX(s)\\&-\int_0^{t\wedge T_n}e^{-\delta s}H'(X^\pi(s-))dD^\pi(s)+\frac{\sigma^2}{2}\int_0^{t\wedge T_n}e^{-\delta s}H''(X^\pi(s-))1_{\{X^\pi(s-)\notin E\}} ds\\&+\sum_{0<s\leq t\wedge T_n}e^{-\delta s}\Big(H(X^\pi(s-)-\Delta D^\pi(s))-H(X^\pi(s-))+H'(X^\pi(s-))\Delta D^\pi(s)\Big)\\
=~&\int_0^{t\wedge T_n}e^{-\delta s}(\mathscr{A}-\delta)H(X^\pi(s-))1_{\{X^\pi(s-)\notin E\}}ds+
\int_0^{t\wedge T_n} e^{-\delta s}H'(X^\pi(s-))\sigma dW(s)\\&+\sum_{0<s\leq t\wedge T_n}e^{-\delta s}\Big(H(X^\pi(s-)-\Delta D^\pi(s))-H(X^\pi(s-))\Big),
\end{align*}}
where $W=\{W(t);t\geq 0\}$ is a standard Brownian motion. We now decompose $D^\pi$ into $D^\pi_p(t)=\int_0^tdD^\pi_p(s)dN_\gamma(s)$ and $D^\pi_c(t)$, where we denote the jump times of $D^\pi_c$ as $\Tc$. In this sense, after some algebraic effort, we can rewrite the above as 
\begin{align*}
&e^{-\delta(t\wedge T_n)}H(X^\pi(t\wedge T_n))-H(x)\\
=~&\int_0^{t\wedge T_n}e^{-\delta s}\Big((\mathscr{A}-\delta)H(X^\pi(s-))+\gamma\Big(\Delta D^\pi_p(s)+H(X^\pi(s-)-\Delta D^\pi_p(s))-H(X^\pi(s-))\Big)\Big)1_{\{X^\pi(s-)\notin E\}}ds\\&+\int_0^{t\wedge T_n}e^{-\delta s}\Big(\Delta D^\pi_p(s)+H(X^\pi(s-)-\Delta D^\pi_p(s))-H(X^\pi(s-))\Big)(dN_\gamma(s)-\gamma ds)\\
&+\sum_{s\in(0, t\wedge T_n]\cap \Tc}e^{-\delta s}\Big(\beta \Delta D^\pi_c(s)-\chi+H(X^\pi(s-)-\Delta D^\pi_c(s))-H(X^\pi(s-))\Big)\\&+
\int_0^{t\wedge T_n} e^{-\delta s}H'(X^\pi(s-))\sigma dW(s)\\
&-\Big(\int_0^{t\wedge T_n}e^{-\delta s}\Delta D^\pi_p(s)dN_\gamma(s)+\sum_{s\in(0, t\wedge T_n]\cap \Tc}e^{-\delta s}(\beta \Delta D^\pi_c(s)-\chi)\Big).
\end{align*}
By denoting
{\small$$M(t):=\int_0^{t}e^{-\delta s}\Big(\Delta D^\pi_p(s)+H(X^\pi(s-)-\Delta D^\pi_p(s))-H(X^\pi(s-))\Big)(dN_\gamma(s)-\gamma ds)+\int_0^{t} e^{-\delta s}H'(X^\pi(s-))\sigma dW(s),$$}
we can rewrite the above as
{\small\begin{align*}
&H(x)\\=~&e^{-\delta(t\wedge T_n)}H(X^\pi(t\wedge T_n))\\
&-\int_0^{t\wedge T_n}e^{-\delta s}\Big((\mathscr{A}-\delta)H(X^\pi(s-))+\gamma\Big(\Delta D^\pi_p(s)+H(X^\pi(s-)-\Delta D^\pi_p(s))-H(X^\pi(s-))\Big)\Big)1_{\{X^\pi(s-)\notin E\}}ds\\
&-\sum_{s\in(0, t\wedge T_n]\cap \Tc}e^{-\delta s}\Big(\beta\Delta D^\pi_c(s)-\chi+H(X^\pi(s-)-\Delta D^\pi_c(s))-H(X^\pi(s-))\Big)\\
&+\Big(\int_0^{t\wedge T_n}e^{-\delta s}\Delta D^\pi_p(s)dN_\gamma(s)+\sum_{s\in(0, t\wedge T_n]\cap \Tc}e^{-\delta s}(\beta\Delta D^\pi_c(s)-\chi)\Big)-M(t\wedge T_n).
\end{align*}}
Now, by hypothesis (Conditions 1,2,4,5), the first 3 lines on the right hand side of the equation are non-negative, which implies 
\begin{equation*}
H(x)\geq \Big(\int_0^{t\wedge T_n}e^{-\delta s}\Delta D^\pi_p(s)dN_\gamma(s)+\sum_{s\in(0, t\wedge T_n]\cap \Tc}e^{-\delta s}(\beta\Delta D^\pi_c(s)-\chi)\Big)-M(t\wedge T_n)
\end{equation*}
Note that $M$ is a zero-mean martingale as all the terms in the inegral are finite, by hypothesis (Condition 3). Hence, by taking expectation, we have
\begin{equation*}
H(x)\geq \Ex\Big(\int_0^{t\wedge T_n}e^{-\delta s}\Delta D^\pi_p(s)dN_\gamma(s)+\sum_{s\in(0, t\wedge T_n]\cap \Tc}e^{-\delta s}(\beta\Delta D^\pi_c(s)-\chi)\Big).
\end{equation*}
Finally, we observe that $T_n\rightarrow \tau^\pi$ a.s. and the terms inside the expectation are non-negative. Hence, by applying Fatou's Lemma, we get
\begin{align*}
H(x)\geq~&\liminf_{t,n\uparrow \infty} \Ex\Big(\int_0^{t\wedge T_n}e^{-\delta s}\Delta D^\pi_p(s)dN_\gamma(s)+\sum_{s\in(0, t\wedge T_n]\cap \Tc}e^{-\delta s}(\beta\Delta D^\pi_c(s)-\chi)\Big)\\
\geq ~&\Ex\Big(\liminf_{t,n\uparrow \infty}\Big(\int_0^{t\wedge T_n}e^{-\delta s}\Delta D^\pi_p(s)dN_\gamma(s)+\sum_{s\in(0, t\wedge T_n]\cap \Tc}e^{-\delta s}(\beta\Delta D^\pi_c(s)-\chi)\Big)\Big)\\
=~&\Ex\Big(\int_0^{\tau^\pi}e^{-\delta s}\Delta D^\pi_p(s)dN_\gamma(s)+\sum_{s\in(0,\tau^\pi]\cap\Tc}e^{-\delta s}(\beta\Delta D^\pi_c(s)-\chi)\Big)\\
=~&V(x;\pi),
\end{align*}
which completes the proof for strategies $\pi\in\widetilde{\Pi}$ such that $D^\pi(0)=0$. For strategies $\pi\in\widetilde{\Pi}$ such that $D^\pi(0)>0$, we denote $\widetilde{\pi}$ the same strategy for $t>0$, i.e. $D^{\widetilde{\pi}}(t)=\{D^\pi_p(t),D^\pi_c(t)-D^\pi(0)\}$. Then we have
\begin{equation*}
V(x;\pi)=\Ex(\beta D^\pi(0)-\chi+V(x-D^\pi(0);\widetilde{\pi}))\leq\Ex( \sup_{\xi\in(0,x]}\Big(\beta \xi-\chi+H(x-\xi)\Big))\leq H(x)
\end{equation*} 
by an application of the previous result for $\widetilde{\pi}$ and Condition 5.

\section{Proof of Proposition \ref{prop.1}} \label{A_Prop55}

Note Existence is established in Proposition \ref{prop.2} and here we assume $(a^*,l^*,y^*)$ exists.
	
	It is clear from \eqref{obj.fcn} that the objective function is differentiable w.r.t. $(a,l,y)$. Therefore, being optimal implies the partial derivatives are zero (except at the boundary). It is straight-forward to show 
	$$\parD{y}\frac{V(a)}{\Wq(a)}=0 \iff V'(b-)=\beta.$$
	From this, we see that at $(a^*,l^*,y^*)$
	\begin{align*}
	&0=\parD{y}\Big(V(a_c)-\beta a_c\Big)=\parD{y}\Big(\frac{V(a)}{\Wq(a)}G(a,l)\Big)=\Big(\parD{y}\frac{V(a)}{\Wq(a)}\Big)G(a,l)\\
	\implies ~&\parD{y}\frac{V(a)}{\Wq(a)}=0\iff V'(b-)=\beta.
	\end{align*}
	A further calculation (to appear later in \eqref{PC.geq.0.min}) shows that it is never optimal to have $y={\chi}/{\beta}$ (i.e. at the boundary) so the equality above always hold.
	
	Now, using 
	$$V(a_c)+\beta y-\chi=V(b)=\frac{V(a)}{\Wq(a)}G(a,y+l)-\gamma\WqrBB(y+l),$$
	we get
	\begin{align*}
	\parD{l}\Big(V(a_c)-\beta a_c\Big)
	=~&\parD{l}\Big(\frac{V(a)}{\Wq(a)}G(a,y+l)-\gamma\WqrBB(y+l)\Big)-\beta\\
	=~&\Big(\parD{l}\frac{V(a)}{\Wq(a)}\Big)G(a,y+l)+V'(b)-\beta\\
	=~&\Big(\parD{l}\frac{V(a)}{\Wq(a)}\Big)G(a,y+l)
	\end{align*}
	so we have
	\begin{align*}
	\parD{l}\Big(V(a_c)-\beta a_c\Big)=\parD{l}V(a_c)-\beta
	=~&\parD{l}\Big(\frac{V(a)}{\Wq(a)}G(a,l)-\gamma\WqrBB(l)\Big)-\beta\\
	=~&\Big(\parD{l}\frac{V(a)}{\Wq(a)}\Big)G(a,l)+\frac{V(a)}{\Wq(a)}\parD{l}G(a,l)-\gamma\WqrB(l)-\beta\\
	=~&\Big(\parD{l}\frac{V(a)}{\Wq(a)}\Big)G(a,l)+V'(a_c)-\beta\\
	\implies \Big(\parD{l}\frac{V(a)}{\Wq(a)}\Big)\Big(G(a,y+l)-G(a,l)\Big)=~&V'(a_c)-\beta
	\end{align*}
	so we have 
	$$V'(a_c)=\beta.$$
	if $l^*>0$. Otherwise, if $l^*=0$ (i.e. at the boundary), we see that $l\mapsto V(a_c)-\beta a_c$ is decreasing in $l$ near zero, i.e. $$0\geq\parD{l}\Big(V(a_c)-\beta a_c\Big)=\Big(\parD{l}\frac{V(a)}{\Wq(a)}\Big)G(a,y+l)$$ and therefore by noting $G(a,y+l)>G(a,l)$ we get
	$$V'(a_c)\leq \beta.$$
	
	Similarly, we have
	\begin{align}
	\parD{a}\Big(V(a_c)-\beta a_c\Big)=~&\Big(V'(b)-\beta\Big)+G(a,y+l)\parD{a}\frac{V(a)}{\Wq(a)}+\gamma\WqrB(y+l)\Big(1-V'(a)\Big),\label{eq1}\\
	\parD{a}\Big(V(a_c)-\beta a_c\Big)=~&\Big(V'(a_c)-\beta\Big)+G(a,l)\parD{a}\frac{V(a)}{\Wq(a)}+\gamma\WqrB(l)\Big(1-V'(a)\Big).\label{eq2}
	\end{align}
	Now, Assumption \ref{Ass0} implies for $l>0$
	\begin{equation}\label{use.ass0}
	 \Delta:=G(a,y+l)\gamma\WqrB(l)- G(a,l)\gamma\WqrB(y+l)<0,
	\end{equation}
	which also holds for $l=0$ as the first term is null and the second term is positive.
	Hence, we can eliminate the term with $\parD{a}(V(a)/\Wq(a))$ to get
	\begin{equation}\label{eq.con}
	\Big(G(a,y+l)-G(a,l)\Big)\parD{a}\Big(V(a_c)-\beta a_c\Big)=G(a,y+l)\Big(V'(a_c)-\beta\Big)+|\Delta| (V'(a)-1).
	\end{equation}
	
	Now, suppose $a^*=\bar{a}$, then we have
	$$\frac{V(\bar{a})}{\Wq(\bar{a})}\leq \frac{\frac{\Wq(\bar{a})}{\Wq'(\bar{a})}}{\Wq(\bar{a})}=\frac{1}{\Wq'(\bar{a})}\implies V'(a)=V'(\bar{a})=\frac{V(\bar{a})}{\Wq(\bar{a})}\Wq(\bar{a})\leq 1,$$
	because the value function (of our strategy) is smaller in the current setting than the setting when there is no transaction costs (e.g. in \cite{Loe08}) and the optimal value function at $\bar{a}$ is given above. Hence the right hand side of the above equation is negative and so as the left hand side. This means it is impossible for $a^*=\bar{a}$ to be a maximiser for $V(a_c)-\beta a_c$. On the other hand, it is possible for $a^*=0$. In that case, we have $\parD{a}(V(a_c)-\beta a_c)\leq 0$. If furthermore $l^*>0$, we have $V'(a_c)=\beta $ and therefore we can conclude $V'(a)\leq 1$. 
	
	Suppose $a^*>0$, i.e. $\parD{a}(V(a_c)-\beta a_c)=0$. If $l^*=0$, we have $V'(a)=V'(a_c)\leq \beta$ which is a contradiction in view of \eqref{eq.con}. Therefore, we must have $l^*>0$ and therefore we can conclude $V'(a)=1$.
	
	This completes the proof.
	
\section{Proof of Proposition \ref{Prop.3}}\label{A.1}

Note $\gamma\WqrB(x)=kJ(x)$ for some positive constant $k$, and recall from equation \eqref{Vac} and the definition of $J$ that
\begin{align}
G(a,x)=~&g(x)\frac{f'(a)-s_1\frac{\delta}{\gamma+\delta}f(a)}{r_1-s_1}+e^{s_1x}\frac{\delta}{\gamma+\delta}f(a)+\frac{\gamma}{\gamma+\delta}f(a),\\
J(x)=~&g(x)(-s_1)+e^{s_1}(r_1-s_1)+(-(r_1-s_1)),
\end{align}
we want to show that (for any $a\geq 0$)
\begin{equation}
J(x)\parD{x}G(a,x)-G(a,x)\parD{x}J(x)<0,\quad x>0.
\end{equation}
By direct computation, we see that
\begin{align}
&J(x)\parD{x}G(a,x)-G(a,x)\parD{x}J(x)\nonumber\\
=~&-f'(a)\Big(g'(x)-(r_1-s_1)e^{(r_1+s_1)x}\Big)\label{eq.a}
\\&+s_1f(a)r_1g(x)\nonumber.
\end{align}
Denote the function $F_1$ with
$$F_1(x):=g'(x)-(r_1-s_1)e^{(r_1+s_1)x}.$$
It is easy to see $F_1(0)=0$ and $F_1'(x)>0$ for $x>0$ and hence we have $F_1(x)>0$ for $x>0$.

In view of the above equality
$$J(x)\parD{x}G(a,x)-G(a,x)\parD{x}J(x)=-f'(a)F_1(x)+s_1f(a)r_1g(x),$$
we can conclude that 
$$J(x)\parD{x}G(a,x)-G(a,x)\parD{x}J(x)<0.$$

This completes the proof.

\section{Proof of Proposition \ref{prop.2}}\label{A.prop2}

The statement is a directly consequence of Propositions \ref{Prop.3} and \ref{prop.1}. Therefore, we are left to show the hypothesis in Proposition \ref{prop.1}.

Using the formulas in Proposition \ref{Prop.Vfcn}, we have 
\begin{align}
(r_1-s_1)A=~&\frac{\alpha(r_1-s_1)(y-\frac{\chi}{\alpha})(f'(a)-s_1\frac{\delta}{\gamma+\delta}f(a))+\frac{\gamma}{\gamma+\delta}(J(d,l)+s_1g(d,l)))}{\frac{\delta}{\gamma+\delta}f(a)J(d,l)+f'(a)g(d,l)},\\
\parD{l}A=~&-A\frac{\frac{\delta}{\gamma+\delta}f(a)J'(d,l)+f'(a)g'(d,l)}{\frac{\delta}{\gamma+\delta}f(a)J(d,l)+f'(a)g(d,l)},\\ (\text{with }J'(d,l)=~&J'(d)-J'(l),~g'(d,l)=g'(d)-g'(l))\\
\lim_{l\rightarrow \infty} C=~&\frac{\frac{\gamma\mu}{(\gamma+\delta)^2}-\frac{\gamma}{\gamma+\delta}\frac{1}{s_1}}{\frac{\delta}{\gamma+\delta}f(a)-\frac{f'(a)}{s_1}}<\infty,\\
\lim_{l\rightarrow\infty}\Big(g(l)(1+s_1\frac{g(d,l)}{J(d,l)})\Big)=~&\lim_{l\rightarrow\infty}\frac{(r_1-s_1)g(l)(e^{s_1d}-e^{s_1l})}{J(d,l)}=0,\\
\lim_{l\rightarrow\infty}Ag(l)=~&\frac{\alpha}{e^{r_1y}-1}\big(y-\frac{\chi}{\alpha}\big)<\infty,\\
\lim_{l\rightarrow\infty}\frac{g(d,l)}{J(d,l)}=~&\frac{-1}{s_1},\quad 
\lim_{l\rightarrow\infty}\frac{g'(d,l)}{J(d,l)}=\frac{-r_1}{s_1},\quad 
\lim_{l\rightarrow\infty}\frac{J'(d,l)}{J(d,l)}=r_1,\quad
\lim_{l\rightarrow\infty}\parD{l}C=0.
\end{align}
This implies
\begin{align*}
\lim_{l\rightarrow\infty}\parD{l}(Ag(l))=~&\lim_{l\rightarrow\infty}\Big(Ag'(l)+g(l)\parD{l}A\Big)\\=~&\lim_{l\rightarrow\infty}Ag(l)\Big(\lim_{l\rightarrow\infty}\frac{g'(l)}{g(l)}-\lim_{l\rightarrow\infty}\frac{\frac{\delta}{\gamma+\delta}f(a)J'(d,l)+f'(a)g'(d,l)}{\frac{\delta}{\gamma+\delta}f(a)J(d,l)+f'(a)g(d,l)}\Big)\\
=~&\lim_{l\rightarrow\infty}Ag(l)(r_1-r_1)=0.
\end{align*}

Therefore, we have
\begin{align*}
\lim_{l\rightarrow\infty}\parD{l}	(V(a_c)-\frac{\gamma}{\gamma+\delta}l)=~&\lim_{l\rightarrow\infty}\parD{l}(Ag(l))=0
\end{align*}
and hence
$$\lim_{l\rightarrow\infty}\parD{l}	(V(a_c)-\beta l)=\frac{\gamma}{\gamma+\delta}-\beta<0.$$
From this, we see that $V(a_c)-\beta a_c$ is decreasing for large enough $l$ (independent of $(a,y)$), say $\bar{l}$, i.e. $l\mapsto (V(a_c)-\beta a_c)$ cannot attain its local maximum for $l>\bar{l}$.
Since we have already chosen $a\in[0,\bar{a}]$, we do not worry about the $a$ dimension. 

On the other hand, we have
\begin{align}
V(a_c)=~&C\Bigg(f(a)\Big(\frac{\gamma}{\gamma+\delta}+\frac{\delta}{\gamma+\delta}(e^{s_1l}-s_1\frac{g(l)}{r_1-s_1})\Big)+f'(a)\frac{g(l)}{r_1-s_1}\Bigg)\nonumber\\
&+\frac{g(l)}{r_1-s_1}(s_1\frac{\gamma\mu}{(\gamma+\delta)^2}-\frac{\gamma}{\gamma+\delta})-e^{s_1l}\frac{\gamma\mu}{(\gamma+\delta)^2}+\frac{\gamma}{\gamma+\delta}(l+\frac{\mu}{\gamma+\delta}).\label{Vac}
\end{align}
\begin{remark}
	It is easy to see that $C=V(a)/\Wq(a)$, the terms inside the bracket after $C$ is $G(a,l)$, and the terms in the second line correspond to $\gamma\WqrBB(l)$.
\end{remark}

We are left to work with the $y$ dimension.	
From \eqref{Vac}, it is clear that in terms of $y$, the objective function $V(a_c)-\beta a_c$ solely depends on $C$, as we have also discovered before using scale functions. There is no shortcut but to compute the derivative w.r.t. $y$. From 
$$
C=\frac{(r_1-s_1)\big((\beta-\frac{\gamma}{\gamma+\delta})(d-l)-\chi\big)+\frac{\gamma}{\gamma+\delta}g(d,l)+\frac{\gamma\mu}{(\gamma+\delta)^2}J(d,l)}{\frac{\delta}{\gamma+\delta}f(a)J(d,l)+f'(a)g(d,l)},
$$
we get (after some tedious algebric operations)
\begin{align*}
&\parD{y}C\times (\frac{\delta}{\gamma+\delta}f(a)J(d,l)+f'(a)g(d,l))^2\\
=~&\alpha(r_1-s_1)\Big(\frac{\delta}{\gamma+\delta}f(a)J(d,l)+f'(a)g(d,l)-(y-\frac{\chi}{\alpha})(\frac{\delta}{\gamma+\delta}f(a)J'(d)+f'(a)g'(d))\Big)\\
&+\frac{\gamma}{\gamma+\delta}f'(a)\Big(\frac{\mu}{\gamma+\delta}-\frac{\frac{\delta}{\gamma+\delta}f(a)}{f'(a)}\Big)\Big(g(d,l)J'(d)-J(d,l)g'(d)\Big),
\end{align*}	
where 
$$\frac{\mu}{\gamma+\delta}-\frac{\frac{\delta}{\gamma+\delta}f(a)}{f'(a)}	\geq 0$$
for $a\in[0,\bar{a}]$ and it can be checked by taking derivative w.r.t. $y$ that
$$	g(d,l)J'(d)-J(d,l)g'(d)=g'(d)e^{s_1l}-s_1g(l)e^{s_1d}-(r_1-s_1)e^{(r_1+s_1)d}> 0 $$
for $y\geq {\chi}/{\alpha}$. This implies 
\begin{equation}\label{PC.geq.0.min}
\parD{y}\frac{V(a)}{\Wq(a)}\Big|_{y=\chi/\alpha}=\parD{C}\Big|_{y=\chi/\alpha}>0,\quad (a,l)\in[0,\bar{a}]\times[0,\bar{l}].
\end{equation}

Next, we take the limit $y\rightarrow\infty$, and see (after some algebraic operations) that
\begin{align*}
&\parD{y}C\times (\frac{\delta}{\gamma+\delta}f(a)J(d,l)+f'(a)g(d,l))^2\\
=~&f'(a)
\alpha(r_1-s_1)\Big(g(d,l)-(y-\frac{\chi}{\alpha})g'(d)+\frac{1}{\alpha(r_1-s_1)}\frac{\gamma\mu}{(\gamma+\delta)^2}\big(g'(d)e^{s_1l}-s_1g(l)e^{s_1d}-(r_1-s_1)e^{(r_1+s_1)d}\big)\Big)\\
&+\frac{\delta}{\gamma+\delta}f(a)\alpha(r_1-s_1)\Big(J(d,l)-(y-\frac{\chi}{\alpha})J'(d)+\frac{\gamma}{\gamma+\delta}\big(g'(d)e^{s_1l}-s_1g(l)e^{s_1d}-(r_1-s_1)e^{(r_1+s_1)d}\big)\Big)\\
\leq~&f'(a)
\alpha(r_1-s_1)\Big(g(d)-(y-\frac{\chi}{\alpha}-\frac{1}{\alpha(r_1-s_1)}\frac{\gamma\mu}{(\gamma+\delta)^2})g'(d)\Big)\\
&+\frac{\delta}{\gamma+\delta}f(a)\alpha(r_1-s_1)\Big(J(d)-(y-\frac{\chi}{\alpha})J'(d)+\frac{\gamma}{\gamma+\delta}g'(d)\Big)
\end{align*}
which drifts to $-\infty$ when $y\rightarrow\infty$. Hence, we can choose $\underline{d}$ such that $d>\underline{d}$ implies $\parD{y}C$ is decreasing for all $(a,l)\in[0,\bar{a}]\times[0,\bar{l}]$. In particular, we can choose $\bar{y}=(\bar{l}+\underline{d})\vee 2\chi/\alpha$ such that the same holds for $y\geq \bar{y}$.
	
To conclude, we have find a box for $\mathscr{B}:=[0,\bar{a}]\times[0,\bar{l}]\times[\chi/\alpha,\bar{y}]$ for $(a,l,y)$ such that
\begin{enumerate}
	\item The objective function $V(a_c)-\beta a_c$ attains its maximum inside $\mathscr{B}$,
	\item Its maximum $(a^*,l^*,y^*)$ either occurs in the interior of $\mathscr{B}$, or we have $a^*=0$ or $l^*=0$ or both, but not other cases.
\end{enumerate}	

This concludes the hypothesis in Proposition \ref{prop.1} and hence completes the proof.

\section{Proof of Lemma \ref{Lemma.VD.opt} }\label{A.Lemma6.1}	
	
	Denote $\widetilde{A}=A$ and $\widetilde{B}=B-A$ so that the derivative of the value function on $[a_p,b]$ is 
	\begin{equation}\label{eqt.VDnice.mid}
	V'(a_p+x)=\widetilde{A}r_1e^{r_1x}+\widetilde{B}s_1e^{s_1 x}+\frac{\gamma}{\gamma+\delta},~x\in[0,d].
	\end{equation}
	From $V'(b)=\beta$, we have 
	\begin{equation*}
	V'(b)=\widetilde{A}r_1e^{r_1d}+\widetilde{B}s_1e^{s_1d}+\frac{\gamma}{\gamma+\delta}=\beta,
	\end{equation*}
	or
	\begin{equation}\label{eqt.vbD.eq1}
	\widetilde{A}r_1e^{r_1d}+\widetilde{B}s_1e^{s_1d}=\alpha.
	\end{equation}
	Moreover, we have
	\begin{equation}\label{eqt.vDD}
	V''(a_p+x)=\widetilde{A}r_1^2e^{r_1x}+\widetilde{B}s_1^2e^{s_1x}
	\end{equation}
	and
	\begin{equation}\label{eqt.vDDD}
	V'''(a_p+x)=\widetilde{A}r_1^3e^{r_1x}+\widetilde{B}s_1^3e^{s_1x}.
	\end{equation}
	
	We first show that $\widetilde{A}>0$ by contradiction. Suppose $\widetilde{A}\leq0$ and $\widetilde{B}\geq0$, then the L.H.S. of (\ref{eqt.vbD.eq1}) is negative, which is impossible. On the other hand, if we assume $\widetilde{A}\leq0$ and $\widetilde{B}<0$, then we have from (\ref{eqt.vDD}) $V''< 0$, which implies that $V'$ is decreasing on $[a_p,b]$. However, from $V\in\mathscr{C}^1(\mathbb{R}_+)$, we have
	\begin{equation*}
	\int_{a_c}^{b-}(\beta-{V'(x)})dx=\chi>0,
	\end{equation*}
	which implies that $V'\geq\beta$ on $[a_c,b]$, which is also impossible. 
	
	Now we have established $\widetilde{A}>0$. If we further assume $\widetilde{B}\geq0$, then we have from (\ref{eqt.vDD}) $V''\geq0$, which implies that $V'$ is increasing on $[a_p,b]$. Note that this would not be possible unless $V'(0)<\beta\implies a_c=a_p=0$ because otherwise we have $V'(a_c)=\beta=V'(b)$. Regardless, as $V'$ increases to $V'(b)=\beta$, we have $V'<\beta$ on $[0,b)$ and $V'\equiv\beta$ on $[b,\infty)$, which also holds if $b=0$. Furthermore, the fact that $V$ is positive (by the definition of the value function) implies that $V'(0)>0$, which in turn implies that $V'>0$ on $[0,\infty)$.
	
	For the last case, $\widetilde{A}>0$ and $\widetilde{B}<0$, we can deduce from (\ref{eqt.vDDD}) that $V'''\geq0$ on $[a_p,b]$, or equivalently, $V'$ is convex on $[a_p,b]$. This together with $V'(a_c)\leq \beta=V'(b)$ gives $V'\leq\beta$ on $[a_c,b]$. This fact combining with $V'(a_p)\leq 1$ shows that $V'$ is decreasing from $a_p$ to $a_c$, then further decreasing and finally increasing to $\beta$ at $b$, as $V'$ is convex on $[a_p,b]$, or simply increasing to $\beta$ if $a_p=a_c=0$ and $V'(0)<\beta$. Furthermore, in view of (\ref{eqt.VDnice.mid}), we have that $V'>0$ on $[0,\infty)$.
	
	\begin{remark}\label{Remark.VDDatB}
		Note in any cases, we have $V''(b-)>0$ as $V'$ is increasing at $b-\varepsilon$ for all small enough $\varepsilon>0$. On the other hand, we have $V''(b+)=0$.
	\end{remark}

\section{Proof of Lemma \ref{Lemma.VD.Nice} }\label{A.Lemma6.2}		
	
		For $x> b$, we have
	\begin{align*}
	&(\mathscr{A}-\delta)V(x)+\gamma \Big(x-a_p+V(a_p)-V(x)\Big)\\
	=~&(\mathscr{A}-\delta)V(b+)+\gamma \Big(b-a_p+V(a_p)-V(b)\Big)-(\gamma+\delta)(V(x)-V(b))+\gamma(x-b)\\
	=~&(\mathscr{A}-\delta)V(b+)+\gamma \Big(b-a_p+V(a_p)-V(b)\Big)-(\gamma+\delta)\Big(\beta (x-b)-\frac{\gamma}{\gamma+\delta}(x-b)\Big)\\
	\leq ~&(\mathscr{A}-\delta)V(b+)+\gamma \Big(b-a_p+V(a_p)-V(b)\Big)
	\end{align*}
	as $\beta>{\gamma}/({\gamma+\delta})$ by the assumption in \eqref{Ass.betaLarge}.
	
	Together with Remark \ref{Remark.VDDatB} and the fact that $V\in\mathscr{C}^1(\mathbb{R}_+)$, we have
	\begin{align}
	0=~&(\mathscr{A}-\delta)V(b-)+\gamma \Big(x-a_p+V(a_p)-V(b)\Big)>(\mathscr{A}-\delta)V(b+)+\gamma \Big(x-a_p+V(a_p)-V(b)\Big)\nonumber\\
	\geq ~&(\mathscr{A}-\delta)V(x)+\gamma \Big(x-a_p+V(a_p)-V(x)\Big)\label{eq.HJB1.bGeq0}
	\end{align}
	for $x>b$. Now, denote 
	\begin{equation}
	H_1(\xi):=\xi+V(x-\xi)-V(x)
	\end{equation}
	and by taking derivative with respect to $\xi$, we have 
	$$H_1^\prime(\xi)=1-V'(x-\xi).$$
	In view of Lemma \ref{Lemma.VD.opt}, $V'(x-\xi)<1$ is equivalent to $x-\xi> a_p$, or equivalently $\xi< x-a_p$. Therefore, we can deduce that $H_1$ is increasing on $[0,a_p]$ if $a_p>0$ then decreasing on $(a_p,\infty)$, which implies that in any case it attains its maximum at $\xi=x-a_p$. Therefore, we have
	\begin{align*}
	&(\mathscr{A}-\delta)V(x)+\gamma \sup_{\xi\in[0,x]}\Big(\xi+V(x-\xi)-V(x)\Big)\\=~&\begin{cases}
	(\mathscr{A}-\delta)V(x)+\gamma \Big(x-a_p+V(a_p)-V(x)\Big),\quad&x>b\\
	(\mathscr{A}-\delta)V(x)+\gamma \Big(x-a_p+V(a_p)-V(x)\Big),\quad&a_p\leq x<b\\
	(\mathscr{A}-\delta)V(x),\quad&x<a_p
	\end{cases}\\
	\leq ~&0
	\end{align*}
	in view of (\ref{eq.HJB1.bGeq0}) and the PDEs satisfied by the value function.
	
	For the second equation, clearly the left hand side cannot be positive when $x\leq {\chi}/{\beta}$, since $V$ is increasing. For $x>{\chi}/{\beta}$, we denote 
	\begin{equation*}
	H_2(\xi):=\beta\xi-\chi+V(x-\xi)-V(x),\quad \xi\geq 0.
	\end{equation*}
	Clearly, we have $H_2(0)=-\chi<0$. Taking derivative w.r.t. $\xi$, we get
	$$H_2^\prime(\xi)=\beta-V'(x-\xi).$$
	In view of Lemma \ref{Lemma.VD.opt}, $V'(x-\xi)<\beta$ is equivalent to $x-\xi\in(a_c,b)$, or equivalently $x-b<\xi<x-a_c$. Therefore, we can deduce that $H_2$ is decreasing on $[0,x-b]$ if $x-b\geq 0$ then increasing on $[\max(x-b,0),x-a_c]$ if $x-a_c\geq 0$, then decreasing on $(a_c,\infty)$. Hence, on $[0,x]$, its maximum is attained at either $0$, or $x-a_c$ if $x-a_c>0$. Suppose $\xi=x-a_c>0$, we have
	$$H_2(x-a_c)=\beta(x-a_c)-\chi+V(a_c)-V(x),$$
	which as a function of $x$, is increasing on $[a_c+{\chi}/{\beta},b]$. This implies that 
	$$H_2(x-a_c)\leq H_2(b-a_c)=\beta(b-a_c)-\chi+V(a_c)-V(b)=0>H_2(0).$$
	Therefore, $\sup_{\xi\in[0,x]}H_2(\xi)\leq 0$. This completes the proof as $H_2$ and the term inside the bracket in the second component differs only at $\xi=0$, where both of them have value less than or equal to $0$.

\section{Proof of Equation \eqref{Property.Lambda}}\label{A.Proof.Lambda}
First, we have $\Lambda(V'(0;\pi_0))=0$ as $\beta=V'(0;\pi_0)$ implies that $a_\beta=0$. In addition, $\Lambda$ is an increasing function because
\begin{align*}
\parD{\beta}\Lambda(\beta)=~&\parD{\beta}a_\beta-\frac{V'(a_\beta;\pi_0)\parD{\beta}a_\beta}{\beta}+\frac{V(a_\beta;\pi_0)}{\beta^2}\\
=~&\Big(\parD{\beta}a_\beta\Big)\Big(1-\frac{V'(a_\beta;\pi_0)}{\beta}\Big)+\frac{V(a_\beta;\pi_0)}{\beta^2}\\
>~&0.
\end{align*}
Moreover, when $\beta\uparrow\frac{\gamma}{\gamma+\delta}$, we have $a_\beta\uparrow\infty$ and therefore
\begin{align*}
\lim_{\beta\uparrow\frac{\gamma}{\gamma+\delta}}\Lambda(\beta)=~&\lim_{x\uparrow\infty}(x-\frac{V(x;\pi_0)}{(\frac{\gamma}{\gamma+\delta})})\\
=~&\frac{\gamma+\delta}{\gamma}\lim_{x\uparrow \infty}(\frac{\gamma}{\gamma+\delta}x-\frac{\gamma\mu}{(\gamma+\delta)^2}(1-e^{s_1x})-\frac{\gamma}{\gamma+\delta}x)\\
=~&\frac{-\mu}{\gamma+\delta},
\end{align*}
where we have used Lemma \ref{Lemma.ValueFunction} to compute $V(x;\pi_0)$.

\section{Proof of Lemma \ref{Lemma.b1b2.Exist} }\label{A.Lemma8.6}

	Using the fact that the value function is continuous at $b_2$, we get
	\begin{equation} \label{eqt.cont.at.b}			Be^{s_1b_2}=\beta b_2-\chi-\frac{\gamma}{\gamma+\delta}b_2-\frac{\gamma}{\gamma+\delta}\frac{\mu}{\gamma+\delta},
	\end{equation}
	which implies that 
	$$V'(b_2+;\pi_{b_1,b_2})=\frac{\gamma}{\gamma+\delta}+Bs_1e^{s_1b_2}=(\beta b_2-\chi)s_1-\frac{\gamma}{\gamma+\delta}b_2s_1-\frac{\gamma}{\gamma+\delta}\frac{\mu}{\gamma+\delta}s_1+\frac{\gamma}{\gamma+\delta}$$
	Therefore $V'(b_2+;\pi_{b_1,b_2})=\beta$ is equivalent to
	\begin{equation}\label{eqt.for.b2}
	s_1(\beta-\frac{\gamma}{\gamma+\delta}) b_2-\Big(\chi s_1+s_1\frac{\gamma\mu}{(\gamma+\delta)^2} -(\frac{\gamma}{\gamma+\delta}-\beta)\Big)=0.
	\end{equation}
	Similarly, by replacing the equality of (\ref{eqt.cont.at.b}) by ``$\leq$'' and $B$ by $\widetilde{B}=\frac{-\gamma\mu}{(\gamma+\delta)^2}$, we have
	$$\beta=V'(a_\beta;\pi_0)\geq \frac{\gamma}{\gamma+\delta}+\widetilde{B}s_1e^{s_1a_\beta}=(\beta a_\beta-\chi)s_1-\frac{\gamma}{\gamma+\delta}a_\beta s_1-\frac{\gamma}{\gamma+\delta}\frac{\mu}{\gamma+\delta}s_1+\frac{\gamma}{\gamma+\delta},$$
	which shows that when $b_2=a_\beta$, the left hand side of (\ref{eqt.for.b2}) is negative. As the left hand side of (\ref{eqt.for.b2}) is linear in $b_2$ with positive coefficient, we establish the existence (and uniqueness) of $b_2\in(a_\beta,\infty)$.
	
	For $b_1$, we note (from definitions of $c_{\beta,\chi}$, $\pi_{c_{\beta,\chi},b_2}$ and Remark \ref{Remark.mu.neg.obs}) that 
	$$A(c_{\beta,\chi})g(c_{\beta,\chi})+\beta c_{\beta,\chi}-\chi=V(c_{\beta,\chi};\pi_{c_{\beta,\chi},b_2})=\beta c_{\beta,\chi}-\chi,$$
	which shows that $A(c_{\beta,\chi})=0$. Therefore, we have 
	$$V'(c_{\beta,\chi};\pi_{c_{\beta,\chi},b_2})=A(c_{\beta,\chi})g'(c_{\beta,\chi})+V'(c_{\beta,\chi};\pi_0)=V'(c_{\beta,\chi};\pi_0)<V'(a_\beta;\pi_0)=\beta.$$
	Similarly, we have
	$$V'(a_\beta-;\pi_{a_\beta,b_2})=A(a_\beta)g'(a_\beta)+V'(a_\beta;\pi_0)=A(a_\beta)g'(a_\beta)+\beta.$$
	As $g'(a_\beta)>0$, if $A(a_\beta)\leq 0$, then Lemma \ref{Lemma.A.increasing} implies that there is an interval $I\subsetneq (c_{\beta,\chi},a_\beta)$ where $A(b)$ is decreasing, since $A(b)$ is increasing on the neighbourhood of $c_\beta$. Again from Lemma \ref{Lemma.A.increasing}, we can deduce that $V'(b-;\pi_{b,b_2})>\beta$ for $b\in I$, which is what we have to show. If $A(a_\beta)>0$, then from the above equation, we have $V'(a_\beta-;\pi_{a_\beta,b_2})>\beta$ and hence by continuity, there is a $c_{\beta,\chi}<b<a_\beta$ such that $V'(b-;\pi_{b,b_2})=\beta$. For uniqueness, we will choose $b_1$ to be the smallest one if there are more than one such $b$.

\section{Proof of Lemma \ref{Lemma.b.Exist} }\label{A.Lemma8.7}

    We first need to establish that 
%
	for $\kappa>0$, it holds that 
	$$\frac{-\mu}{\gamma+\delta}(1-e^{s_1\kappa})-\kappa<0.$$
	Denote $h$ the left hand side of the above inequality, then we have obviously $h(0)=0$. By taking derivative with respect to $\kappa$ and note $r_1s_1({\mu}/({\gamma+\delta}))=r_1s_1$, we can have 
	\begin{align*}
	h'(\kappa)=\frac{-\mu}{\gamma+\delta}(-s_1)e^{s_1\kappa}-1
	=(e^{s_1\kappa}-1)+\frac{s_1}{r_1}e^{s_1k}
	<0.
	\end{align*}

	From Lemma \ref{Lemma.ValueFunction}, we have
	$$
	V'(b-;\pi_{b,\infty})=A(b)g'(b)-\frac{\gamma\mu}{(\gamma+\delta)^2}s_1e^{s_1b}+\frac{\gamma}{\gamma+\delta}
	$$
	with
	$$A(b)=\frac{(\beta-\frac{\gamma}{\gamma+\delta})b-\chi-\frac{\gamma\mu}{(\gamma+\delta)^2}(1-e^{s_1b})}{g(b)}.$$
	This implies that there is a $\widetilde{b}\geq 0$ such that $b>\widetilde{b}$ implies $A(b)>0$, and
	$$\lim_{b\rightarrow\infty}A(b)g'(b)=+\infty$$
	and hence we are done if
	$$V'(\frac{\chi}{\beta}-;\pi_{\frac{\chi}{\beta},\infty})\leq \beta.$$
	In particular, this would be the case if $A({\chi}/{\beta})<0$, since $g'({\chi}/{\beta})>0$ and
	from Lemma \ref{Lemma.ValueFunction} we have
	$$V'(\frac{\chi}{\beta}-;\pi_{\frac{\chi}{\beta},\infty})=A(\frac{\chi}{\beta})g'(\frac{\chi}{\beta})+V'(\frac{\chi}{\beta};\pi_0)\leq A(\frac{\chi}{\beta})g'(\frac{\chi}{\beta})+\frac{\gamma}{\gamma+\delta}<A(\frac{\chi}{\beta})g'(\frac{\chi}{\beta})+\beta.$$
	This is indeed the case, thanks to very first inequality developed in this section,
	$$A(\frac{\chi}{\beta})=\frac{\frac{\gamma}{\gamma+\delta}(\frac{-\mu}{\gamma+\delta}(1-e^{s_1\frac{\chi}{\beta}})-\frac{\chi}{\beta})}{g(\frac{\chi}{\beta})}<0.$$

\section{Proof of Lemma \ref{L_conv}}\label{A.Lemma9.1}

Suppose $\beta\leq {\gamma}/{(\gamma+\delta)}$, then from Theorem \ref{Thm.Small.Beta}, we know that a periodic barrier strategy $\pi_b$ is optimal. Based on the observation in e.g. \citet{AvTuWo14}, we can conclude that $\pi_0$ is optimal for $\mu=0$ (we omit the proof here although a separate check is possible). Now, in view of Theorem \ref{Thm.mu.neg}, we are in Case 1: $\chi\geq\beta({-\mu}/({\gamma+\delta}))$ for small enough $|\mu|$ (note $\mu<0$). Hence, when $\beta\leq {\gamma}/{(\gamma+\delta)}$, we conclude that a periodic barrier strategy $\pi_0$ is optimal. Hence, the continuity is established.
	
	For $\beta>{\gamma}/{(\gamma+\delta)}$, again in view of Theorem \ref{Thm.mu.neg} Case 1, we can conclude that a liquidation $(b_1,\infty)$ is optimal, with the barrier $b_1$ such that the derivative of the value function at the barrier is $\beta$. Note that such strategy is also the same as a hybrid $(a_p,a_c,b)$ strategy with $a_p=a_c=0$ and $b=b_1$. Since the \emph{optimal} hybrid $(a_c,a_p,b)$ strategy imposes that the derivative of the value function at the upper barrier $b$ is $\beta$ (which is unique by Corollary \ref{Corrolary.unique.barriers}), the continuity of the barriers would be established if we can show that for $\mu=0$, the hybrid $(0,0,b)$ strategy is optimal.
	
	The value function of a hybrid $(0,0,b)$ strategy, denoted by $V$ (instead of $V(\cdot;\pi_{0,0,b})$ for convenience), is given by
	$$V(x)=\begin{cases}
	Ag(x)+\frac{\gamma}{\gamma+\delta}x,&x\leq b\\\beta x-\chi,&x>b
	\end{cases},$$
	which can be derived with ease. When $\mu=0$, the function $g(x)$ can be rewritten as 
	$$g(x)=e^{r_1x}-e^{-r_1x},$$
	which has the property $g''(x)=r_1^2g(x)>0$ for $x>0$. This implies that
	\begin{equation}\label{eq.mu0.1}
	g'(b)>g'(0).
	\end{equation}
	On the other hand, by direct computation, we have
	$$V'(b)=Ag'(b)+\frac{\gamma}{\gamma+\delta}=\beta$$
	which shows that $A>0$. Therefore, we have
	\begin{align*}
	V'(0)=~&Ag'(0)+\frac{\gamma}{\gamma+\delta}\\
	\leq ~&Ag'(b)+\frac{\gamma}{\gamma+\delta}=\beta.
	\end{align*}
	This shows that the hybrid $(0,0,b)$ strategy is optimal for $\mu=0$ (see Definition \ref{Def.Nice.hybrid} and Theorem \ref{Thm}) and completes the proof.

\section{Proof of Lemma \ref{Lemma.convergence.beta}}\label{A.proof.Convergence.beta}

For a hybrid $(a_p,a_c,b)$ strategy, if $V'(a_c)=V'(b)=\beta$, we have
$$
V'(a)-\frac{-s_1\frac{\gamma}{\gamma+\delta}}{-s_1+q(r_1+s_1)}=(r_1-s_1)I(y+l,q)=(r_1-s_1)I(l,q),$$
with
\begin{align*}
I(x,q)=~&
\frac{\beta-\frac{\gamma}{\gamma+\delta}+\Big(\frac{\gamma}{\gamma+\delta}-\frac{-s_1\frac{\gamma}{\gamma+\delta}}{-s_1+q(r_1+s_1)}\Big)
	e^{s_1 x}}{g'(x)+g(x)(-r_1s_1)\frac{\mu}{\gamma+\delta}(1-q)},\\
q:=~&Q(a)=1-\frac{f(a)/f'(a)}{\mu/\delta}\in[0,1].
\end{align*}

When $\beta\downarrow \frac{\gamma}{\gamma+\delta}$, from the fact that $q>0$ for our candidate strategy, it is easy to see that $y\uparrow \infty$. This implies that $I(y+l,q)\downarrow 0$. Suppose $\frac{-s_1}{r_1}\frac{\gamma}{\gamma+\delta}>1$,
we get $V'(a_p)=\frac{-s_1\frac{\gamma}{\gamma+\delta}}{-s_1+q(r_1+s_1)}= 1$, and 
Therefore, in order for $I(l,q)=I(y+l,q)$, we must have 
$$l\uparrow\infty,\quad a_p=Q^{-1}\Big(\frac{s_1\frac{\delta}{\gamma+\delta}}{r_1+s_1}\Big),\quad y\uparrow\infty.$$
Similarly, if $\frac{-s_1}{r_1}\frac{\gamma}{\gamma+\delta}\leq 1$, we have
$$l\uparrow \infty,\quad a_p=0,\quad y\uparrow\infty.$$

If $a>0$, after some tedious algebra, we can show that
\begin{align*}
&(r_1-s_1)A+\frac{\gamma}{\gamma+\delta}
=\frac{\gamma}{\gamma+\delta},
\end{align*}
implying that $A=0$.

Otherwise, if $a=0$, we have $q=1$ and therefore $\frac{-s_1\frac{\gamma}{\gamma+\delta}}{-s_1+q(r_1+s_1)}=\frac{\gamma}{\gamma+\delta}\frac{-s_1}{r_1}$. Similar to the above calculation, we also have
\begin{align*}
&(r_1-s_1)A+\frac{\gamma}{\gamma+\delta}
=\frac{\gamma}{\gamma+\delta},
\end{align*}
implying that $A=0$. Therefore, when $\beta\downarrow{\gamma}/({\gamma+\delta})$, the value function $V(x;\pi_{a_p,a_c,b\uparrow\infty})$ converges to $V(x;\pi_{a_p})$, the value function of the periodic barrier strategy with barrier level $a_p$. We are left to verify that $a_p$ is the optimal barrier in the pure periodic setting, e.g. \citet{NoPeYaYa17}.

Recall from Remark \ref{Remark.suff.cons} that the condition for $a=0$ are the same when $\beta\downarrow{\gamma}/({\gamma+\delta})$. For $a>0$, we have $Cf'(a)=\frac{-s_1\frac{\gamma}{\gamma+\delta}}{r_1+q(r_1+s_1)}=1$, which is the same as 
$$\frac{f(a)}{f'(a)}=(1-\frac{s_1\frac{\delta}{\gamma+\delta}}{r_1+s_1})\frac{\mu}{\delta}.$$
Use these, we get
\begin{align*}
V'''(a+)=Bs_1^3=~&\frac{\delta}{\gamma+\delta}Cf(a)s_1^3-\frac{\gamma}{\gamma+\delta}\frac{\mu}{\gamma+\delta}s_1^3\\
=~&\frac{\delta}{\gamma+\delta}\frac{f(a)}{f'(a)}Cf'(a)s_1^3-\frac{\gamma}{\gamma+\delta}\frac{\mu}{\gamma+\delta}s_1^3\\
=~&\frac{\delta}{\gamma+\delta}s_1^2
\end{align*}
and
\begin{align*}
f'''(a)=~&\frac{\gamma+\delta}{\sigma^2/2}\Big(\frac{\delta}{\gamma+\delta}f'(a)-\frac{\mu}{\gamma+\delta}f''(a)\Big)\\
=~&(-r_1s_1)\Big((\frac{\delta}{\gamma+\delta}-(r_1+s_1)\frac{\mu}{\gamma+\delta})f'(a)+(r_1+s_1)\frac{\delta}{\gamma+\delta}f(a)\Big).
\end{align*}
Therefore, we have
\begin{align*}
\frac{f'''(a)}{f'(a)(-r_1s_1)}=~&\frac{\delta}{\gamma+\delta}-(r_1+s_1)\frac{\mu}{\gamma+\delta}+(r_1+s_1)(1-\frac{s_1}{r_1+s_1}\frac{\delta}{\gamma+\delta})\frac{\mu}{\gamma+\delta}\\
=~&\frac{\delta}{\gamma+\delta}\frac{-s_1}{r_1},
\end{align*}
which further implies 
$$
V'''(a-)=Cf'''(a)=Cf'(a)\frac{f'''(a)}{f'(a)}=\frac{\delta}{\gamma+\delta}s_1^2=V'''(a+),
$$
which is the ``smoothness condition'', equation (4.1) in \citet{NoPeYaYa17}, which characterises the unique periodic barrier.


\section{Computational considerations} \label{S_comput}

In the computation of the barriers for hybrid $(a_p,a_c,b)$ strategies, instead of using the maximisation of $V(a_c)-\beta a_c$, we solved the derivative conditions directly. The barriers were efficiently calculated without difficulties. 
\begin{enumerate}
	\item For each $l\geq 0$, denote the \emph{unique} $y$ such that $V'(b)=\beta$ as $y_{\bar{a}}(l)$ and $y_0(l)$ for $a=\bar{a}$ and $a=0$ respectively. 
	\item For each $l\geq 0$ and $y\in[y_{\bar{a}}(l),y_0(l)]$, we can find a corresponding $q=Q(a)$ such that $V'(b)=\beta$. As $Q$ is monotone in $a$, we can recover $a$.
	\item For each $l\geq 0$, $(y,a)$ pair in the previous step indexed by $y$, we evaluate the derivative of the value function at $a$. For each $l\geq 0$, we can find a $(y,a)$ pair such that both $V'(a)=1$ (or $a=0$ if $V'(0)\leq 1$) and $V'(a+y+l)=\beta$.
	\item For $l=0$, using the corresponding $(y,a)$ pair to compute $V'(a)$. If $V'(a)\leq \beta$, then set $l=0$. Otherwise, write a function to output the corresponding $V'(a+l)$ using the $(y,a)$ pair from the previous step. By increasing $l$ following by a bisection method, we can find a corresponding $l$ such that $V'(a+l)=\beta$.
\end{enumerate}

All equations can be solved by for example bisection method combining with some searching technique. In case of numerical overflow, we can rescale the scale parameters $(\mu,\sigma,\chi)$ to find the barriers then scale back. It should be clear that rescaling should not change the optimality of an optimal strategy.

In order to make sure that the possibilities of multiple solutions for some equations would not result in some disruptive impact to the numerical procedure, we can verify that the final output indeed satisfies all conditions proposed. 

Finally, when $\beta$ is close to the asymptote $\gamma/(\gamma+\delta)$, we have numerical overflow as $b^*\uparrow\infty$. In this case, an approximation is used based on Proposition \ref{Lemma.convergence.beta}, where we treat $b^*=\infty$ and calculate $a_p^*$ and $a_c^*-a_p^*$ independently. Sometimes, $b^*$ may not be large enough to validate such approximation and a ``bias'' is resulted. In this case, we will adjust the bias term such that the approximation piece glues to the piece without approximation. We decreases the bias term (linearly for convenience) so that it eventually vanishes at the asymptote $\beta=\gamma/(\gamma+\delta)$. The region where we employ such approximation is indicated between the dotted line (where the numerical overflow starts) and the solid line (the asymptote).

For $\mu<0$, it is straightforward from Theorem \ref{Thm.mu.neg}. Specifically, we proceed the following:
\begin{enumerate}
	\item First check whether $\chi\geq \beta({-\mu}/({\gamma+\delta}))$.
	\item Suppose $\chi\geq \beta({-\mu}/({\gamma+\delta}))$. If $\beta\leq {\gamma}/({\gamma+\delta})$, $\pi_0$ is optimal. No numerical method is needed. On the other hand, if $\beta> {\gamma}/({\gamma+\delta})$, we express $V'(b-;\pi_{b,\infty})$ as a function of $b$ and search for $V'(b-;\pi_{b,\infty})=\beta$.
	\item Suppose $\chi< \beta({-\mu}/({\gamma+\delta}))$. Invert $\Lambda$ at ${\chi}/{\beta}$ to output $\beta_0=\Lambda^{-1}({\chi}/{\beta})$. If $\beta\leq \beta_0$, then $\pi_0$ is optimal and no numerical method is needed. If $\beta\in(\beta_0,{\gamma}/({\gamma+\delta}))$, first output $c_{\beta,\chi}$ by solving $V'(a_\beta)=\beta$. It is then followed by solving 2 equations: $V'(b_1-;\pi_{b_1,b_2})=\beta$ in $b_1$ and $V'(b_2+;\pi_{b_1,b_2})=\beta$ in $b_2$, respectively. Note that for both equations, the other parameter is not used. Finally, if $\beta\in[{\gamma}/({\gamma+\delta}),1]$, as in the previous case, we express $V'(b-;\pi_{b,\infty})$ as a function of $b$ and search for $V'(b-;\pi_{b,\infty})=\beta$.
\end{enumerate}

Again, all equations can be solved by for example bisection method combining with some search techniques. 

To compute $b_0^*$, we do not use the results directly from \citet{NoPeYaYa17}. Instead, we use the formula given by the third item in Proposition \ref{Lemma.convergence.beta}, which holds for any $\beta$.

\section{Continuity for different cases in Theorem \ref{Thm.mu.neg}} \label{A_Thm.mu.neg.cont}

We consider here the four different cases enumerated in Theorem \ref{Thm.mu.neg} sequentially.

\subsection{$\frac{\chi}{\beta}<\frac{-\mu}{\gamma+\delta}$, $\beta\uparrow\frac{\gamma}{\gamma+\delta}$.}
When $\beta\uparrow\gamma/(\gamma+\delta)$, we have (from $V'(\cdot;\pi_0)\uparrow\gamma/(\gamma+\delta)$) that $a_\beta\rightarrow \infty$ which implies that $b_2\rightarrow \infty$ and therefore from Theorem \ref{Thm.mu.neg} case 2b we can conclude that $b_1$ is continuous at $\beta=\gamma/(\gamma+\delta)$, i.e. $$\pi_{b_1,b_2\rightarrow\infty}\rightarrow\pi_{b_1,\infty}$$
when $\beta\uparrow\gamma/(\gamma+\delta)$.

\subsection{$\frac{\chi}{\beta}\geq \frac{-\mu}{\gamma+\delta}$, $\beta\downarrow\frac{\gamma}{\gamma+\delta}$.}

Recall the optimal strategy in the setting $\pi_{b,\infty}$ is characterised by the derivative condition at $b$, i.e. $V'(b;\pi_{b,\infty})=\beta$. In view of the first $2$ equations in the proof of Lemma \ref{Lemma.b.Exist}, we have
\begin{equation}\label{eq.VDb}
V'(b-;\pi_{b,\infty})=\frac{(\beta-\frac{\gamma}{\gamma+\delta})b-\chi-\frac{\gamma\mu}{(\gamma+\delta)^2}(1-e^{s_1b})}{g(b)}g'(b)-\frac{\gamma\mu}{(\gamma+\delta)^2}s_1e^{s_1b}+\frac{\gamma}{\gamma+\delta}.
\end{equation}
Therefore, we have $$\parD{\beta}V'(b-;\pi_{b,\infty})=b\frac{g'(b)}{g(b)}>0.$$
This implies that when $\beta$ decreases from $\beta_1$ to $\beta_2$, the original $b=b(\beta_1)$ yields $V'(b-;\pi_{b,\infty})<\beta_2$ and therefore from the proof of Lemma \ref{Lemma.b.Exist} we need to use a larger $b$. This implies that $b(\beta_2)>b(\beta_1)$. In other words, when $\beta\downarrow\gamma/(\gamma+\delta)$, the corresponding $b=b(\beta)$ is increasing. It remains to show that $b(\beta)$ is not converging so that we have 
$$\pi_{b\uparrow\infty,\infty}\rightarrow \pi_0.$$
Suppose $b(\beta)\uparrow b<\infty$ as $\beta\downarrow\gamma/(\gamma+\delta)$. Then by taking the limit, we have 
$$V'(b-;\pi_{b,\infty})=\beta=\frac{\gamma}{\gamma+\delta}.$$
Thus, from \eqref{eq.VDb} we have 
$$-\chi=\frac{\gamma\mu}{(\gamma+\delta)^2}\Big(s_1e^{s_1b}+(1-e^{s_1b})\frac{g'(b)}{g(b)}\Big)=\frac{\gamma\mu}{(\gamma+\delta)^2}s_1,$$
which is impossible since $-\chi$ is negative but the very last term is positive.

\subsection{$\frac{\chi}{\beta}\uparrow\frac{-\mu}{\gamma+\delta}$, $\beta\in[\beta_0,\frac{\gamma}{\gamma+\delta}]$.}

Here, we want to show that if $\beta<\gamma/(\gamma+\delta)$ then when $\chi/\beta$ is ``close'' to $-\mu/(\gamma+\delta)$, we have $\beta_0>\beta$. This means the two conditions $\chi/\beta\uparrow-\mu/(\gamma+\delta)$ and $\beta\in[\beta_0,\gamma/(\gamma+\delta)]$ cannot be satisfied simultaneously unless $\beta=\gamma/(\gamma+\delta)$. In other words, the cells in the second row of the Table \ref{T_roadmap}  are continuous only at $\beta=\gamma/(\gamma+\delta)$, which has already been taken care of.

Recall that $\beta_0$ is defined by the inverse of the increasing function $\Lambda$ at ${\chi}/{\beta}$, i.e. $\beta_0=\Lambda^{-1}(\chi/\beta)$. In addition, $\Lambda$ maps $\beta\in[V'(0;\pi_0),\gamma/(\gamma+\delta))$ to $[0,-\mu/(\gamma+\delta))$. Therefore, when $\chi/\beta\uparrow\gamma/(\gamma+\delta)$, we have $\beta_0\uparrow\gamma/(\gamma+\delta)$, which implies that $\beta\rightarrow\gamma/(\gamma+\delta)$. 

\subsection{$\frac{\chi}{\beta}<\frac{-\mu}{\gamma+\delta}$, $\beta\downarrow\beta_0$.}

Recall $\beta_0=\Lambda^{-1}(\chi/\beta)$ and therefore $\beta\downarrow\beta_0$ implies that \eqref{Condition.Liqb1b2.Optimal} is an equality at the limit and we have $a_\beta-c_{\beta,\chi}\rightarrow 0$. Consequently, from $c_{\beta,\chi}\leq b_1\leq a_\beta$ we can conclude that $b_1\rightarrow a_\beta$. 

It remains to show $b_2\rightarrow a_\beta$. In view of Lemma 
\ref{Lemma.b1b2.Exist}, we can establish $b_2\rightarrow a_\beta$ if we can show that at $\beta=\beta_0$, we have \eqref{eqt.for.b2} with $b_2=a_\beta$, i.e. $$s_1(\beta-\frac{\gamma}{\gamma+\delta}) a_\beta-\Big(\chi s_1+s_1\frac{\gamma\mu}{(\gamma+\delta)^2} -(\frac{\gamma}{\gamma+\delta}-\beta)\Big)=0.$$

Using $$V(x;\pi_0)=\frac{-\gamma\mu}{(\gamma+\delta)^2}e^{s_1x}+\frac{\gamma}{\gamma+\delta}\Big({x+\frac{\mu}{\gamma+\delta}}\Big),$$
we can re-express $V'(a_\beta;\pi_0)=\beta$ as
\begin{equation}\label{dev.eqvuiv}
s_1\frac{-\gamma\mu}{(\gamma+\delta)^2}e^{s_1a_\beta}=\beta-\frac{\gamma}{\gamma+\delta}
\end{equation}
and \eqref{Condition.Liqb1b2.Optimal} (with inequality replaced by equality) as 
\begin{align*}
&s_1\Big(\frac{-\gamma\mu}{(\gamma+\delta)^2}e^{s_1a_\beta}+\frac{\gamma}{\gamma+\delta}\Big({a_\beta+\frac{\mu}{\gamma+\delta}}\Big)\Big)=s_1\beta a_\beta-s_1\chi\\\iff~& s_1(\beta-\frac{\gamma}{\gamma+\delta})a_\beta-s_1\chi-s_1\frac{\gamma\mu}{(\gamma+\delta)^2}=s_1\frac{-\gamma\mu}{(\gamma+\delta)^2}e^{s_1a_\beta}=\beta-\frac{\gamma}{\gamma+\delta},
\end{align*}
where the last equality is from \eqref{dev.eqvuiv} and the last line is essentially what we are trying to show, i.e. \eqref{eqt.for.b2}.

Since we have $b_1,b_2\rightarrow a_\beta$ when $\beta\downarrow\beta_0$, we have 
$$\pi_{b_1\rightarrow a_\beta,b_2\rightarrow a_\beta}\rightarrow \pi_0.$$

\end{document}